\documentclass{lmcs} 
\pdfoutput=1

\usepackage{lastpage}
\lmcsdoi{20}{3}{9}
\lmcsheading{}{\pageref{LastPage}}{}{}%
{Dec.~22,~2022}{Jul.~23,~2024}{}

\usepackage{hyperref}


\usepackage[T1]{fontenc}

\usepackage[utf8]{inputenc}

\usepackage{microtype}
\usepackage{pgfplots}

\usepackage{amssymb,amsmath}

\newcommand{\autorefexas}[1]{\renewcommand{\exaautorefname}{Examples}%
  \autoref{#1}\renewcommand{\exaautorefname}{Example}%
}%

\newcommand{\autoreflems}[1]{\renewcommand{\lemautorefname}{Lemmas}%
  \autoref{#1}\renewcommand{\lemautorefname}{Lemma}%
}%

\usepackage{xspace,relsize,stmaryrd}
\usepackage{amssymb,amsmath}
\usepackage{mathpartir}
\usepackage{bussproofs}
\usepackage{multirow} \usepackage{enumitem}
\usepackage[author=anonymous,marginclue,footnote,draft]{fixme}
\FXRegisterAuthor{dh}{adh}{DH} \FXRegisterAuthor{ls}{als}{LS}

\allowdisplaybreaks[2]

\usepackage{array}

\usepackage{tikz}
 \usetikzlibrary{trees}
 \usetikzlibrary{shapes}
 \usetikzlibrary{fit}
 \usetikzlibrary{shadows}
 \usetikzlibrary{backgrounds}
 \usetikzlibrary{arrows,automata}

\tikzset{
   n/.style= {circle,fill,inner sep=1.5pt,node distance=2cm}
  ,acc/.style={circle,draw,inner sep=3pt,node distance=2cm}
  ,phantom/.style={circle},
  ,arr/.style={->, >=stealth, semithick, shorten <= 3pt, shorten >= 3pt}
}

\usepackage{tikz-cd}

\newcommand{\Inf}{\mathsf{Inf}}
\newcommand{\Sat}{\mathit{Sat}}
\newcommand{\PropAts}{\mathsf{P}}
\newcommand{\AtGames}{\mathcal{A}}
\newcommand{\game}{\mathsf{G}}
\newcommand{\sfgame}{\game^{\sub}}
\newcommand{\sfsymgame}{\game^{\sub,\mathsf{sym}}}
\newcommand{\initstate}{q_{\mathsf{init}}}

\newcommand{\sub}{\mathsf{sub}}
\newcommand{\osOne}{\gamma}
\newcommand{\osZero}{\Theta}

\makeatletter
\def\moverlay{\mathpalette\mov@rlay}
\def\mov@rlay#1#2{\leavevmode\vtop{%
   \baselineskip\z@skip \lineskiplimit-\maxdimen
   \ialign{\hfil$\m@th#1##$\hfil\cr#2\crcr}}}
\newcommand{\charfusion}[3][\mathord]{
    #1{\ifx#1\mathop\vphantom{#2}\fi
        \mathpalette\mov@rlay{#2\cr#3}
      }
    \ifx#1\mathop\expandafter\displaylimits\fi}
\makeatother

\newcommand\pfun{\rightharpoonup} \newcommand{\FV}{\mathsf{FV}}
\newcommand{\size}{\mathsf{size}} \newcommand{\Land}{\bigwedge}
\newcommand{\Lor}{\bigvee} \newcommand{\prios}{2n_0k'}
\newcommand{\detsize}{((n_0k')!)^2}
 
\newcommand{\ad}{\mathsf{ad}} 
\newcommand{\Rules}{\mathcal R}

\newcommand{\detcarrier}{{D_\target}} \newcommand{\detprio}{\Omega}

\newcommand{\target}{\chi} \newcommand{\FLtarget}{\mathbf F}

\newcommand{\Mon}{\mathcal{M}}
\newcommand{\Neighb}{\mathcal{N}}
\newcommand{\Dist}{\mathcal{D}}

\newcommand{\Rat}{{\mathbb{Q}}}
\newcommand{\Nat}{{\mathbb{N}}}
\newcommand{\Int}{{\mathbb{Z}}}

\newcommand{\Bag}{\mathcal{B}}

\newcommand{\mondiamond}[1]{\langle #1\rangle}
\newcommand{\monbox}[1]{[#1]}
\newcommand{\gldiamond}[1]{\Diamond_{#1}}

\newcommand\ExpTime{\upshape{\textsc{ExpTime}}\xspace}

\newcommand{\Set}{\mathsf{Set}}

\newcommand{\Sem}[1]{{[\![#1]\!]}}
\newcommand{\hearts}{\heartsuit}

\newcommand{\sem}[1]{[\![#1]\!]}

\newcommand{\Pow}{\mathcal{P}}
\newcommand{\contrapow}{\mathcal{Q}}
\newcommand{\Op}{\mathit{op}}
\newcommand{\lO}{\mathcal{O}}

\newcommand{\BV}{\mathop{\mathsf{BV}}}
\newcommand{\outord}{\prec}
\newcommand{\depord}{\outord_{\mathsf{dep}}}
\newcommand{\Var}{\mathbf{V}}
\newcommand{\autom}{\mathsf{A}}

\newenvironment{alg}[1][Algorithm]{\begin{trivlist}
\item[\hskip \labelsep {\bfseries Algorithm (#1).}]}{\end{trivlist}}

\begin{document}

\title[Coalgebraic Satisfiability Checking for Arithmetic
  $\mu$-Calculi]{Coalgebraic Satisfiability Checking\texorpdfstring{\\}{} for Arithmetic
  $\mu$-Calculi}

\thanks{Work by the first author supported by the ERC Consolidator grant D-SynMA (No.~772459) and by the EPSRC through grant EP/Z003121/1. Work by the second author is funded by the Deutsche Forschungsgemeinschaft (DFG, German Research Foundation) – project no.~531706730.}

\author[D.~Hausmann]{Daniel Hausmann\lmcsorcid{0000-0002-0935-8602}}[a]
\author[L.~Schr\"oder]{Lutz Schr\"oder\lmcsorcid{0000-0002-3146-5906}}[b]

\address{University of Liverpool, United Kingdom}
\email{d.hausmann@liverpool.co.uk}

\address{Friedrich-Alexander-Universit\"{a}t
  Erlangen-N\"urnberg, Germany}
\email{lutz.schroeder@fau.de}

\begin{abstract} The coalgebraic $\mu$-calculus provides a generic semantic framework for fixpoint logics over systems whose branching type goes beyond the standard relational setup, e.g.\ probabilistic, weighted, or game-based. Previous work on the coalgebraic $\mu$-calculus includes an exponential-time upper bound on satisfiability checking, which however relies on the availability of tableau rules for the next-step modalities that are sufficiently well-behaved in a formally defined sense; in particular, rule matches need to be representable by polynomial-sized codes, and the sequent duals of the rules need to absorb cut. While such rule sets have been identified for some important cases, they are not known to exist in all cases of interest, in particular ones involving either integer weights as in the graded $\mu$-calculus, or real-valued weights in combination with non-linear arithmetic. In the present work, we prove the same upper complexity bound under more general assumptions, specifically regarding the complexity of the (much simpler) satisfiability problem for the underlying \emph{one-step   logic}, roughly described as the nesting-free next-step fragment of the logic. The bound is realized by a generic algorithm that supports on-the-fly satisfiability checking. Notably, our approach directly accommodates unguarded formulae, and thus avoids use of the guardedness transformation. Example applications include new exponential-time upper bounds for satisfiability checking in an extension of the graded $\mu$-calculus with polynomial inequalities (including positive Presburger arithmetic), as well as an extension of the (two-valued) probabilistic $\mu$-calculus with polynomial inequalities. 
\end{abstract}
\maketitle

\section{Introduction}\label{sec:intro}

Modal fixpoint logics are a well-established tool in the temporal
specification, verification, and analysis of concurrent systems. One
of the most expressive logics of this type is the modal
$\mu$-calculus~\cite{Kozen83,BradfieldStirling06,BradfieldWalukiewicz18},
which features explicit least and greatest fixpoint operators; roughly
speaking, these serve to specify liveness properties (least fixpoints)
and safety properties (greatest fixpoints), respectively. Like most
modal logics, the modal $\mu$-calculus is traditionally interpreted
over relational models such as Kripke frames or labelled transition
systems. The growing interest in more expressive models where
transitions are governed, e.g., by probabilities, weights, or games
has sparked a commensurate growth of temporal logics and fixpoint
logics interpreted over such systems; prominent examples include
probabilistic
$\mu$-calculi~\cite{CleavelandEA05,HuthKwiatkowska97,CirsteaEA11a,LiuEA15}, the
alternating-time $\mu$-calculus~\cite{AlurEA02}, and the monotone
$\mu$-calculus, which contains Parikh's game
logic~\cite{Parikh85}. The graded $\mu$-calculus~\cite{KupfermanEA02}
features next-step modalities that count successors; it is standardly
interpreted over Kripke frames but, as pointed out by D'Agostino and
Visser~\cite{DAgostinoVisser02}, graded modalities are more naturally
interpreted over so-called multigraphs, where edges carry integer
weights, and in fact this modification leads to better bounds on
minimum model size for satisfiable formulae.

Coalgebraic logic~\cite{Pattinson03} has emerged as a unifying
framework for modal logics interpreted over such more general
models. It is based on casting the transition type of the systems at
hand as a set functor, and the systems in question as coalgebras for
this type functor, following the paradigm of universal
coalgebra~\cite{Rutten00}; additionally, modalities are interpreted as
so-called \emph{predicate liftings}~\cite{Pattinson03,Schroder08}. The
\emph{coalgebraic $\mu$-calculus}~\cite{CirsteaEA11a} caters for
fixpoint logics within this framework, and essentially covers all
mentioned (two-valued) examples as instances. It has been shown that
satisfiability checking in a given instance of the coalgebraic
$\mu$-calculus is in \ExpTime, \emph{provided} that one exhibits a set
of tableau rules for the modalities, so-called \emph{one-step rules},
that is \emph{one-step tableau complete} (a condition that, in the
dual setting of sequent calculi, translates into the requirement that
the rule absorb cut~\cite{SchroderPattinson09b,PattinsonSchroder10}) and moreover
\emph{tractable} in a suitable sense (an assumption made also in our
own previous work on the flat~\cite{HausmannSchroeder15} and
alternation-free~\cite{HausmannEA16} fragments of the coalgebraic
$\mu$-calculus). Such rules are known for many important cases,
notably including alternating-time logics, the probabilistic
$\mu$-calculus even when extended with linear inequalities, and the monotone $\mu$-calculus~\cite{SchroderPattinson09b,KupkePattinson10,CirsteaEA11a}. There
are, however, important cases where such rule sets are currently
missing, and where there is in fact little perspective for finding
suitable rules. Prominent cases of this kind are the graded
$\mu$-calculus and more expressive logics over integer weights
featuring, e.g., Presburger arithmetic\footnote{One-step tableau-complete sets of rules for the graded $\mu$-calculus and more generally for the Presburger $\mu$-calculus have been claimed in earlier work~\cite{SchroderPattinson09b,KupkePattinson10} but have since turned out to be in fact incomplete~\cite{KupkeEA22,GorlitzEA23}; see also \autoref{rem:graded-rules}.}. Further cases arise when
logics over systems with non-negative real weights, such as
probabilistic systems, are taken beyond linear arithmetic to include
polynomial inequalities.

The object of the current paper is to fill this gap by proving a
generic \ExpTime upper bound for coalgebraic $\mu$-calculi even in the
absence of tractable sets of modal tableau rules. The method we use
instead is to analyse the so-called \emph{one-step satisfiability}
problem of the logic on a semantic level -- this problem is
essentially the satisfiability problem of a very small fragment of the
logic, the \emph{one-step logic}, which excludes not only fixpoints,
but also nested next-step modalities, with a correspondingly
simplified semantics that no longer involves actual transitions. E.g.\
the one-step logic of the standard relational $\mu$-calculus is
interpreted over models essentially consisting of a set with a
distinguished subset, which abstracts the successors of a single state
(which is not itself part of the model). We have applied this
principle to satisfiability checking in coalgebraic (next-step) modal
logics~\cite{SchroderPattinson08}, coalgebraic hybrid
logics~\cite{MyersEA09}, and reasoning with global assumptions in
coalgebraic modal logics~\cite{KupkeEA15,KupkeEA22}. It also appears
implicitly in work on automata for the coalgebraic
$\mu$-calculus~\cite{FontaineEA10}, which however establishes only a
doubly exponential upper bound in the case without tractable modal
tableau rules.

Our main result states roughly that if the satisfiability problem of
the one-step logic is in \ExpTime under a particular stringent measure
of input size, then satisfiability in the corresponding instance of
the coalgebraic $\mu$-calculus is in \ExpTime. Since the criteria on
rule sets featuring in previous work on the coalgebraic $\mu$-calculus
imply the \ExpTime bound on satisfiability checking in the one-step
logic, this result subsumes previous complexity estimates for the
coalgebraic $\mu$-calculus~\cite{CirsteaEA11}. Our leading example
applications are on the one hand the graded $\mu$-calculus and its
extension with (monotone) polynomial inequalities (including
Presburger modalities, i.e.\ monotone linear inequalities), and on the
other hand the extension of the (two-valued) probabilistic
$\mu$-calculus~\cite{CirsteaEA11a,LiuEA15} with (monotone) polynomial
inequalities. While the graded $\mu$-calculus as such is known to be
in \ExpTime~\cite{KupfermanEA02}, the other mentioned instances of our
result are, to our best knowledge, new. At the same time, our proofs
are fairly simple, even compared to specific ones, e.g.\ for the
graded $\mu$-calculus.

Technically, we base our results on an automata- and game-theoretic
treatment by means of standard parity automata and parity
(satisfiability) games. Our satisfiability games are quite different
from satisfiability games formulated previously for the very similar
\emph{$\Lambda$-automata}~\cite{FontaineEA10}. Both types of game are
of overall doubly exponential size; however, we show that our game can
nevertheless be solved in singly exponential time even in the absence
of a tractable set of modal tableau rules. Our algorithm witnessing
the singly exponential time bound is able to decide the satisfiability
of nodes on-the-fly, that is, possibly before the tableau is fully
expanded, in the spirit of global caching algorithms for description
logics~\cite{GoreWidmann09,GoreNguyen13}. It thus offers a perspective
for practically feasible reasoning. Our construction implies a known
singly-exponential bound on minimum model size for satisfiable
formulae in coalgebraic
$\mu$-calculi~\cite{CirsteaEA11a,FontaineEA10}, calculated only in
terms of the closure size of the formula and its alternation depth.
Moreover, we identify a criterion for a polynomial bound on branching
in models, which holds in all our examples.

This paper extends a previous conference
publication~\cite{HausmannSchroder19}. Besides providing full proofs
and additional background material, we also make do without the
problematic guardedness transformation
(cf.~\cite{BruseEA15,KupkeEA20}) and directly establish our results
for the unguarded coalgebraic $\mu$-calculus. On a technical level, we
rebase large parts of the development on newly introduced notions of
model checking game and satisfiability game for the unguarded
coalgebraic $\mu$-calculus, obtaining streamlined proofs. A
restriction of the algorithm to guarded formulae has been implemented
within the \emph{Coalgebraic Ontology Logic Solver (COOL
  2)}~\cite{GorlitzEA23}, which in particular provides the first
implemented reasoner for the graded $\mu$-calculus.

\paragraph{Overview of the material}
In \autoref{sec:prelims}, we recall the basics of coalgebra,
coalgebraic logic, and the coalgebraic $\mu$-calculus, including
central syntactic notions such as closure and alternation depth.  We
discuss the automata-theoretic approach and \emph{model checking
  games} in \autoref{sec:tracking}. Specifically, we introduce the
\emph{tracking automaton}, a nondeterministic parity automaton that
essentially detects bad sequences of fixpoint unfoldings, and its
co-determinization. The nondeterministic tracking automaton implicitly
governs the model checking game, while the co-determinized tracking
automaton forms the basis of our notion of \emph{tableaux}, introduced
in \autoref{sec:tableaux}. Tableaux are certain partial
subautomata of the co-determinized tracking automaton, characterized
on the one hand by being \emph{totally accepting} in the sense that
every infinite run is accepting, and on the other hand by one-step
satisfiability requirements on modal constraints. Tableaux allow for
the construction of a model of the target formula: We show that on
every tableau, one has a so-called coherent coalgebra structure, in
what in coalgebraic logic is usually termed the existence lemma
(\autoref{lem:existence}). The associated truth lemma
(\autoref{lem:truth}) uses the model checking game to show that the
target formula is actually satisfied in such a coherent coalgebra,
precisely exploiting the fact that the co-determinized tracking
automaton complements the tracking automaton.

The \emph{satisfiability game}, introduced in
\autoref{sec:satgames}, then essentially serves to determine
whether there exists a tableau for the target formula. Its
correctness, in the sense of actually capturing satisfiability of the
target formula, is split into two implications: On the one hand, we
show that a tableau can indeed be extracted from a winning strategy of
the existential player (\autoref{lem:games}), and on the other hand
we prove a soundness lemma (\autoref{lem:modeltogame}) stating that
a winning strategy for the existential player in the satisfiability
game can be extracted from a winning strategy in the model checking
game on a given model of the target formula.

It remains to obtain an exponential-time satisfiability checking
algorithm from the satisfiability game, which as indicated above has
doubly exponential size. We show in \autoref{section:alg} that
winning regions in the satisfiability game can be characterized as a
nested fixpoint that lives on the singly-exponential sized subset of
those game positions that correspond to states of the co-determinized
tracking automaton (\autoref{lem:sat-game-fp}). Our
satisfiability checking algorithm, also presented in
\autoref{section:alg}, essentially computes this fixpoint, or more
precisely decides containment of game positions, in particular the
root position, in the fixpoint on-the-fly. As indicated above,
efficiency of the algorithm relies on sufficiently low complexity of
the one-step satisfiability problem. We show that all our example
logics satisfy this criterion, and are hence decidable in exponential
time

\paragraph{Related Work}

Our work builds on existing approaches to the algorithmic treatment of
the relational $\mu$-calculus, most centrally on the game-theoretic
approach pioneered by \cite{NiwinskiWalukiewicz96}. Our correctness
proof for the model checking game takes orientation from
Venema~\cite{Venema06} in that it makes do without (ordinal) timeouts,
i.e.\ without (transfinite) Kleene iteration, for least fixpoints. As
indicated above, our approach is distinguished by working directly
with unguarded formulae. Friedman and Lange~\cite{FriedmannLange13a}
have previously provided a tableau method dealing with unguarded
formulae in the relational $\mu$-calculus; our approach differs
technically from theirs, with details discussed in
\autoref{rem:unguarded}.

We have already mentioned previous work on the algorithmics of the
coalgebraic $\mu$-calculus~\cite{CirsteaEA11a,FontaineEA10} (other
recent work on the coalgebraic $\mu$-calculus concerns model theory,
notably completeness~\cite{SchroderVenema18,EnqvistEA19} and
expressive completeness~\cite{EnqvistEA17}). As noted above, the main
difference with work on tableau-based algorithms~\cite{CirsteaEA11a}
is that we make do without a tractable complete set of tableau rules
and cover unguarded formulae, a point that was previously explicitly
left open~\cite[Section~1]{CirsteaEA11a}. While the overall technical
layout of our approach is similar, e.g.\ in the use of tracking
automata as well as model checking and satisfiability games, the added
generality in our work implies substantial differences in the actual
implementation of these tools. These concern, for instance, the
alphabet used by the tracking automata and the treatment of
propositional operators. Fontaine et al.~\cite{FontaineEA10} present
an approach, interestingly not assuming guardedness, where a
coalgebraic $\mu$-calculus formula is first transformed into a
so-called $\Lambda$-automaton, whose emptiness can then be checked by
means of a satisfiability game, which in fact is a regular game but
not a parity game. This game has doubly exponential size and in
general yields only a doubly exponential complexity bound. We also use
a satisfiability game in our approach, which works directly on the
formula syntax; altogether, the design of our game appears to differ
rather markedly from that used by Fontaine et al. First off, as
mentioned above, our game has a parity winning condition. It does also
have doubly exponential size but can nevertheless be solved in singly
exponential time using a fixpoint calculation on a
singly-exponential-sized subset of the game positions. It seems unlikely
that similar ideas will work for the game used by Fontaine et al., as
out of the two types of positions present in this game, one is of doubly
exponential size and the other of only polynomial size; more details
are found in \autoref{rem:flv}.

\section{The Coalgebraic \texorpdfstring{$\mu$}{mu}-Calculus}\label{sec:prelims}
We recall basic definitions in coalgebra~\cite{Rutten00}, coalgebraic
logic~\cite{Pattinson03,Schroder08}, and the coalgebraic
$\mu$-calculus~\cite{CirsteaEA11a}. While we do repeat definitions of
some categorical terms, we assume some familiarity with
basic category theory (e.g.~\cite{AdamekEA90}).

\paragraph{Coalgebra} The semantics of our logics will be based on
transition systems in a general sense, which we abstract as coalgebras
for a type functor. The most basic example are relational transition
systems, or Kripke frames, which are just pairs $(C,R)$ consisting of
a set~$C$ of \emph{states} and a binary \emph{transition relation}
$R\subseteq C\times C$. We may equivalently view such a structure as a
map of type $\xi\colon C\to\Pow C$ (where~$\Pow$ denotes the powerset
functor and we omit brackets around functor arguments following
standard convention, while we continue to use brackets in uses of
powerset that are not related to functoriality), via a bijective
correspondence that maps~$R\subseteq C\times C$ to the map
$\xi_R\colon C\to\Pow C$ given by $\xi_R(c)=\{d\mid (c,d)\in R\}$ for
$c\in C$; that is, $\xi_R$ assigns to each state~$c$ a collection of
some sort, in this case a set, of successor states. The essence of
\emph{universal coalgebra}~\cite{Rutten00} is to base a general theory
of state-based systems on encapsulating this notion of
\emph{collections of successors} in a functor, for our purposes on the
category $\Set$ of sets and maps. We recall that such a functor
$F\colon\Set\to\Set$ assigns to each set~$U$ a set~$FU$, and to each
map $f\colon U\to V$ a map $Ff\colon FU\to FV$, preserving identities
and composition.  As per the intuition given above, $FU$ should be
understood as consisting of some form of collections of elements
of~$U$. An \emph{$F$-coalgebra} $(C,\xi)$ then consists of a set~$C$
of \emph{states} and a \emph{transition map} $\xi\colon C\to FC$,
understood as assigning to each state~$c$ a collection $\xi(c)\in FC$
of successors. As indicated above, a basic example is $F=\Pow$ (with
$\Pow f(Y)=f[Y]$ for $f\colon U\to Y$), in which case $F$-coalgebras
are relational transition systems. To see just one further example
now, the \emph{discrete distribution functor}~$\Dist$ is given on
sets~$U$ by $\Dist U$ being the set of all discrete probability
distributions on~$U$. We represent such distributions by their
\emph{probability mass functions}, i.e.~functions
$d\colon U\to(\Rat\cap[0,1])$ such that $\sum_{u\in U}d(u)=1$,
understood as assigning to each element of~$U$ its probability. (Note
that the \emph{support} $\{u\in U\mid d(u)>0\}$ of~$d$ is then
necessarily at most countable.)  By abuse of notation, we also
write~$d$ for the induced discrete probability measure on~$U$,
i.e.~$d(Y)=\sum_{y\in Y}d(y)$ for $Y\in\Pow(U)$. The action of~$\Dist$
on maps $f\colon U\to V$ is then given by $\Dist f(d)(Z)=d(f^{-1}[Z])$
for $d\in\Dist U$ and $Z\in\Pow(V)$, i.e.~$\Dist f$ takes image
measures. A $\Dist$-coalgebra $\xi\colon C\to\Dist C$ assigns to each
state $x\in C$ a distribution $\xi(x)$ over successor states, so
$\Dist$-coalgebras are Markov chains.

\paragraph{Coalgebraic logic} The fundamental abstraction that
underlies coalgebraic logic is to encapsulate the semantics of
next-step modalities in terms of \emph{predicate liftings}. As a basic
example, consider the standard diamond modality~$\Diamond$, whose
semantics over a transition system $(C,R)$ is given by
\begin{equation*}
  c\models \Diamond\phi\quad\text{iff}\quad \exists d.\,(c,d)\in R\land d\models\phi
\end{equation*}
(where~$\phi$ is a formula in some ambient syntax whose semantics is
assumed to be already given by induction). We can rewrite this
definition along the correspondence between transition systems and
$\Pow$-coalgebras $\xi\colon C\to\Pow C$ described above, obtaining
\begin{align*}
  c\models\Diamond\phi&\quad\text{iff}\quad\xi(c)\cap\Sem{\phi}\neq\emptyset\\
  &\quad\text{iff}\quad\xi(c)\in\{Y\in\Pow C\mid Y\cap\Sem{\phi}\neq\emptyset\}
\end{align*}
where $\Sem{\phi}=\{d\in C\mid d\models\phi\}$ denotes the
\emph{extension} of~$\phi$ in~$(C,\xi)$. We can see the set
$\{Y\in\Pow C\mid Y\cap\Sem{\phi}\neq\emptyset\}$ as arising from the
application of a set operation $\Sem{\Diamond}$ to the
set~$\Sem{\phi}$, thought of as a predicate on~$C$: For a
predicate~$P$ on~$C$, we put
\begin{align*}
  \Sem{\Diamond}(P)=\{Y\in\Pow C\mid Y\cap P\neq\emptyset\}
\end{align*}
and then have
\begin{equation*}
  c\models\Diamond\phi\quad\text{iff}\quad\xi(c)\in\Sem{\Diamond}(\Sem{\phi}).
\end{equation*}
Notice that the operation $\Sem{\Diamond}$ turns predicates on~$C$
into predicates on $\Pow C$, so we speak of a \emph{predicate
  lifting}.

Generally, the notion of predicate lifting is formally defined as
follows. The \emph{contravariant powerset} functor
$\contrapow\colon\Set^\Op\to\Set$ is given by $\contrapow U=\Pow U$
for sets~$U$, and by
$\contrapow f\colon\contrapow V\to\contrapow U$,
$\contrapow f(Z)=f^{-1}[Z]$, for $f\colon U\to V$. We think of
$\contrapow U$ as consisting of predicates on~$U$. Then, an
\emph{$n$-ary predicate lifting} for~$F$ is a natural
transformation
\begin{equation*}
  \lambda\colon \contrapow ^n\to \contrapow\circ F^\Op 
\end{equation*}
between functors $\Set^\Op\to\Set$. Here, we write $\contrapow^n$ for
the functor given by $\contrapow^n U=(\contrapow U)^n$. Also, recall
that $F^\Op$ is the functor $\Set^\Op\to\Set^\Op$ that acts like~$F$
on both sets and maps. Unravelling this definition, we see that a
predicate lifting~$\lambda$ is a family of maps
\begin{equation*}
  \lambda_U\colon(\contrapow U)^n\to\contrapow(FU),
\end{equation*}
indexed over all sets~$U$, subject to the \emph{naturality} condition
requiring that the diagram
\begin{equation*}
  \begin{tikzcd}
    (\contrapow V)^n \arrow[r,"\lambda_V"]\arrow[d,"\contrapow(f)^n"] &
    \contrapow(F V)\arrow[d,"\contrapow(Ff)"]\\
    (\contrapow U)^n\arrow[r,"\lambda_U"] & \contrapow(F U)
  \end{tikzcd}
\end{equation*}
commutes for all $f\colon U\to V$. Explicitly, this amounts to the
equality
\begin{equation*}
  \lambda_U(f^{-1}[Z_1],\dots,f^{-1}[Z_n]) = (Ff)^{-1}\lambda_V(Z_1,\dots,Z_n)
\end{equation*}
for predicates $Z_1,\dots,Z_n\in\contrapow(V)$. It is easily checked
that this condition does hold for the unary predicate lifting
$\lambda=\Sem{\Diamond}$ as defined above. As examples of predicate
liftings for the above-mentioned discrete distribution
functor~$\Dist$, consider the unary predicate liftings~$\lambda^b$,
indexed over $b\in\Rat\cap[0,1]$, given by
\begin{equation*}
  \lambda^b_U(Y)=\{d\in\Dist U\mid d(Y)> b\},
\end{equation*}
which induce next-step modalities `with probability more than~$b$'
(see also \autoref{ex:logics}.\ref{item:prob-mu}). Again, the
naturality condition is readily checked. Further instances are
discussed in \autoref{ex:logics}.

\paragraph{Fixpoints} We recall basic results on fixpoints, and fix
some notation. Recall that by the Knaster-Tarski fixpoint theorem,
every monotone function $f\colon X\to X$ on a complete lattice~$X$
(such as a powerset lattice $X=\Pow(Y)$) has a least fixpoint~$\mu f$
and a greatest fixpoint $\nu f$, and indeed~$\mu f$ is the least
prefixpoint and, dually, $\nu f$ is the greatest postfixpoint
of~$f$. Here, $x\in X$ is a prefixpoint (postfixpoint) of~$f$ if
$f(x)\le x$ ($x\le f(x)$). We use~$\mu$ and~$\nu$ also as binders in
expressions $\mu X.\,E$ or $\nu X.\,E$ where~$E$ is an expression, in
an informal sense, possibly depending on the parameter~$X$, thus
denoting $\mu f$ and $\nu f$, respectively, for the function~$f$ that
maps~$X$ to~$E$. By the Kleene fixpoint theorem, if~$X$ is finite,
then we can compute $\mu f$ as the point $f^n(\bot)$ at which the
ascending chain $\bot\le f(\bot)\le f^2(\bot)\dots$ becomes
stationary; dually, we compute $\nu f$ as the point $f^n(\top)$ at
which the descending chain $\top\ge f(\top)\ge f^2(\top)\dots$ becomes
stationary. More generally, this method can be applied to
unrestricted~$X$ by extending the mentioned chains to ordinal indices,
taking suprema or infima, respectively, in the limit
steps~\cite{CousotCousot79}.
  
\paragraph{The coalgebraic $\mu$-calculus} Next-step modalities
interpreted as predicate liftings can be embedded into ambient logical
frameworks of various levels of expressiveness, such as plain modal
next-step logics~\cite{SchroderPattinson08,SchroderPattinson09b} or
hybrid next-step logics~\cite{MyersEA09}. The \emph{coalgebraic
  $\mu$-calculus}~\cite{CirsteaEA11a} combines next-step modalities
with a Boolean propositional base and fixpoint operators, and thus
serves as a generic framework for temporal logics.

The \textbf{syntax} of the coalgebraic $\mu$-calculus is parametric in
a \emph{modal similarity type}~$\Lambda$, that is, a set of modal
operators with assigned finite arities, possibly including
propositional atoms as nullary modalities. \emph{We fix a modal
  similarity type~$\Lambda$ for the rest of the paper.}  We assume
that $\Lambda$ is closed under duals, i.e., that for each modal
operator $\hearts\in\Lambda$, there is a \emph{dual}
$\overline{\hearts}\in\Lambda$ (of the same arity) such that
$\overline{\overline{\hearts}}=\hearts$ for all $\hearts\in\Lambda$.
Let $\Var$ be a countably infinite set of \emph{fixpoint
  variables}~$X,Y,\dots$.  Formulae $\phi,\psi,\dots$ of the
\emph{coalgebraic $\mu$-calculus} (over $\Lambda$) are given by the
grammar
\begin{align*}
\psi,\phi ::= \bot \mid \top \mid \psi\wedge\phi \mid 
\psi\vee\phi \mid \hearts(\phi_1,\dots,\phi_n) \mid X\mid\neg X \mid \mu X.\,\phi \mid \nu X.\,\phi\qquad
\end{align*}
where $\hearts\in \Lambda$, $X\in\Var$, and~$n$ is the arity
of~$\hearts$. We require that negated variables $\neg X$ do not occur
in the scope of $\mu X$ or $\nu X$. As usual,~$\mu$ and~$\nu$ take
least and greatest fixpoints, respectively. We write $\phi[\psi/X]$
for the formula obtained by substituting~$\psi$ for~$X$
in~$\phi$. Full negation is not included but can be defined as usual
(in particular, $\neg\hearts\phi=\overline\hearts\neg\phi$, and
moreover $\neg\mu X.\,\phi=\nu X.\,\neg\phi[\neg X/X]$ and dually; one
easily checks that the restriction on occurrences of negated variables
is satisfied in the resulting formula).  Throughout, we use
$\eta\in\{\mu,\nu\}$ as a placeholder for fixpoint operators; we
briefly refer to formulae of the form $\eta X.\,\phi$ as
\emph{fixpoints} or \emph{fixpoint literals}. We follow the usual
convention that the scope of a fixpoint extends as far to the right as
possible. Fixpoint operators \emph{bind} their fixpoint variables, so
that we have standard notions of bound and free fixpoint variables; we
write $\FV(\phi)$ and $\BV(\phi)$ for the sets of free and bound
variables, respectively, that occur in a formula~$\phi$. A
formula~$\phi$ is closed if it contains no free fixpoint variables,
i.e.~$\FV(\phi)=\emptyset$. (Note in particular that closed formulae
do not contain negated fixpoint variables, and hence no negation at
all.)  We assume w.l.o.g.~ that fixpoints are \emph{irredundant},
i.e.\ use their fixpoint variable at least once.  In \emph{guarded}
formulae, all occurrences of fixpoint variables are separated by at
least one modal operator from their binding fixpoint
operator. Guardedness is a wide-spread assumption although the actual
blowup incurred by the transformation of unguarded into guarded
formulae depends rather subtly on the notion of formula
size~\cite{BruseEA15,KupkeEA20}. We do \emph{not} assume guardedness
in this work, see also \autoref{rem:unguarded}. For
$\hearts\in\Lambda$, we denote by $\size(\hearts)$ the length of a
suitable representation of~$\hearts$; for natural or rational numbers
indexing~$\hearts$ (cf.~\autoref{ex:logics}), we assume binary
representation.  The \emph{length} $|\psi|$ of a formula is its length
over the alphabet
$\{\bot,\top,\wedge,\vee\}\cup \Lambda\cup \Var \cup \{\eta X.\mid
\eta\in\{\mu,\nu\},X\in\Var \}$, while the \emph{size} $\size(\psi)$
of~$\psi$ is defined by counting~$\size(\hearts)$ for each
$\hearts\in\Lambda$ (and~$1$ for all other operators).

The \textbf{semantics} of the coalgebraic $\mu$-calculus, on the other
hand, is parametrized by the choice of a functor~$F\colon\Set\to\Set$
determining the branching type of systems, as well as predicate
liftings interpreting the modalities, as indicated in the above
summary of coalgebraic logic. That is, we interpret each modal
operator $\hearts\in\Lambda$ as an $n$-ary predicate lifting
\begin{equation*}
  \Sem{\hearts}\colon\contrapow^n\to\contrapow\circ F^\Op
\end{equation*}
for~$F$, where~$n$ is the arity of~$\hearts$, extending notation
already used in our lead-in example on the semantics of the diamond
modality~$\Diamond$. To ensure existence of fixpoints, we require that
all~$\sem{\hearts}$ are \emph{monotone}, i.e.\ whenever
$A_i\subseteq B_i\subseteq U$ for all $i=1,\dots,n$, where~$n$ is the
arity of~$\hearts$, then
$\sem{\hearts}_U(A_1,\dots,A_n)\subseteq
\sem{\hearts}_U(B_1.\dots,B_n)$.

For sets $U\subseteq V$, we write $\overline{U}=V\setminus U$ for the
\emph{complement} of~$U$ in~$V$ when~$V$ is understood from the
context. We require that the assignment of predicate liftings to
modalities respects duality, i.e.\
\begin{equation*}
  \sem{\overline\hearts}_U(A_1,\dots,A_n)=\overline{\sem{\hearts}_U(\overline{A_1},\dots,\overline{A_n})}
\end{equation*}
for $n$-ary $\hearts\in\Lambda$ and $A_1,\dots,A_n\subseteq U$.  

We interpret formulae over $F$-coalgebras $\xi\colon C\to FC$. A
\emph{valuation} is a partial function $i\colon\Var\pfun \Pow(C)$ that
assigns sets $i(X)$ of states to fixpoint variables $X$ (we generally
write $\pfun$ to indicate partial functions). Given $A\subseteq C$ and
$X\in\Var$, we write $i[X\mapsto A]$ for the valuation given by
$(i[X\mapsto A])(X)=A$ and $(i[X\mapsto A])(Y)=i(Y)$ for $Y\neq X$. We
write~$\epsilon$ for the empty valuation (i.e.~$\epsilon$ is undefined
on all variables). For a list $X_1,\dots,X_n$ of distinct variables,
we write $i[X_1\mapsto A_1,\dots,X_n\mapsto A_n]$ for
$i[X_1\mapsto A_1]\dots[X_n\mapsto A_n]$ (i.e.\
$i[X_1\mapsto A_1,\dots,X_n\mapsto A_n]$ maps~$X_j$ to~$A_j$ for
$j=1,\dots,n$ and otherwise acts like~$i$), and
$[X_1\mapsto A_1,\dots,X_n\mapsto A_n]$ for
$\epsilon[X_1\mapsto A_1,\dots,X_n\mapsto A_n]$. The \emph{extension}
$\sem{\phi}_i\subseteq C$ of a formula~$\phi$ in $(C,\xi)$, under a
valuation~$i$ such that $i(X)$ is defined for all $X\in\FV(\phi)$, is
given by the recursive clauses
\begin{align*}
  \sem{\bot}_i & = \emptyset\\
  \sem{\top}_i & = C\\
  \sem{X}_i &= i(X)\\
  \sem{\neg X}_i &= C\setminus i(X)\\
  \sem{\phi\land\psi}_i& = \sem{\phi}_i\cap\sem{\psi}_i\\
  \sem{\phi\lor\psi}_i& = \sem{\phi}_i\cup\sem{\phi}_i\\
  \sem{\hearts(\phi_1,\dots,\phi_n)}_i &= \xi^{-1}[\sem{\hearts}_C(\sem{\phi_1}_i.\dots,\sem{\phi_n}_i)]\\
  \sem{\mu X.\,\phi}_i &= \mu A.\,\sem{\phi}_{i[X\mapsto A]}\\
  \sem{\nu X.\,\phi}_i &= \nu A.\,\sem{\phi}_{i[X\mapsto A]}
\end{align*}
(using notation introduced in the fixpoint paragraph above).  Thus we
have $x\in\sem{\hearts(\psi_1,\dots,$ $\psi_n)}_i$ if and only if
$\xi(x)\in\sem{\hearts}_C(\sem{\psi_1}_i,\dots,\sem{\psi_n}_i)$.  By
monotonicity of predicate liftings and the restrictions on occurrences
of negated variables, one shows inductively that the functions
occurring in the clauses for $\mu X.\,\phi$ and $\nu X.\,\phi$ are
monotone, so the corresponding extremal fixpoints indeed exist
according to the Knaster-Tarski fixpoint theorem.  By an evident
substitution lemma, we obtain that the extension is invariant under
\emph{unfolding} of fixpoints, i.e.\
\begin{equation*}
  \sem{\eta X.\,\psi}_i=\sem{\psi[\eta X.\,\psi/X]}_i.
\end{equation*}
For closed formulae~$\psi$, the valuation~$i$ is irrelevant, so we
write $\sem{\psi}$ instead of $\sem{\psi}_i$. A state $x\in C$
\emph{satisfies} a closed formula $\psi$ (denoted $x\models \psi$) if
$x\in\sem{\psi}$.  A closed formula $\target$ is \emph{satisfiable} if
there is a coalgebra $(C,\xi)$ and a state $x\in C$ such that
$x\models\target$.

For readability, we restrict the further technical development to
unary modalities, noting that all proofs generalize to higher arities
by just writing more indices; in fact, we will liberally use higher
arities in examples. \emph{For the remainder of the paper, we fix a
  functor~$F$ and predicate liftings $\Sem{\hearts}$ interpreting the
  modalities $\hearts\in\Lambda$.}

\begin{exa}[Coalgebraic $\mu$-calculi]\label{ex:logics}
  We proceed to discuss instances of the coalgebraic $\mu$-calculus,
  focusing mainly on cases where no tractable set of modal tableau
  rules is known (details on this point are in
  \autoref{rem:rules}). Examples where such rule sets are
  available, including the alternating-time $\mu$-calculus, have been
  provided by C\^i{}rstea et al.~\cite{CirsteaEA11a}. We do discuss
  two examples where tractable sets of modal tableau rules are known,
  viz, the relational $\mu$-calculus and the monotone
  $\mu$-calculus. We include the former to show how our present
  coalgebraic notions match the most familiar example; and the latter
  to illustrate the importance of covering unguarded formulae by
  showing that these arise in the embedding of game
  logic~\cite{PaulyParikh03}.
  \begin{enumerate}[wide]
  \item\label{item:K} The \emph{relational modal
      $\mu$-calculus}~\cite{Kozen83} (which contains CTL as a
    fragment), has
    $\Lambda=\{\Diamond, \Box\}\cup \PropAts\cup\{\neg a\mid a\in \PropAts\}$
    where~$\PropAts$ is a set of propositional atoms, seen as nullary
    modalities, and $\Diamond,\Box$ are unary modalities. The
    modalities $\Diamond$ and $\Box$ are mutually dual
    ($\overline\Diamond=\Box$, $\overline\Box=\Diamond$), and the dual
    of~$a\in \PropAts$ is $\neg a$. The semantics is defined over
    \emph{Kripke models}, which are coalgebras for the functor~$F$
    given on sets~$U$ by $FU=(\Pow U)\times\Pow(\PropAts)$ -- an $F$-coalgebra
    assigns to each state a set of successors and a set of atoms
    satisfied in the state.  The relevant predicate liftings are
    \begin{align*}
      \sem{\Diamond}_U(A)&=\{(B,Q)\in FU\mid A\cap B\neq \emptyset\}\\
      \sem{\Box}_U(A)&=\{(B,Q)\in FU\mid B\subseteq A\}\\
      \sem{a}_U & =\{(B,Q)\in FU\mid a\in Q\} & (a\in \PropAts)\\
      \sem{\neg a}_U &=\{(B,Q)\in FU\mid a\notin Q\}
    \end{align*}
    (of course, the predicate liftings for~$a\in \PropAts$ and $\neg a$ are
    nullary, i.e.~take no argument predicates).  Standard example
    formulae include the CTL-formula
    $\mathsf{AF}\,a=\mu X.\,(a\vee \Box X)$, which states that on all
    paths,~$a$ eventually holds, and the fairness formula
    $\nu X.\,\mu Y.\, ((a\wedge \Diamond X)\vee \Diamond Y)$, which
    asserts the existence of a path on which~$a$ holds infinitely
    often.
  \item\label{item:monotone-mu} The \emph{monotone
      $\mu$-calculus}~\cite{PaulyThesis} has a set~$\PropAts$ of
    propositional atoms, seen as nullary modalities in the same way as
    in our treatment of the relational $\mu$-calculus, and modalities
    $\mondiamond{g}$ with duals $\monbox{g}$, indexed over
    \emph{atomic games}~$g$ from a fixed set~$\AtGames$. It is
    interpreted over \emph{monotone neighbourhood models}, which are
    coalgebras for~$\Mon^\AtGames\times \Pow(\PropAts)$. Here,
    $(-)^\AtGames$ denotes exponentiation with~$\AtGames$ (generally,
    for sets~$U$ and~$V$, $V^U$ is the set of maps $U\to V$),
    and~$\Mon$ is the \emph{monotone neighbourhood functor}, defined
    as follows. We first define the \emph{neighbourhood functor}
    $\Neighb$ as the composite $\contrapow\circ\contrapow^\Op$ of the
    contravariant powerset functor~$\contrapow$ with itself;
    explicitly, $\Neighb U=\contrapow(\contrapow U)$ is the double
    powerset of~$U$, and for $f\colon U\to V$ and $N\in\Neighb U$,
    $\Neighb f(N)=\{Z\subseteq V\mid f^{-1}[Z]\in N\}$. Call
    $N\in\Neighb U$ \emph{upclosed} if $W\in N$ whenever $V\in N$ and
    $V\subseteq W$. Then,~$\Mon$ is the subfunctor of~$\Neighb$ given
    by $\Mon U =\{N\in \Neighb U\mid N\text{ upclosed}\}$. Thus, a
    monotone neighbourhood model $(C,\xi)$ consists of a set~$C$ of
    states and a transition map~$\xi$ assigning to each state~$x\in C$
    a set of propositional atoms that hold in~$x$, and for each
    $g\in\AtGames$ an upwards closed set of
    \emph{$g$-neighbourhoods}. The semantics of propositional atoms is
    given like in the relational $\mu$-calculus. The interpretation of
    the modalities as predicate liftings is given by
    \begin{align*}
      \Sem{\mondiamond{g}}_U(A)&=\{(N,Q)\in(\Mon U)^\AtGames\times
                                 \Pow(\PropAts)\mid A\in N(g)\}\\
      \Sem{\monbox{g}}_U(A)&=\{(N,Q)\in(\Mon U)^\AtGames\times
                             \Pow(\PropAts)\mid \forall B\in N(g).\,B\cap A\neq\emptyset\}.
    \end{align*}
    As indicated by the nomenclature, one way to understand the
    monotone $\mu$-calculus is as a logic of two-player games, where
    $\mondiamond{g}\phi$ says that the first player (`Angel') can
    enforce~$\phi$ after playing game~$g$. \emph{Game
      logic}~\cite{Parikh83} features modalities
    $\mondiamond{\gamma}$, $\monbox{\gamma}$ for (non-atomic)
    games~$\gamma$; here, \emph{games} $\gamma,\delta$ are defined by
    the grammar
    \begin{equation*}
      \gamma,\delta::=g\mid \gamma\cup\delta \mid \gamma;\delta\mid\gamma^d\mid\gamma^* \qquad(g\in\AtGames)
    \end{equation*}
    (where for simplification we omit the test construct $\phi?$, for
    game logic formulae~$\phi$). Here, $\gamma\cup\delta$ is a game in
    which Angel decides whether to play~$\gamma$ or~$\delta$;
    $\gamma;\delta$ is the game where~$\gamma$ and~$\delta$ are played
    in sequence;~$\gamma^*$ is a game in which~$\gamma$ is played
    repeatedly, and Angel decides before each round (including the
    first) whether~$\gamma$ is played again; and~$\gamma^d$ is
    like~$\gamma$ but with the roles of the players reversed. Thus,
    $\gamma\cap\delta:=(\gamma^d\cup\delta^d)^d$ is a game where the
    second player (`Demon') decides whether~$\gamma$ or~$\delta$ is
    played, and $\gamma^\times:=((\gamma^d)^*)^d$ is a game
    where~$\gamma$ is played repeatedly, with Demon deciding before
    each round whether~$\gamma$ is played another time. This logic can
    be embedded into the monotone $\mu$-calculus, with polynomial
    blowup when formulae are measured in terms of subformula or
    closure size~\cite{PaulyThesis}, as we will do here. We slightly
    simplify the original translation, which aims at using only two
    fixpoint variables; the translation~$t$ is then defined by mutual
    recursion with operators $\tau_\gamma$ that translate the effect
    of applying $\mondiamond{\gamma}$ to an argument formula. The
    translation~$t$ is given by commutation with all propositional
    operators and by
    \begin{equation*}
      t(\mondiamond{\gamma}\phi)=\tau_\gamma(t(\phi)).
    \end{equation*}
    We refrain from listing the clauses for~$\tau_\phi$ in full; the
    most salient clauses are
    \begin{equation*}
      \tau_g(\phi)=\mondiamond{g}\phi\qquad \tau_{\gamma^d}(\phi)=\neg\tau_\gamma(\neg\phi)\qquad
      \tau_{\gamma^*}(\phi)=\mu X.\,(\phi\lor\tau_\gamma(X))
    \end{equation*}
    where~$X$ is a fresh fixpoint variable. For instance, we have
    \begin{equation*}
      t(\mondiamond{(g^*)^\times}a)=\nu X.\,(a\land\mu Y.\,(X \lor\mondiamond{g}Y));
    \end{equation*}
    that is, the translation of game logic formulae may produce
    unguarded formulae in the monotone $\mu$-calculus.
  \item\label{item:prob-mu} The two-valued \emph{probabilistic
      $\mu$-calculus}~\cite{CirsteaEA11a,LiuEA15,ChakrabortyKatoen16}
    (not to be confused with the real-valued probabilistic
    $\mu$-calculus~\cite{HuthKwiatkowska97,McIverMorgan07}) is
    modelled using the distribution functor~$\Dist$ as discussed
    above; recall in particular that $\Dist$-coalgebras are Markov
    chains. To avoid triviality (see
    \autoref{rem:markov-triviality}), we include propositional
    atoms, i.e.~we work with coalgebras for the functor
    $F=\Dist\times\Pow(\PropAts)$ for a set~$\PropAts$ of
    propositional atoms like in the previous items.
    We use the modal similarity type
    $\Lambda=\{\langle b\rangle,[b]\mid b\in\mathbb{Q}\cap[0,1]\}\cup
    \PropAts\cup\{\neg a\mid a\in \PropAts\}$, with $\langle b\rangle$ and $[b]$
    mutually dual, and again with $a\in \PropAts$ and $\neg a$ mutually dual;
    we interpret these modalities by the predicate liftings
    \begin{align*}
      \sem{\langle b\rangle}_U(A)&=\{(d,Q)\in FU\mid d(A)> b\}\\
      \sem{[b]}_U(A)&=\{(d,Q)\in FU\mid d(U\setminus A)\le b\}\\
      \sem{a}_U & =\{(d,Q)\in FU\mid a\in Q\}\\
      \sem{\neg a}_U & =\{(d,Q)\in FU\mid a\notin Q\}
    \end{align*}
    for sets $U$ and $A\subseteq U$ (so $\sem{\langle b\rangle}$
    applies the above-mentioned predicate liftings $\lambda^b$ to the
    $\Dist$-component of~$F$); that is, a state satisfies
    $\langle b\rangle\phi$ if the probability of reaching a state
    satisfying~$\phi$ in the next step is more than~$b$, and a state
    satisfies~$[b]\phi$ if the probability of reaching a state not
    satisfying~$\phi$ in the next step is at most~$b$. For example,
    the formula
    \begin{equation*}
      \phi=\nu X.\,\mathsf{safe} \land \langle 0.95\rangle X
    \end{equation*}
    expresses that current state is safe and will reach a state in
    which~$\phi$ holds again with probability more than~$0.95$ (a
    property that may be more realistically expected to hold in
    practice than formulae demanding that safety holds forever with a
    given probability). 
  \item\label{item:graded} Similarly, we interpret the \emph{graded
      $\mu$-calculus}~\cite{KupfermanEA02} over
    \emph{multigraphs}~\cite{DAgostinoVisser02}, in which states are
    connected by directed edges that are annotated with non-negative integer
    \emph{multiplicities}. Multigraphs correspond to coalgebras for
    the \emph{multiset functor} $\Bag$, defined on sets~$U$ by
    \begin{equation*}
      \Bag(U)=\{\beta:U\to \mathbb{N}\cup\{\infty\}\}.
    \end{equation*}
    We view $d\in\Bag(U)$ as an integer-valued measure on~$U$, and for
    $V\subseteq U$ employ the same notation $d(V)=\sum_{x\in V}d(x)$
    as in the case of probability distributions. In this notation, we
    can, again, define $\Bag f$, for $f\colon U\to V$, as taking image
    measures, i.e.\ $\Bag f(d)(W)=d(f^{-1}[W])$ for $W\subseteq V$.
    A $\Bag$-coalgebra $\xi\colon C\to\Bag C$ assigns a multiset
    $\xi(x)$ to each state $x\in C$, which we may read as indicating
    that given states $x,y\in C$, there is an edge from~$x$ to~$y$
    with multiplicity $\xi(x)(y)$.  We use the modal similarity type
    $\Lambda=\{\langle m\rangle,[m]\mid
    m\in\mathbb{N}\cup\{\infty\}\}$, with~$\langle m\rangle$ and~$[m]$
    mutually dual, and define the predicate liftings
    \begin{align*}
      \sem{\langle m\rangle}_U(A)&=\{\beta\in\Bag U\mid \beta(A)>m\} \\
      \sem{[m]}_U(A)&=\{\beta\in\Bag U\mid \beta(U\setminus A)\leq m\}
    \end{align*}
    for sets $U$ and $A\subseteq U$.  For instance, somewhat
    informally speaking, a state satisfies
    $\nu X.\,(\psi\wedge\langle 1 \rangle X)$ if it is the root of an
    embedded infinite binary tree in which all states satisfy~$\psi$
    (more formally, this holds in the tree unfolding of the
    coalgebra).
  \item\label{item:prob-poly} The \emph{probabilistic $\mu$-calculus
      with polynomial inequalities}~\cite{KupkeEA15} extends the
    probabilistic $\mu$-calculus (item~\ref{item:prob-mu}) by
    introducing, in addition to propositional atoms as in
    item~\ref{item:prob-mu}, mutually dual \emph{polynomial
      modalities} $\langle p\rangle,[p]$ for $n\in\Nat$,
    $p\in \mathbb{Q}[X_1,\ldots, X_n]$ (the set of polynomials in
    variables $X_1,\dots,X_n$ with rational coefficients) such
    that~$p$ is monotone on~$[0,1]$ in all variables (in particular
    this holds when all coefficients of non-constant monomials in~$p$
    are non-negative, but also, e.g.,
    $X_1-\frac{1}{2}X_1^2-\frac{1}{4}$ is monotone on $[0,1]$). These
    modalities are interpreted by predicate liftings defined by
    \begin{align*}
      \sem{\langle p\rangle}_U(A_1,\ldots,A_n)&=\{(d,Q)\in FU\mid 
                                       p(d(A_1),\ldots,d(A_n))> 0\}\\
      \sem{[p]}_U(A_1,\ldots,A_n)&=\{(d,Q)\in FU\mid 
                                       p(d(\overline{A_1}),\ldots,d(\overline{A_n}))\leq 0\}
    \end{align*}
    for sets $U$ and $A_1,\ldots,A_n\subseteq U$. Of course, the
    polynomial modalities subsume the modalities mentioned in
    item~\ref{item:prob-mu}, explicitly via
    $\langle b\rangle =\langle X_1-b\rangle$ and $[b]=[X_1-b]$. The
    monotonicity restriction on polynomials ensures that polynomial
    modalities are monotone, which in turn is needed to guarantee
    existence of least and greatest fixpoints. Polynomial inequalities
    over probabilities have received some previous interest in
    probabilistic logics (e.g.~\cite{FaginEA90,GutierrezBasultoEA17}),
    in particular as they can express constraints on independent
    events (and hence play a role analogous to independent products as
    used in the real-valued probabilistic
    $\mu$-calculus~\cite{Mio11}). E.g.\ the formula
    \begin{equation*}
      \nu Y.\,\langle X_1X_2-0.9\rangle(\mathsf{ready}\land Y,\mathsf{idle}\land Y)
    \end{equation*}
    says roughly that two independently sampled successors of the
    current state will satisfy~$\mathsf{ready}$ and~$\mathsf{idle}$,
    respectively, and then satisfy the same property again, with
    probability more than~$0.9$.

    We note that polynomial inequalities clearly increase the
    expressiveness of the logic strictly, not only in comparison to
    the logic with operators $\langle b\rangle$, $[b]$ as in
    item~\ref{item:prob-mu} but also w.r.t.~the logic with only linear
    inequalities: The restrictions on probability distributions
    imposed by linear inequalities on probabilities are rational
    polytopes, so already $\langle X_1^2+X_2^2-1\rangle(a,c)$ is
    not expressible using only linear inequalities.
  \item\label{item:graded-poly} Similarly, the \emph{graded
      $\mu$-calculus with polynomial inequalities} extends the graded
    $\mu$-calculus with more expressive modalities
    $\langle p\rangle,[p]$, again mutually dual, where in this case
    $n\in\Nat$, $p\in\Int[X_1,\ldots,X_n]$ (that is,~$p$ ranges over
    multivariate polynomials with integer coefficients). Again, we
    restrict polynomials to be monotone in all variables; we do in
    this case in fact require that all coefficients of non-constant
    monomials are non-negative. To avoid triviality, we
    correspondingly require the coefficient~$b_0$ of the constant
    monomial to be non-positive; we refer to the number $-b_0$ as the
    \emph{index} of the modality.  These modalities are interpreted by
    the predicate liftings
    \begin{align*}
      \sem{\langle p\rangle}_U(A_1,\ldots,A_n)&=\{\beta\in\Bag U\mid
                                        p(\beta(A_1),\ldots,\beta(A_n))>0)\}\\
      \sem{[p]}_U(A_1,\ldots,A_n)&=\{\beta\in\Bag U\mid
                                        p(\beta(\overline{A_1}),\ldots,\beta(\overline{A_n}))\leq 0)\}.
    \end{align*}
    This logic subsumes the \emph{Presburger $\mu$-calculus}, that is,
    the extension of the graded $\mu$-calculus with linear
    inequalities (with non-negative coefficients), which may be seen
    as the fixpoint variant of \emph{Presburger modal
      logic}~\cite{DemriLugiez06}. E.g.\ the formula
    \begin{equation*}
      \mu Y.\,(a\lor \langle 3X_1+X_2^2-10\rangle (c\land Y,a\land Y))
    \end{equation*}
    says that (in the tree unfolding of the coalgebra) the current
    state is the root of a finite tree all whose leaves satisfy~$a$,
    and each of whose inner nodes has~$n_1$ children satisfying~$c$
    and $n_2$ children satisfying~$a$ where $3n_1+n_2^2-10>0$. The
    index of the modality $\langle 3X_1+X_2^2-10\rangle$ is~$10$.

    Unlike in the probabilistic case (item~\ref{item:prob-poly}),
    polynomial inequalities do not increase the expressiveness of the
    graded $\mu$-calculus (item~\ref{item:graded}): Given a polynomial
    $p\in\Int[X_1,\dots,X_n]$, one can just replace
    $\langle p\rangle(\phi_1,\dots,\phi_n)$ with the disjunction of
    all formulae $\Land_{i=1}^n\langle m_i-1\rangle\phi_i$ over all
    solutions $(m_1,\dots,m_n)$ of $p(X_1,\dots,X_n)>0$ that are
    minimal w.r.t.~the componentwise ordering of~$\Nat^n$. As these
    minimal solutions form an antichain in $\Nat^n$, there are only
    finitely many of them~\cite{Laver76}, so the disjunction is indeed
    finite. However, the blowup of this translation is exponential in
    the binary size of the coefficient of the constant monomial and
    in~$n$: For instance, even the linear inequality
    $X_1+\dots+X_n-nb>0$ has, by a somewhat generous estimate, at
    least $(b+1)^{n-1}$ minimal solutions (one for each assignment of
    numbers between $0$ and $b$ to the variables $X_1,\dots,X_{n-1}$;
    more precisely, the number of minimal solutions is the number
    $\binom{n+nb}{n-1}=\binom{n+nb}{nb+1}$ of weak compositions of
    $nb+1$ into~$n$ parts),
    which is exponential both in~$n$ and in the binary size of~$b$ (in
    which the polynomial itself has linear size when~$n$ is fixed).
    Therefore, this translation does not allow inheriting an
    exponential-time upper complexity bound on satisfiability checking
    from the graded $\mu$-calculus.
  \item\label{item:comp} Coalgebraic logics in general combine along
    functor composition, and essentially all their properties
    including their algorithmic treatment propagate; the arising
    composite logics are essentially typed sublogics of the fusion, to
    which we refer as \emph{multi-sorted coalgebraic
      logics}~\cite{SchroderPattinson11}. For instance, Markov
    decision processes or simple Segala systems may (in the simplest
    version) be seen as coalgebras for the composite functor
    $\Pow\circ\Dist$. A logic for such systems has nondeterministic
    modalities $\Diamond,\Box$ as well as probabilistic modalities
    such as~$\langle p\rangle$,~$[p]$ (items~\ref{item:K}
    and~\ref{item:prob-poly}). The logic distinguishes two sorts of
    \emph{nondeterministic} and \emph{probabilistic} formulae,
    respectively, with $\Diamond,\Box$ taking probabilistic formulae
    as arguments and producing nondeterministic formulae, and
    with~$\langle p\rangle$,~$[p]$ working the other way around (so
    that, e.g., $\langle X^2-0.5\rangle\Diamond\top$ is a
    probabilistic formula while
    $\Diamond\top\land\langle X^2-0.5\rangle\Diamond\top$ is not a
    formula). Two further particularly pervasive and basic functors
    are $(-)^A$, which represents indexing of transitions (and
    modalities) over a fixed set~$A$ of actions or labels (like in
    item~\ref{item:monotone-mu} above for $A=\AtGames$), and
    $(-)\times\Pow(\PropAts)$, which represents state-dependent
    valuations for a set~$\PropAts$ of propositional atoms as featured
    already in several items above. Our results will therefore cover
    modular combinations of these features with the logics discussed
    above, e.g.~logics for labelled Markov chains or Markov decision
    processes. We refrain from repeating the somewhat verbose full
    description of the framework of multi-sorted coalgebraic
    logic. Instead, we will consider only the fusion of logics in the
    following -- since, as indicated above, multi-sorted logics embed
    into the fusion~\cite{SchroderPattinson11}, this suffices for
    purposes of algorithmic satisfiability checking. We discuss the
    fusion of coalgebraic logics in Remarks~\ref{rem:fusion}
    and~\ref{rem:fusion-oss}.
  \end{enumerate}
  We remark that the modalities in items~\ref{item:prob-poly}
  and~\ref{item:graded-poly} are less general than in the
  corresponding next-step
  logics~\cite{FaginEA90,GutierrezBasultoEA17,DemriLugiez06,KupkeEA15}
  in that the polynomials involved are restricted to be monotone;
  e.g.\ they do not support statements of the type `$\phi$ is more
  probable than~$\psi$'. This is owed to their use in fixpoints, which
  requires modalities to be monotone as indicated above.
\end{exa}
\noindent Coalgebraic logics relate closely to a generic notion of
behavioural equivalence, which generalizes bisimilarity of transition
systems. A \emph{morphism} $h\colon (C,\xi)\to (D,\zeta)$ of
$F$-coalgebras is a map $h\colon C\to D$ such that the square
\begin{equation*}
  \begin{tikzcd}
    C \arrow[r,"h" above]\arrow[d,"\xi" left] & D\arrow[d,"\zeta" right]\\
    FC \arrow[r,"Fh" below] & FD
  \end{tikzcd}
\end{equation*}
commutes. For instance, a morphism $h\colon (C,\xi)\to(D,\zeta)$ of
$\Pow$-coalgebras, i.e.~of transition systems or Kripke frames, is
precisely what is known as a bounded morphism or a
p-morphism~(e.g.~\cite{BlackburnEA01}), that is, (i) whenever
$c'\in\xi(c)$ for $c,c'\in C$, then $h(c')\in\zeta(h(c))$, and (ii)
whenever $d'\in\zeta(d)$ for $d,d'\in D$, then there exists
$c'\in\xi(c)$ such that $h(c')=d'$. States $x\in C$, $y\in D$ in
coalgebras $(C,\xi)$, $(D,\zeta)$ are \emph{behaviourally equivalent}
if there exist a coalgebra $(E,\theta)$ and morphisms
$g\colon (C,\xi)\to (E,\theta)$, $h\colon(D,\zeta)\to(E,\theta)$ such
that $g(x)=h(y)$. This notion instantiates to standard equivalences in
our running examples; e.g.~as indicated above, states in (labelled)
transition systems are behaviourally equivalent iff they are bisimilar
in the usual sense~\cite{AczelMendler89}, and states in labelled
Markov chains are behaviourally equivalent iff they are
probabilistically bisimilar~\cite{RuttenDeVink99}, 
\cite{BartelsEA04}. 

In view of the definition of behavioural equivalence via morphisms of
coalgebras, we can phrase invariance of the coalgebraic $\mu$-calculus
under behavioural equivalence, generalizing the well-known
bisimulation invariance of the relational $\mu$-calculus, as
invariance under coalgebra morphisms. Formally, this takes the
following shape:

\begin{lem}[Invariance under behavioural equivalence]\label{lem:invariance}
  Let~$h\colon (C,\xi)\to(D,\zeta)$ be a morphism of coalgebras,
  let~$\phi$ by a coalgebraic $\mu$-calculus formula, and let
  $i\colon\Var\pfun \Pow(D)$ be a valuation. Then
  \begin{equation*}
    \sem{\phi}_{h^{-1}i} =  h^{-1}\sem{\phi}_i
  \end{equation*}
  where $h^{-1}i$ denotes the valuation given by
  $h^{-1}i(X)=h^{-1}[i(X)]$.
\end{lem}
\noindent Since as mentioned above, fixpoints can be approximated by
ordinal-indexed iteration, this follows from the fact that
(fixpoint-free) coalgebraic modal logic with infinite conjunctions and
disjunctions is invariant under behavioural
equivalence~\cite{Pattinson04,Schroder08}. Similar statements have
been shown for a version of the coalgebraic $\mu$-calculus featuring
the coalgebraic cover modality instead of predicate-lifting-based
modalities~\cite{Venema04} and for the single-variable fragment of the
coalgebraic $\mu$-calculus in the present
sense~\cite{SchroderVenema18}, using essentially the same argument.
\begin{rem}\label{rem:markov-triviality}
  As indicated in \autoref{ex:logics}.\ref{item:prob-mu}, Markov
  chains are coalgebras for the discrete distribution
  functor~$\Dist$. The reason why we immediately extended Markov
  chains with propositional atoms is that in plain Markov chains, all
  states are behaviourally equivalent and therefore satisfy the same
  formulae of any coalgebraic $\mu$-calculus interpreted over Markov
  chains: Since for any singleton set~$1$, $\Dist 1$ is again a
  singleton, we have a unique $\Dist$-coalgebra structure on~$1$, and
  every $\Dist$-coalgebra has a (unique) morphism into this
  coalgebra. This collapse under behavioural equivalence is avoided by
  adding propositional atoms.
\end{rem}

\begin{rem}[Multigraph semantics of the graded $\mu$-calculus]
  One important consequence of the invariance of the coalgebraic
  $\mu$-calculus under coalgebra morphisms according to
  \autoref{lem:invariance} is that the multigraph semantics of the
  graded $\mu$-calculus as per
  \autoref{ex:logics}.\ref{item:graded} (and its extension with
  polynomial inequalities introduced in
  \autoref{ex:logics}.\ref{item:graded-poly}) is equivalent to the
  more standard Kripke semantics, i.e.~the semantics over
  $\Pow$-coalgebras~\cite{KupfermanEA02}. We can obtain the latter
  from the definitions in \autoref{ex:logics}.\ref{item:graded} by
  converting Kripke frames into multigraphs in the expected way,
  i.e.~by regarding transitions in the given Kripke frame as
  transitions with multiplicity~$1$ in a multigraph. The conversion
  shows trivially that every formula that is satisfiable over (finite)
  Kripke frames is also satisfiable over (finite) multigraphs. The
  converse is shown by converting a given multigraph into a Kripke
  model that has copies of states and transitions according to the
  multiplicity of transitions in the multigraph. The arising Kripke
  frame, converted back into a multigraph as described previously, has
  a coalgebra morphism into the original multigraph, implying by
  \autoref{lem:invariance} that the copies satisfy the same formulae
  as the original multigraph states. Details are given
  in~\cite[Lemma~2.4]{SchroderVenema18}.
\end{rem}

\begin{rem}[Fusion]\label{rem:fusion}
  The standard term \emph{fusion} refers to a form of combination of
  logics, in particular modal logics, where the ingredients of both
  logics are combined disjointly and essentially without semantic
  interaction, but may then be intermingled in composite formulae. In
  the framework of coalgebraic logic, this means that the fusion of
  coalgebraic logics with modal similarity types~$\Lambda_i$
  interpreted over functors~$F_i$, $i=1,2$, is formed by taking the
  disjoint union~$\Lambda$ of~$\Lambda_1$ and~$\Lambda_2$ (assuming
  w.l.o.g.\ that~$\Lambda_1$ and~$\Lambda_2$ are disjoint to begin
  with) as the modal similarity type, and by interpreting modalities
  over the product functor $F=F_1\times F_2$ (with the product of
  functors computed componentwise, e.g.~$FU=F_1U\times F_2U$ for
  sets~$U$). The predicate lifting $\Sem{\hearts}^F$ interpreting
  $\hearts\in\Lambda_i$ over~$F$, for $i=1,2$, is given by
  \begin{equation*}
    \Sem{\hearts}^F_U(A)=\{(t_1,t_2)\in F_1U\times F_2U\mid t_i\in\Sem{\hearts}^{F_i}_U(A)\}
  \end{equation*}
  where $\Sem{\hearts}^{F_i}$ is the predicate lifting for~$F_i$
  interpreting~$\hearts$ in the respective component logic.  For
  example, the fusion of the standard relational modal $\mu$-calculus
  (\autoref{ex:logics}.\ref{item:K}) and the probabilistic
  $\mu$-calculus with polynomial inequalities
  (\autoref{ex:logics}.\ref{item:prob-poly}) allows for
  unrestricted use of non-deterministic modalities $\Diamond,\Box$ and
  probabilistic modalities~$\langle p\rangle$, $[p]$ in formulae such
  as $\Diamond\top\land\langle X^2-0.5\rangle\Diamond\top$, instead of only
  the alternating discipline imposed by the two-sorted logic for
  Markov decision processes described in
  \autoref{ex:logics}.\ref{item:comp}; such formulae are
  interpreted over coalgebras for the functor $\Pow\times\Dist$, in
  which every state has both nondeterministic and probabilistic
  successors.
\end{rem}

\paragraph{Closure and Alternation Depth}

A key measure of the complexity of a formula is its
\emph{alternation depth}, which roughly speaking describes the maximal
number of alternations between~$\mu$ and~$\nu$ in chains of
dependently nested fixpoints. The treatment of this issue is simplified
if one excludes reuse of fixpoint variables:
\begin{defi}[Clean formulae]
  A closed formula~$\phi$ is \emph{clean} if each fixpoint variable
  appears in at most one fixpoint operator in~$\phi$.
\end{defi}
\noindent \emph{We fix a clean closed target formula~$\target$ for the
  remainder of the paper}, and define the following syntactic notions
relative to this target formula. Given $X\in\BV(\target)$, we write
$\theta(X)$ for the unique subformula of~$\target$ of the shape
$\eta X.\,\psi$. We say that a fixpoint variable~$X$ is an
\emph{$\eta$-variable}, for $\eta\in\{\mu,\nu\}$, if~$\theta(X)$ has
the form $\eta X.\,\psi$. As indicated, the formal definition of
alternation depth simplifies within clean formulae
(e.g.,~\cite{Wilke01,KupkeEA20,KupkeEA22b}):
\begin{defi}[Alternation depth]\label{def:ad}
  For $X,Y\in\BV(\target)$, we write 
  $Y\depord X$ if $X\in\FV(\theta(Y))$ (which implies that $\theta(Y)$
  is a subformula of $\theta(X)$). A \emph{dependency chain}
  in~$\target$ is a chain of the form
  \begin{equation*}
    X_n\depord X_{n-1}\depord\dots \depord X_0\qquad(n\ge 0);
  \end{equation*}
  the \emph{alternation number} of the chain is $k+1$ where~$k$ is the
  number of indices $i\in\{0,\dots,n-1\}$ such that~$X_i$ is a
  $\mu$-variable iff $X_{i+1}$ is a $\nu$-variable (i.e.~the chain
  toggles between $\mu$- and $\nu$-variables at~$i$). The
  \emph{alternation depth} $\ad(X)$ of a variable $X\in\BV(\target)$
  is the maximal alternation number of dependency chains as above
  ending in~$X_0=X$, and the \emph{alternation depth} of~$\target$ is
  $\ad(\target):=\max\{\ad(X)\mid X\in\BV(\target)\}$.
\end{defi}
\noindent (In particular, the alternation depth of~$\target$ is~$0$
if~$\target$ does not contain any fixpoints, and ~$1$ if~$\target$ is
\emph{alternation-free}, i.e.~no $\nu$-variable occurs freely in
$\theta(X)$ for a $\mu$-variable~$X$, and vice versa.)
\begin{exa}[Alternation depth]
  The formula $\nu X.\,(\mu Y.\ a\lor \Diamond Y)\land\Box X$ (in the
  relational modal $\mu$-calculus) has alternation depth~$1$, i.e.~is
  alternation-free. The formula
  $\nu X.\,\Box(\mu Y.\ X\lor \Diamond Y)$ has alternation depth~$2$,
  as witnessed by the alternating chain $Y\depord X$.
\end{exa}
\noindent In the automata-theoretic approach to checking for infinite
deferrals, automata states will be taken from the Fischer-Ladner
closure of~$\target$ in the usual sense~\cite{Kozen83}, in which
subformulae of~$\target$ are expanded into closed formulae by means of
fixpoint unfolding:
\begin{defi}[Fischer-Ladner closure]
The
\emph{Fischer-Ladner closure} $\FLtarget$ of $\target$ is the least
set of formulae containing~$\target$ such that
\begin{align*}
  \phi\land\psi\in\FLtarget & \implies \phi,\psi\in\FLtarget\\
  \phi\lor\psi\in\FLtarget & \implies \phi,\psi\in\FLtarget\\
  \hearts\phi\in\FLtarget & \implies \phi\in\FLtarget\\
  \eta X.\,\phi\in\FLtarget & \implies \phi[\eta X.\,\phi/X]\in\FLtarget.
\end{align*}
\end{defi}
\noindent (One should note that although~$\chi$ is clean, the elements
of $\FLtarget$ will in general fail to be clean, as fixpoint unfolding
of $\eta X.\,\phi$ as per the last clause may create multiple copies
of~$\eta X.\,\phi$.)

Furthermore, we let $\mathsf{sub}(\target)$ denote the set of
subformulae of~$\target$; unlike formulae in~$\FLtarget$, formulae in
$\sub(\target)$ may contain free fixpoint variables. The
\emph{innermost} free fixpoint variable in a subformula of~$\target$
is the one whose binder lies furthest to the right
in~$\target$. Each $\phi\in\mathsf{sub}(\target)$ induces a formula
$\theta^*(\phi)\in\FLtarget$,
which is obtained by repeatedly
transforming~$\phi$, in each step substituting the innermost fixpoint
variable~$X$ occurring in the present formula with
$\theta(X)$~\cite{Kozen83}. This map witnesses finiteness of the
closure (\autoref{lem:kozen}) and moreover will serve as a
connection between two variants of the model checking game
respectively based on subformulae and on the closure
(\autoref{lem:sfgame-bounded-mor}). For instance, for
$\chi=\mu X.\,\nu Y.\,(Y\vee\hearts X)$, we have
$Y\vee\hearts X\in \mathsf{sub}(\target)$, $\theta(X)=\chi$ and
$\theta(Y)=\nu Y.\,(Y\vee\hearts X)$, and thus
\begin{align*}
  \theta^*(Y\vee\hearts X)&=\theta^*(\theta(Y)\vee\hearts X)\\
                          &=\theta^*((\nu Y.\,(Y\vee\hearts X))\vee\hearts X)\\
                          &=(\nu Y.\,(Y\vee\hearts \theta(X)))\vee\hearts\theta(X)\\
                          &=(\nu Y.\,(Y\vee\hearts \chi))\vee\hearts\chi.
\end{align*}\smallskip

\begin{lemC}[\cite{Kozen83}]\label{lem:kozen} The function
  $\theta^*\colon\mathsf{sub}(\target)\to\FLtarget$ is surjective; in
  particular, the cardinality of~$\FLtarget$ is at most the number of
  subformulae of~$\target$.
\end{lemC}

\begin{rem}[Closure vs.\ subformulae]
  The main reason we base our constructions on the closure~$\FLtarget$
  is that this provides the most succinct size measure for
  $\mu$-calculus formulae~\cite{BruseEA15}, thus strengthening our
  complexity results. Subformulae are needed occasionally for
  technical purposes; in particular, they support proofs by structural
  induction.
\end{rem}

\section{Tracking Automata and Model Checking Games}\label{sec:tracking}

Generally, the main problem both in satisfiability checking for
temporal logics and the associated model constructions on the one
hand, and in model checking on the other hand, is to ensure that the
unfolding of least fixpoints does not lead to infinite deferrals,
i.e.~that least fixpoints are indeed eventually satisfied. We encode
this condition using \emph{parity automata} (e.g.~\cite{GraedelEA02})
that track the evolution of formulae in such procedures.

Recall that a \emph{(nondeterministic) parity automaton} is a tuple
\begin{equation*}
\autom=(V,\Sigma,\Delta,\initstate,\alpha)
\end{equation*}
where $V$ is a set of \emph{nodes}; $\Sigma$ is a finite set, the
\emph{alphabet}; $\Delta\subseteq V\times \Sigma\times V$ is the
\emph{transition relation}; $\initstate\in V$ is the \emph{initial
  node}; and $\alpha:V\to\mathbb{N}$ is the \emph{priority function},
which assigns priorities $\alpha(v)\in\mathbb{N}$ to states $v\in
V$. Given a ternary relation
$R\subseteq A\times B\times A$ and $a\in A$, $b\in B$,
$B'\subseteq B$, we generally write
$R(a,b)=\{a'\in A\mid (a,b,a')\in R\}$,
$R(a,B')=\bigcup_{b\in B'} R(a,b)$ and $R(a)=R(a,B)$.
 If $\Delta$ is a (partial) functional relation,
then~$\autom$ is said to be \emph{deterministic}, and we denote the
corresponding partial function by $\delta:V\times\Sigma\pfun V$. We
often treat infinite sequences $s=x_1,x_2,\dots$ over a base set~$X$ as
maps $s\colon\Nat\to X$. We generally write
\begin{equation*}
  \Inf(s)=\{x\in X\mid |s^{-1}[\{x\}]|=\infty\}
\end{equation*}
for the set of elements that occur infinitely often in~$s$.  The
automaton $\autom$ \emph{accepts} an \emph{infinite word}, i.e.\ an
infinite sequence $w=\sigma_0,\sigma_1,\ldots\in\Sigma^\omega$ over~$\Sigma$, if
there is a $w$-path through $\autom$ on which the highest priority
that is passed infinitely often is even; formally, the language that
is accepted by $\autom$ is defined by
\begin{equation*}
  L(\autom)=\{w\in\Sigma^\omega\mid \exists
  \rho\in\mathsf{run}(\autom,w).\,
  \max(\Inf(\alpha\circ\rho))\text{ is even}\},
\end{equation*}
where $\mathsf{run}(\autom,w)$ denotes the set of \emph{runs
  of~$\autom$ on~$w$}, i.e.\ infinite sequences
$q_0,q_1,q_2,\ldots\in V^\omega$ (starting in the initial
state~$q_0=\initstate$) such that $q_{i+1}\in \Delta(q_i,w_i)$ for all
$i\geq 0$.

Recall moreover that we are working with a fixed clean closed target
formula $\target$. We put
\begin{equation*}
  n_0=|\target|,\quad n_1=\size(\target), \quad k=\ad(\phi).
\end{equation*}
The states of the tracking automaton for~$\target$ will be the
elements of the Fischer-Ladner closure~$\FLtarget$ of~$\target$ (cf.\
\autoref{sec:prelims}); in particular, $|\FLtarget|\le n_0$.

\begin{defi}[Modal literals]
  Given a set~$Z$, we write
    $\Lambda(Z)=\{\hearts z\mid \hearts\in\Lambda, z\in Z\}$,
  and refer to elements of $\Lambda(Z)$ as \emph{modal literals}
  (over~$Z$).
\end{defi}
\noindent In particular, $\FLtarget\cap\Lambda(\FLtarget)$ is the set
of modal literals in~$\FLtarget$. We put
\begin{equation*}
  \mathsf{selections}=\Pow(\FLtarget\cap\Lambda(\FLtarget)),
\end{equation*}
and indeed refer to elements of this set as
\emph{selections}, with a view to using selections as letters for modal
steps in our automata construction. Furthermore, we let
$\FLtarget_\vee=\{\psi\vee\chi\mid \psi\vee\chi\in\FLtarget\}$ denote
the set of disjunctive formulae contained in $\FLtarget$, and put
\begin{align*}
\mathsf{choices}=\{\tau:\FLtarget_\vee\to\FLtarget\mid
\forall (\psi\vee\chi)\in \FLtarget_\vee.\, \tau(\psi\vee\chi)\in\{\psi,\chi\}\};
\end{align*}
that is, $\mathsf{choices}$ consists of \emph{choice functions} that
pick disjuncts from disjunctions. These choice functions will be used
as letters for propositional steps in the tracking automaton.
We note $|\mathsf{selections}|,|\mathsf{choices}|\leq 2^{n_0}$.
\begin{defi}[Tracking automaton]\label{def:tracking} The \emph{tracking automaton}
  for~$\target$ is the nondeterministic parity automaton
  $\autom_\target=(\FLtarget,\Sigma,\Delta,\initstate,\alpha)$ where
  $\initstate=\target$, 
$\Sigma= \mathsf{choices}\cup\mathsf{selections}$,
and for $\psi\in\FLtarget$, $\tau\in\mathsf{choices}$ and $\kappa\in\mathsf{selections}$,
\begin{align*}
  \Delta(\psi,\tau)
  &= \begin{cases}
                       \{\tau(\psi)\} & \text{if $\psi\in\FLtarget_\vee$}\\
                       \{\psi_0,\psi_1\} &  \text{if $\psi=\psi_0\wedge\psi_1$}\\
                       \{\psi_1[\psi/X]\} & \text{if $\psi=\eta X.\,\psi_1$} \\
                       \{\psi\} & \text{if $\psi=\hearts\psi_0$ for some $\hearts\in\Lambda$}\\
  					   \emptyset & \text{if $\psi\in\{\top,\bot\}$}
                     \end{cases}\\
  \Delta(\psi,\kappa)&=
                       \begin{cases}\{\psi_0\} & \text{if $\psi=\hearts\psi_0\in\kappa$ for some $\hearts\in\Lambda$}\\
                         \emptyset &\text{otherwise}
                       \end{cases}
\end{align*}
The priority function $\alpha$ is derived from the alternation depths
of variables, counting only unfoldings of fixpoints (i.e.\ all other
formulae have priority~$1$) and ensuring that least fixpoints receive
even priority and greatest fixpoints receive odd priority. That is, we
put
\begin{align*}
  \alpha(\mu X.\phi) & = 2\lfloor(\ad(X)-1)/2\rfloor+2\\
  \alpha(\nu X.\phi) & = 2\lfloor \ad(X)/2\rfloor+1\\
  \alpha(\psi) & =1 && \text{if $\psi$ is not a fixpoint literal.} 
\end{align*}
\end{defi}

(For instance, a formula of the form
$\chi_1=\nu X.\,\mu Y.\nu Z.\,\phi(X,Y,Z)$ has $\ad(X)=3$ and
$\alpha(\chi_1)=3$, while its unfolding
$\chi'_1=\mu Y.\nu Z.\,\phi(\chi_1,Y,Z)$ has $\ad(Y)=2$ and
$\alpha(\chi_1')=2$. Contrastingly, a formula of the form
$\chi_2=\mu X.\,\nu Y.\mu Z.\,\phi(X,Y,Z)$ has $\ad(X)=3$ and
$\alpha(\chi_2)=4$, while its unfolding
$\chi'_2=\nu Y.\mu Z.\,\phi(\chi_2,Y,Z)$ has $\ad(Y)=2$ and
$\alpha(\chi_2')=3$.)

Propositional tracking along a choice function
$\tau\in\mathsf{choices}$ thus follows the choice of $\tau$ for
disjunctions. Conjunctions are tracked nondeterministically to one of
their conjuncts; fixpoint literals $\psi=\eta X.\,\psi_1$ are tracked to
their unfolding $\psi_1[\psi/X]$; modal literals are left unchanged by
propositional tracking; and truth constants~$\top,\bot$ are not
further tracked at all.  In modal tracking along a selection $\kappa$,
a modal literal $\psi=\hearts\psi_0$ is tracked (to~$\psi_0$) only if
$\hearts\psi_0\in\kappa$, i.e.~if~$\kappa$ selects $\hearts\psi_0$ to
be tracked.

The priority function~$\alpha$ of~$\autom_\target$ is designed to
ensure that a run~$\rho$ -- that is, a sequence of formulae -- is
accepting iff a least fixpoint formula~$\psi$ is unfolded infinitely
often on~$\rho$ without being dominated by any outer fixpoint
formula~$\phi$, i.e.~one with $\ad(\phi)>\ad(\psi)$.  Here, we use the
term \emph{dominated} to indicate both the greater alternation depth
of~$\phi$ and the fact that~$\phi$ is also unfolded infinitely
often. As indicated above, the model checking game will relate closely
to non-deterministic tracking, and the proof of its correctness
(\autoref{thm:satisfaction-game}) will clarify that alternation
depth indeed provides an adequate mechanism to detect which of two
fixpoints is the inner one. For purposes of the nondeterministic
tracking automaton, the alphabet~$\Sigma$ is in fact overlarge, and
could be reduced to just individual choices picking a disjunct from a
single disjunction; the importance of~$\Sigma$ arises in the
(co-)determinization of~$\autom_\target$, where it ensures that enough
branching is retained.

\begin{exa}[Tracking automaton]\label{ex:ntrack}
For an example of the tracking automaton construction,
recall from \autoref{ex:logics}.\ref{item:monotone-mu} the 
monotone $\mu$-calculus formula $\chi=\nu X.\,(a\land\mu Y.\,(X \lor\mondiamond{g}Y))$
obtained from the game logic formula $\mondiamond{(g^*)^\times}a$.
As Fischer-Ladner closure of $\chi$, we have
\begin{align*}
\FLtarget=\{\chi,\,
a\land\phi,\,
a,\, \phi, \,
\chi \lor\mondiamond{g}\phi,\,
\mondiamond{g}\phi
\}
\end{align*}
with $\phi$ abbreviating the formula $\mu Y.\,(\chi \lor\mondiamond{g}Y)$.
Furthermore, we have $\mathsf{ad}(X)=2$ and $\mathsf{ad}(Y)=1$.
For this small example we have $\FLtarget_\vee=\{\chi \lor\mondiamond{g}\phi\}$, 
so that $\mathsf{choices}$ consists of just the two functions
$\tau_l$ and $\tau_r$, defined by $\tau_l(\chi \lor\mondiamond{g}\phi)=\chi$
and $\tau_r(\chi \lor\mondiamond{g}\phi)=\mondiamond{g}\phi$,
respectively. Omitting the treatment of propositional atoms as modal operators,
we have $\FLtarget\cap\Lambda(\FLtarget)=\{\mondiamond{g}\phi\}$
so that
$\mathsf{selections}=\{\emptyset,\{\mondiamond{g}\phi\}\}$.
Then we obtain the tracking automaton $\autom_\target$ depicted below,
where $\tau$ stands for any letter from $\mathsf{choices}$ and
$\kappa_{\mondiamond{g}\phi}$ denotes the letter $\{\mondiamond{g}\phi\}\in\mathsf{selections}$.
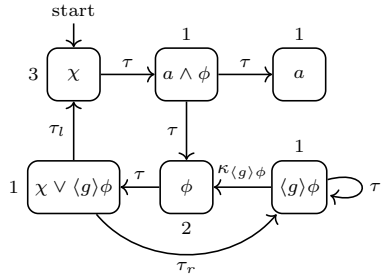
\begin{center}
\begin{tiny}
  \begin{tikzpicture}[
		auto,
    node distance=1.5cm,
    semithick
    ]
     \node[state,rectangle,rounded corners,initial above,label={left: $3$}] (e) {$\chi$};
     \node[state,rectangle,rounded corners,label={above: $1$}] (f) [right of=e] {$a\wedge\phi$};
     \node[state,rectangle,rounded corners,label={below: $2$}] (g) [below of=f] {$\phi$};
     \node[state,rectangle,rounded corners,label={left: $1$}] (h) [below of=e] {$\chi\vee\mondiamond{g}\phi$};
     \node[state,rectangle,rounded corners,label={above: $1$}] (i) [right of=g] {$\mondiamond{g}\phi$};
     \node[state,rectangle,rounded corners,label={above: $1$}] (a) [right of=f] {$a$};
     \path[->] (e) edge node [above] {$\tau$} (f);
     \path[->] (f) edge node [left] {$\tau$} (g);
     \path[->] (f) edge node [above] {$\tau$} (a);
     \path[->] (g) edge node [above] {$\tau$} (h);
     \path[->] (h) edge node [left] {$\tau_l$} (e);
     \path[->] (h) edge [bend right=50] node [below] {$\tau_r$} (i);
     \path[->] (i) edge node [above] {$\kappa_{\mondiamond{g}\phi}$} (g);
     \path[->] (i) edge [loop right] node [right] {$\tau$} (i);

  \end{tikzpicture}
\end{tiny}
\end{center}
Thus $\autom_\target$ accepts infinite words over
$\mathsf{choices}\cup\mathsf{selections}$ that branch to the left on
the disjunction $\chi \lor\mondiamond{g}\phi$ only finitely often, and
on which the automaton can infinitely often read the letter
$\{\mondiamond{g}\phi\}$ in the node $\mondiamond{g}\phi$ (and avoid
reading letters from $\mathsf{selections}$ at any other node).  Such
words have a run in which the maximal priority that is visited
infinitely often is $2$; as intended, such words encode situations
where the least fixpoint $\phi$ is unfolded infinitely often while the
outer greatest fixpoint $\chi$ is unfolded only finitely often.  We
point out that the fixpoint variable $X$ is unguarded in $\chi$; this
induces the left cycle in $\autom_\target$, on which no letter from
$\mathsf{selections}$ is ever read.
\end{exa}

\begin{rem}
  The above definition of tracking automata deviates from the one we
  used in earlier work~\cite{HausmannSchroder19} in two respects: We
  now attach priorities to the states of the automaton rather than to
  its transitions; and we have changed the propositional part of the
  alphabet in such a way that every top-level propositional or
  fixpoint operator is processed when a propositional letter is read
  (while we previously had a separate letter for every conjunction,
  disjunction, and fixpoint literal in~$\FLtarget$). Both choices
  serve mainly to ease and clarify the presentation.
\end{rem}

\noindent The non-deterministic parity automaton $\autom_\target$
introduced above has size at most~$n_0$ and priorities at most~$1$ to~$k+1$.
In order to use ${L(\autom_\target)}$ as an objective in our upcoming
satisfiability games, we require a deterministic automaton accepting
this language.  To this end, we use a standard construction
(e.g.~\cite{KKV01}) to transform $\mathsf{\autom_\target}$ into an
equivalent B\"uchi automaton of size $n_0k'$ (which has additional
states $(\phi,i)$ where $\phi\in\FLtarget$ and $1\le i\le k+1$ is
even) where
\begin{equation}
  k'=\lfloor (k+1)/2\rfloor+1\in\mathcal{O}(k).\label{eq:k-prime}
\end{equation}
Then we
determinize the B\"uchi automaton using, e.g., the
Safra/Piterman-construction~\cite{Safra88,Piterman07} and obtain an
equivalent deterministic parity automaton with $\prios$ priorities and
size $\mathcal{O}(\detsize)$.  Alternatively, direct determinization
from parity automata to parity automata~\cite{ScheweVarghese14} can be
used.  Finally we complement the obtained automaton by decreasing
every priority by~$1$, obtaining a deterministic parity automaton
\begin{equation}\label{eq:det-automaton}
  \mathsf{B}_\target=(\detcarrier, \Sigma,\delta,\initstate,\detprio)
\end{equation}
with priorities $0$ to $\prios-1$ and of size $\mathcal{O}(\detsize)$
such that $L(\mathsf{B}_\target)=\overline{L(\autom_\target)}$, i.e.\
$\mathsf{B}_\target$ is a deterministic parity automaton that accepts
the words that encode sequences of fixpoint unfoldings without
infinite deferral of least fixpoints.

We refer to $\mathsf{B}_\target$ as the \emph{co-determinized tracking
  automaton}. The states of~$\mathsf{B}_\target$ are like macrostates
in the standard powerset construction for finite-word automata, but
instead of being mere sets of states, they organize the states of the
original automaton into a tree structure. Due to
the preceding conversion of $\autom_\target$ into a Büchi automaton,
the tree nodes are labelled with sets of pairs consisting of a formula
in~$\FLtarget$ and a priority. We define a labelling function
\begin{equation*}
  l^A\colon\detcarrier\to\Pow(\FLtarget)
\end{equation*}
that maps each state~$q$ of~$\mathsf{B}_\target$ (e.g.\ a Safra tree)
to the set of formulae that occur in~$q$. That is, the labelling
function forgets the structuring of the set of formulae in macrostates
(in fact, for compact Safra trees~\cite{Piterman07}, $l^A(q)$ can be
obtained from the label of the root of~$q$ by forgetting the
priorities). We do not need to know anything about how determinization
works, except the following fact: in all (history-)determinization
procedures that we refer to in this work, labels in the above sense
evolve under transitions like macrostates in the standard powerset
construction, i.e.~we have
\begin{equation}
  \label{eq:det-label}
  l^A(\delta(q,\sigma)) = \bigcup\{\Delta(\psi,\sigma) \mid \psi\in l^A(q)\}\qquad\text{for $q\in\detcarrier$, $\sigma\in\Sigma$, $\psi\in\FLtarget$.}
\end{equation}
In particular,
\begin{equation}\label{eq:det-label-modal}
  l^A(\delta(q,\kappa)) = \{\psi\mid\hearts\psi\in\kappa\cap l^A(q)\}\qquad\text{for $q\in\detcarrier$, $\kappa\in\mathsf{selections}$}.
\end{equation}
Note also that when~$\mathsf{B}_\target$ reads a choice function in
node~$q$, then all top-level propositional operators and fixpoints in
$l^A(q)$ are processed in parallel, i.e.\ one disjunct is picked from
each disjunction, every conjunction is decomposed into its conjuncts,
and every fixpoint is unfolded.

\begin{exa}[Co-determinized tracking automaton] \label{ex:dtrack}
  To give a flavour of the co-determini\-zed tracking automaton, we go
  back to the non-deterministic automaton $\autom_\target$ for the
  monotone $\mu$-calculus formula
  $\chi=\nu X.\,(a\land\mu Y.\,(X \lor\mondiamond{g}Y))$ as detailed
  in \autoref{ex:ntrack}. Applying the above-mentioned method of
  co-determinization via an intermediate B\"uchi automaton and
  Safra/Piterman trees yields an automaton that is too large for
  presentation here. Instead we use as~$\mathsf{B}_\target$ the
  equivalent but minimized automaton given below.  For readability,
  the diagram omits an accepting sink state to which all but two modal
  transitions lead; the only modal transitions that do not lead to
  this sink state are the two
  $\kappa_{\mondiamond{g}\phi}$-transitions that are explicitly
  shown.
  \\
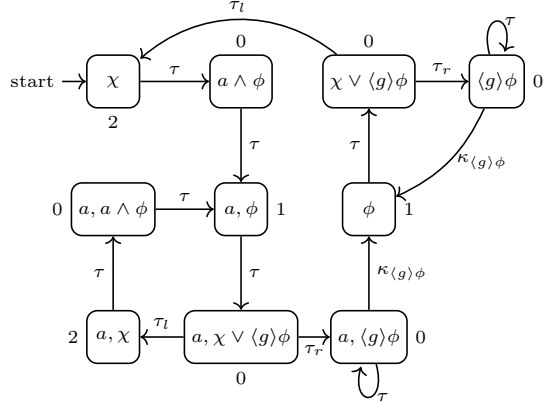
\begin{center}
\begin{tiny}
  \begin{tikzpicture}[
		auto,
    node distance=1.7cm,
    semithick
    ]
     \node[state,rectangle,rounded corners,initial left,label={below: $2$}] (e) {$\chi$};
     \node[state,rectangle,rounded corners,label={above: $0$}] (f) [right of=e] {$a\wedge\phi$};
     \node[state,rectangle,rounded corners,label={right: $1$}] (g) [below of=f] {$a,\phi$};
     \node[state,rectangle,rounded corners,label={below: $0$}] (h) [below of=g] {$a,\chi\vee\mondiamond{g}\phi$};
     \node[state,rectangle,rounded corners,label={right: $0$}] (i) [right of=h] {$a,\mondiamond{g}\phi$};
     \node[state,rectangle,rounded corners,label={right: $1$}] (j) [above of=i] {$\phi$};
     \node[state,rectangle,rounded corners,label={above: $0$}] (k) [above of=j] {$\chi\vee\mondiamond{g}\phi$};
     \node[state,rectangle,rounded corners,label={right: $0$}] (l) [right of=k] {$\mondiamond{g}\phi$};

     \node[state,rectangle,rounded corners,label={left: $2$}] (a) [left of=h] {$a,\chi$};
     \node[state,rectangle,rounded corners,label={left: $0$}] (b) [above of=a] {$a,a\land\phi$};
     \path[->] (e) edge node [above] {$\tau$} (f);
     \path[->] (f) edge node [right] {$\tau$} (g);
     \path[->] (g) edge node [right] {$\tau$} (h);
     \path[->] (h) edge node [above] {$\tau_l$} (a);
     \path[->] (a) edge node [left] {$\tau$} (b);
     \path[->] (b) edge node [above] {$\tau$} (g);
     \path[->] (h) edge node [below] {$\tau_r$} (i);
     \path[->] (i) edge node [right] {$\kappa_{\mondiamond{g}\phi}$} (j);
     \path[->] (i) edge [loop below] node [right] {$\tau$} (i);
     \path[->] (l) edge [loop above] node [right] {$\tau$} (l);
     \path[->] (j) edge node [left] {$\tau$} (k);
     \path[->] (k) edge node [above] {$\tau_r$} (l);
     \path[->] (l) edge [bend left=20] node [right] {$\kappa_{\mondiamond{g}\phi}$} (j);
     \path[->] (k) edge [bend right=40] node [above] {$\tau_l$} (e);

  \end{tikzpicture}
\end{tiny}
\end{center}
Nodes in this automaton are labelled with sets of formulae as shown,
according to the labelling function
$l^A:\detcarrier\to\Pow(\FLtarget)$ (however, we emphasize that a node
need not be uniquely determined by its label, even though this happens
to be the case in the example). For instance, let $q$ denote the node
in the lower left corner of the automaton; then
$l^A(q)=\{a,\target\}$.  Acceptance is dual to the tracking
automaton~$\autom_\target$. That is, an infinite word~$w$ is accepted
by~$\mathsf{B}_\target$ if either a)~$w$ picks the left disjunct from
$\chi\vee\mondiamond{g}\phi$ infinitely often, ensuring that the left
part of the automaton, and thereby also a node with priority~$2$, is
visited infinitely often; or b) $w$ contains only finitely many
letters from $\mathsf{selections}$ (so the automaton eventually loops
forever at the bottom right node); or c)~$w$ contains a letter from
$\mathsf{selections}$ at a position such that~$\mathsf{B}_\target$ is,
after reading the word up to this position, in a node that does not
contain $\mondiamond{g}\phi$ in its label (so the run ends up in the
accepting sink state).
\end{exa}

\begin{rem}
  It has been noted that for the relational $\mu$-calculus, tracking
  automata for \emph{aconjunctive} formulae are
  \emph{limit-deterministic} parity
  automata~\cite{HausmannEA18}. These considerably simpler automata
  can be determinized to deterministic parity automata of size
  $\mathcal{O}((n_0k')!)$ and with $\prios$
  priorities~\cite{EsparzaKRS17, HausmannEA18} (with~$k'$ as
  per~\eqref{eq:k-prime}), an observation that can also be used for
  the tracking automata in this work. For aconjunctive formulae, one
  thus has a correspondingly improved bound on the runtime of our
  satisfiability checking algorithm than stated for the general case
  in \autoref{lem:complexity} below.

  It has also been shown that tracking automata for \emph{guarded} and
  \emph{alternation-free} formulae can be seen as co-B\"uchi
  automata~\cite{FriedmannLatteLange13}
  so that the simpler determinization procedure for co-B\"uchi
  automata~\cite{MiyanoHayashi1984} can be used for such formulae,
  and the resulting satisfiability games have a B\"uchi objective
  rather than the more involved parity objective required in the general case.
  However, in our setting the tracking automata have to correctly
  deal with unguarded fixpoint variables and hence assign
  priorities greater than $1$ exclusively to fixpoint formulae
  (for instance priority $2$ is assigned to formulae of the shape
  $\mu X. (\psi\vee\Diamond X)$ but priority $1$ is assigned to formulae
  of the shape $\psi\vee\Diamond (\mu X. (\psi\vee\Diamond X))$ even though the
  latter formula arises by unfolding the fixpoint in the former formula)
  so that our tracking automata are not immediately co-B\"uchi automata when
  the target formula is alternation-free. Having said that, it indeed is possible to
  use co-B\"uchi automata for \emph{unguarded} alternation-free formulae,
  namely by separating
  local and global tracking and using two different automata, one being a 
  reachability automaton used for detecting infinite local unfolding of unguarded least
  fixpoint formulae and the other being a co-B\"uchi automaton used for
  detecting infinite deferral of guarded least fixpoints; this method has been described
  in~\cite{EmersonJutla99}. We refrain from introducing this more involved setup here,
  so for the moment, our results do not immediately allow the use of co-B\"uchi methods
  for unguarded alternation-free formulae.
\end{rem}

\begin{rem}
  \emph{History-deterministic} automata~\cite{HenzingerPiterman06}
  allow a limited amount of non-determinism but still can easily be
  complemented and combined with a game arena to obtain a game (and
  hence have been introduced under the name \emph{good-for-games
    automata}).  They allow nondeterministic transitions under the
  condition that all nondeterministic choices in an accepting run can
  be resolved by looking only at the history of the run so
  far. Intuitively, the non-determinism in history-deterministic
  automata cannot make guesses about the future of runs. It follows
  from the results of Henzinger and
  Piterman~\cite{HenzingerPiterman06} that instead of full
  determinization of $\mathsf{A}_\target$, it suffices to turn
  $\mathsf{A}_\target$ into an equivalent history-deterministic
  automaton, which then can be complemented and used instead of
  $\mathsf{B}_\target$ in the subsequent development.  For general
  formulae, $\mathsf{A}_\target$ is a parity automaton and can be
  history-determinized by first transforming to a B\"uchi automaton
  and then using the method for history-determinization of B\"uchi
  automata described in~\cite{HenzingerPiterman06}. This method is
  conceptually simpler than full determinization of B\"uchi automata
  by the Safra-Piterman construction and avoids constructing
  Safra-trees, even though it does not reduce the number of states in
  the obtained automaton in comparison to full determinization.
\end{rem}

\begin{rem}[Tracking automata for unguarded
  formulae]\label{rem:unguarded}
  A tableau-based method for deciding satisfiability of unguarded
  formulae of the relational $\mu$-calculus has been introduced by
  Friedmann and Lange~\cite{FriedmannLange13a}. The method augments
  states in the tracking automaton with an additional bit indicating
  \emph{activity} of formulae (meaning that the respective formula has
  been manipulated by the last transition of the tracking automaton),
  doubling the size of the tracking automaton. The acceptance
  condition of the tracking automaton then is modified in order to
  accept only such branches that contain some trace that is both
  infinitely often active and on which some least fixpoint is unfolded
  infinitely often without being dominated. As the propositional
  tableau rules in~\cite{FriedmannLange13a} manipulate one
  propositional formula at a time, one also needs to introduce an
  additional tableau rule in order to ensure fairness of unfolding of
  unguarded fixpoint formulae.

  In the current work we define propositional transitions of the
  tracking automaton $A_\target$ in such a way that all propositional
  formulae are processed whenever a single propositional letter
  $\tau\in\mathsf{choices}$ is read. Hence fairness of fixpoint
  unfolding is inherent to our method. Moreover, the only inactive
  transitions (i.e.\ transitions that do not manipulate the tracked
  formula) in $A_\target$ are transitions of the shape
  $(\hearts\psi,\tau,\hearts\psi)$ for some modal literal
  $\hearts\psi$ and $\tau\in\mathsf{selections}$.  Since $\alpha(\hearts\psi)=1$, all accepting
  runs of $A_\target$ are active by construction. Thus our method
  readily handles unguarded formulae without requiring an activity
  bit.
  
\end{rem}

\paragraph{Model Checking Games}

It will be technically convenient to use a game characterization of
$\mu$-calculus semantics.  Recall that parity games are
infinite-duration two-player games, played by the \emph{existential}
and the \emph{universal} player (also referred to as players $\exists$ and $\forall$). 
A parity game is given by a tuple
\begin{equation*}
\game=(V_\forall,V_\exists,E,v_0,\Omega)
\end{equation*}
where $V=V_\forall\cup V_\exists$ is a set of \emph{positions},
partitioned disjointly into the set $V_\exists$ of positions owned by
the existential player and the set $V_\forall=V\setminus V_\exists$ of
positions owned by the universal player; $E\subseteq V\times V$ is the
set of \emph{moves}; $v_0$ is an \emph{initial position} and
$\Omega:V\to\mathbb{N}$ is the \emph{priority function}, which assigns
priorities $\Omega(v)\in\mathbb{N}$ to states $v\in V$. We write
$E(v)=\{v'\in V\mid (v,v')\in E\}$ for $v\in V$. A \emph{path} through
$(V,E)$ is a finite or infinite sequence $v_0,v_1,\ldots$ such that
$v_{i}\in E(v_{i-1})$ for all $i>0$. A \emph{play} is a \emph{maximal}
path $\pi=v_0,v_1,\ldots$ through $(V,E)$, i.e.~$\pi$ is either
infinite or ends in a position~$v$ such that $E(v)=\emptyset$. A
play~$\pi$ is \emph{winning} for the existential player if it is
finite and ends with the universal player being stuck, or if it is
infinite and the highest priority that is visited infinitely often
on~$\pi$ is even; otherwise,~$\pi$ is winning for the universal
player. Formally,~$\pi$ is \emph{winning} for the existential player
if either $\pi=v_0,v_1,\ldots, v_n$ is finite and $v_n\in V_\forall$,
or~$\pi$ is infinite and $\max(\Inf(\Omega\circ\pi))$ is even.  A
(history-dependent) \emph{strategy} for the existential player is a
partial function $s:V^*V_\exists\rightharpoonup V$ such that
$(v_n,s(v_0,v_1,\ldots,v_n))\in E$ for all partial plays
$v_0,v_1,\ldots,v_n\in V^*V_\exists$ such that $E(v_n)\neq \emptyset$;
for positions $v_n\in V_\exists$ such that $E(v_n)= \emptyset$, $s$ is
undefined on all inputs $v_0,v_1,\ldots,v_n\in V^*V_\exists$.  A play
$\pi=v_0,v_1,\ldots$ is \emph{compatible} with (or \emph{follows}) a
strategy~$s$ for the existential player if $v_{n+1}=s(v_0,\ldots,v_n)$
for all~$i$ such that $v_n\in V_\exists$ and~$v_n$ is not the last
position in~$\pi$. A strategy $s$ is \emph{history-free} if
$s(v_0,\ldots,v_{n-1},v_n)$ depends only on~$v_n$; then $s$ is a
partial function from~$V_\exists$ to~$V$. The existential player wins
a position $v\in V$ if there is a strategy~$s$ such that every play
that starts at $v$ and is compatible with~$s$ is won by the
existential player; similar notions of (winning) strategies for the
universal player are defined dually. The game $\game$ is won by the
player that wins $v_0$.

\begin{lem}[History-free determinacy~\cite{Martin75,EmersonJutla91}]\label{lem:pargame-det}
Parity games are history-free determined, that is,
every position is won by (exactly) one of the two players, and then
there is a history-free strategy that wins the position for the respective
player.
\end{lem}

Given a parity game $\game$ with set~$V$ of positions and a set
$W\subseteq V$, we let $\mathsf{win}^\exists_W$ and
$\mathsf{win}^\forall_W$ denote the set of positions for which the
respective player has a winning strategy in $\game$ such that every
play that is compatible with the strategy remains within~$W$.
Positions $v\in W$ such that
$v\notin\mathsf{win}^\exists_W\cup \mathsf{win}^\forall_W$ are
undetermined (w.r.t. $W$); for such~$v$, neither player has a strategy
that wins~$v$ while staying in~$W$.  As parity games are determined,
we always have $V= \mathsf{win}^\exists_V\cup
\mathsf{win}^\forall_V$. We write $\mathsf{win}^\exists$ for
$\mathsf{win}^\exists_V$ and $\mathsf{win}^\forall$ for
$\mathsf{win}^\forall_V$.

Winning regions in parity games are clearly invariant under
bisimulation (e.g.~\cite{DawarGraedel08,CranenEA18}). We need only
invariance under functional bisimulations. Explicitly, given parity
games $\game=(V_\forall,V_\exists,E,\Omega)$,
$\game'=(V'_\forall,V'_\exists,E',\Omega)$ (we elide initial
positions), a \emph{bounded morphism} $f\colon \game\to\game'$ is a
map $f\colon V\to V'$ whose graph is a bisimulation (see also the
definition for Kripke frames in \autoref{sec:prelims}); that is:~$f$
preserves position ownership and priorities, and (i)~whenever
$(v,v')\in E$, then $(f(v),f(v'))\in E'$; and (ii)~whenever
$(f(v),u')\in E'$, then there exists $(v,v')\in E$ such that
$f(v')=u'$. The mentioned invariance then takes the following
shape:
\begin{lem}[Invariance of parity games under bounded
  morphisms]\label{lem:parity-bisim}
  Let $\game=(V_\forall,V_\exists,E,\Omega)$,
  $\game'=(V'_\forall,V'_\exists,E',\Omega)$ be parity games,
  let~$v\in V$, and let $f\colon\game\to\game'$ be a bounded
  morphism. Then~$\exists$ wins position~$v$ in $\game$ iff~$\exists$
  wins $f(v)$ in~$\game'$.
\end{lem}

We next provide a parity game characterization of formula
satisfaction.  This characterization highlights the close relationship
between satisfaction of fixpoints and \mbox{(non-)}acceptance of runs
in the tracking automaton $\mathsf{A}_\target$ (a nondeterministic
parity automaton). Since the satisfiability game that we introduce in
\autoref{section:alg} will be based on the co-determinized
tracking automaton~$\mathsf{B}_\target$, this relationship will be key
in the correctness proof of the satisfiability game, specifically in
the proof of the truth lemma (\autoref{lem:truth}) and in the
soundness proof (\autoref{lem:modeltogame}).

\begin{defi}[Model checking games]\label{def:mc-games}
  Given a coalgebra $(C,\xi)$, the \emph{model checking game}
  $\game_{\target,(C,\xi)}=(V_\exists,V_\forall,E,\Omega)$ for $\target$
  over $(C,\xi)$ is a parity game with sets of positions
  $V=V_\exists\cup V_\forall$ defined by
\begin{align*}
  V_\exists &= C\times \FLtarget_\exists & V_\forall &= (C\times \FLtarget_\forall)\cup (\Pow(C)\times\FLtarget).
\end{align*}
Here, $\FLtarget_\exists$ consists of those formulae in~$\FLtarget$
that are disjunctions, modal literals, fixpoint literals, or~$\bot$, while
$\FLtarget_\forall=\FLtarget\setminus\FLtarget_\exists$ consists of
those formulae in~$\FLtarget$ that are conjunctions
or~$\top$. 
The moves and priorities in the game are
given by the following table (the ownership of positions is already
defined, and mentioned in the table only for readability)
\begin{center}
\begin{tabular}{|c|c| m{25em} |c|}
\hline
position & owner & set of allowed moves & priority\\
\hline
$(x,\psi)$ & $\exists$/$\forall$ & $\{(x,\phi)\in\{x\}\times\Delta(\psi,\tau)\mid \tau\in\mathsf{choices}\}$ & $\alpha(\psi)-1$\\
$(x,\hearts\psi)$ & $\exists$ & $\{(D,\psi)\in\Pow(C)\times\Delta(\hearts\psi,\kappa)\mid \kappa\in\mathsf{selections},\xi(x)\in\sem{\hearts}D\}$ & $0$\\
$(D,\psi)$ & $\forall$ & $\{(x,\psi)\mid x\in D\}$ & $0$\\
\hline
\end{tabular}
\end{center}  \medskip

\noindent In the above table, the moves available to the players have
been formulated in such a way that the mentioned relationship to the
tracking automaton
$\autom_\target=(\FLtarget,\Sigma,\Delta,\initstate,\alpha)$ becomes
clear. In more detail, given an infinite play~$\pi$ in
$\game_{\target,(C,\xi)}$, the sequence of formulae~$\psi$ encountered
on positions of the form $(x,\psi)$ is a run~$\rho$
of~$\autom_\target$ on a word~$w\in\Sigma^\omega$ extracted from~$\pi$
in an obvious manner (specifically, a move from $(x,\psi)$ to
$(x,\phi)$ where $\phi\in\Delta(\psi,\tau)$ adds~$\tau$ to~$w$, and a
move from $(x,\hearts\psi)$ to $(D,\psi)$ adds some
$\kappa\in\mathsf{selections}$ such that $\hearts\psi\in\kappa$
to~$w$). Even though~$w$ is not uniquely determined by~$\pi$, we
nevertheless refer to~$w$ as the word \emph{induced}
by~$\pi$. Then,~$\pi$ is won by~$\exists$ iff~$\rho$ is
\emph{non-accepting}.

The moves from states of the form $(x,\psi)$ are more explicitly
described as follows, depending on the shape of~$\psi$. The
existential player can move from $(x,\psi_1\vee\psi_2)$ to
$(x,\psi_i)$ for any $i\in\{1,2\}$; each such move is witnessed by any
$\tau\in\mathsf{choices}$ such that $\tau(\psi_1\vee\psi_2)=\psi_i$.
The universal player can move from $(x,\psi_1\wedge\psi_2)$ to
$(x,\psi_i)$ for any $i\in\{1,2\}$.  For fixpoint literals
$\psi=\eta X.\,\phi$, the existential player moves from $(x,\psi)$ to
$(x,\phi[\psi/X])$; ownership of the position is purely formal in this
case. For $\hearts\psi\in\FLtarget$, the existential player can move
from $(x,\hearts\psi)$ to $(D,\psi)$ for any set $D\subseteq C$ such
that $\xi(x)\in\sem{\hearts} D$; each such move is witnessed by any
$\kappa\in\mathsf{selections}$ such that $\hearts\psi\in\kappa$, hence
$\Delta(\hearts\psi,\kappa)=\{\psi\}$. The universal player in turn
can challenge satisfaction of $\psi$ at any state $x\in D$ contained
in the set $D$ provided by the existential player by moving from
$(D,\psi)$ to $(x,\psi)$. The definition of~$\Omega$ implies that the
only positions with non-zero priority are those of the form
$(x,\eta X.\,\psi)$, with $\Omega(x,\eta X.\,\psi)$ being even if
$\eta=\nu$, and odd if $\eta=\mu$.

For technical purposes, we introduce a variant of the model checking
game, the \emph{subformula model checking game}
$\sfgame_{\target,(C,\xi)}$ (the technical advantage of subformulae is
that they allow for proofs by structural induction, a principle that
we will employ in the correctness proof). The positions of
$\sfgame_{\target,(C,\xi)}$ have the same shape as those
of~$\game_{\target,(C,\xi)}$ except that subformulae occur in
positions of $\sfgame_{\target,(C,\xi)}$ wherever elements
of~$\FLtarget$ occur in positions of $\game_{\target,(C,\xi)}$. The
ownership, priorities, and outgoing moves of positions are defined
in~$\sfgame_{\target,(C,\xi)}$ in the same way as
in~$\game_{\target,(C,\xi)}$; in particular, for positions of the form
$(x,\psi)$, these data are defined by case distinction on the
outermost connective of~$\psi$ like in~$\game_{\target,(C,\xi)}$,
except for the following provisos: Given a variable~$X$ such that
$\theta(X)=\eta X.\,\psi$, positions of the shape $(x,\eta X.\,\psi)$
or $(x,X)$ receive
priority~$\Omega(x,\theta^*(X))=\alpha(\theta^*(X))-1$ and belong
to~$\exists$, who has only one move, to $(x,\psi)$.
\end{defi}
\begin{rem}[Higher-arity modalities]\label{rem:arity}
  This is the one place in the technical development where a comment
  is in order on how precisely the treatment of higher-arity
  modalities works: A \emph{modal tracker} $(\psi,i)$ consists of a
  formula $\psi=\hearts(\psi_1,\dots,\psi_n)\in\FLtarget$, where
  $\hearts\in\Lambda$ is an $n$-ary modality, and an index
  $i\in\{1,\dots,n\}$ indicating which argument position will be
  tracked; selections are then sets~$\kappa$ of modal trackers. The
  transition relation~$\Delta$ of the tracking
  automaton~$\autom_\target$ is given on such~$\psi,\kappa$ by
  $\Delta(\psi,\kappa)=\{\psi_i\mid (\psi,i)\in\kappa\}$, i.e.\ the
  selection of arguments to be tracked introduces additional
  nondeterminism. In the model checking
  game~$\game_{\target,(C,\xi)}$, $\exists$ can move from $(x,\psi)$
  to any position of the form $((D_1,\psi_1),\dots,(D_n,\psi_n))$
  (again of priority~$0$) such that $D_1,\dots,D_n\subseteq C$ and
  $\xi(x)\in\sem{\hearts}_X(D_1,\dots,D_n)$; that is,~$\exists$ must
  provide a set of states for each argument of~$\hearts$ in~$\psi$.
  From $((D_1,\psi_1),\dots,(D_n,\psi_n))$,~$\forall$ can then move to
  $(y,\psi_i)$ if $y\in D_i$. In the correspondence between plays
  in~$\game_{\target,(C,\xi)}$ and runs of~$\autom_\target$, the two
  subsequent moves from $(x,\psi)$ to $(y,\psi_i)$ then contribute a
  letter~$\kappa$ such that $(\psi,i)\in\kappa$ to the induced
  word~$w$.
\end{rem}
\begin{exa}[Model checking game]
  We revisit the formula
  $\chi=\nu X.\,(a\land\mu Y.\,(X \lor\mondiamond{g}Y))$ from
  \autoref{ex:ntrack} (the translation of the game logic formula
  $\mondiamond{(g^*)^\times}a$), aiming to check its satisfaction over
  the neighbourhood model $(C,\xi)$ shown below (with~$g$ assumed to
  be the only atomic game). States are depicted by circles, and
  neighbourhoods (that is, sets of states) are depicted by rectangles.
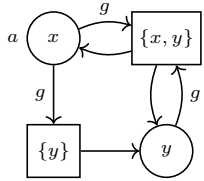
\begin{center}
\begin{tiny}
  \begin{tikzpicture}[
		auto,
    node distance=1.5cm,
    semithick
    ]
     \node[state,label={left: $a$}] (e) {$x$};
     \node[state,rectangle] (f) [right of=e] {$\{x,y\}$};
     \node[state] (g) [below of=f] {$y$};
     \node[state,rectangle] (h) [below of=e] {$\{y\}$};
     \path[->] (e) edge [bend left=20] node [above] {$g$} (f);
     \path[->] (f) edge [bend left=20] node [above] {} (e);
     \path[->] (f) edge [bend right=20] node [left] {} (g);
     \path[->] (g) edge [bend right=20] node [right] {$g$} (f);
     \path[->] (e) edge node [left] {$g$} (h);
     \path[->] (h) edge node [left] {} (g);

  \end{tikzpicture}
\end{tiny}
\end{center}
For instance, at~$x$, Angel can force the game~$g$ into~$y$
(corresponding to~$\{y\}$ being a $g$-neighbourhood at~$x$), while
Angel cannot influence what happens at~$y$ (whose only
$g$-neighbourhood is the whole set~$\{x,y\}$).  Intuitively, the
formula $\chi$ is satisfied at $x$ since Angel can completely avoid
playing $g$ by stopping~$g^*$ immediately and satisfy $a$ at $x$
whenever, in $(g^*)^\times$, Demon decides to play~$g^*$ one more
time.  On the other hand, the formula is not satisfied at $y$ since
that state does not satisfy $a$ and Angel cannot force the game out
of~$y$.  The corresponding model checking game
$\game_{\target,(C,\xi)}$ is as follows, with rounded nodes belonging
to~$\exists$ and rectangles belonging to~$\forall$;
again,~$\phi$ abbreviates the formula
$\mu Y.\,(\chi \lor\mondiamond{g}Y)$.
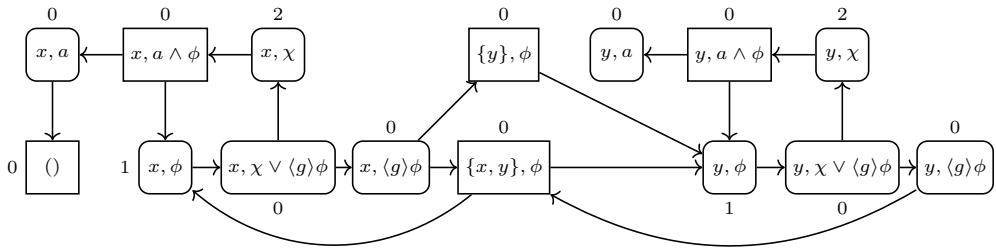
\begin{center}
\begin{tiny}
  \begin{tikzpicture}[
		auto,
    node distance=1.5cm,
    semithick
    ]
     \node[state,rectangle, rounded corners,label={above: $2$}] (a) {$x,\chi$};
     \node[state,rectangle,label={above: $0$}] (b) [left of=a] {$x,a\land\phi$};
     \node[state,rectangle, rounded corners,label={left: $1$}] (c) [below of=b] {$x,\phi$};
     \node[state,rectangle, rounded corners,label={above: $0$}] (d) [left of=b] {$x,a$};
     \node[state,rectangle, label={left: $0$}] (da) [left of=c] {$()$};
     \node[state,rectangle, rounded corners,label={below: $0$}] (e)  [below of=a] {$x,\chi\lor\mondiamond{g}\phi$};
     \node[state,rectangle, rounded corners,label={above: $0$}] (h) [right of=e] {$x,\mondiamond{g}\phi$};
     \node[state,rectangle,label={above: $0$}] (g) [right of=h] {$\{x,y\},\phi$};
     \node[state,rectangle,label={above: $0$}] (f) [above of=g] {$\{y\},\phi$};
     \node (yo1) [right of=f] {};
     \node[state,rectangle,label={above: $0$}] (m) [right of=yo1] {$y,a\land\phi$};
     \node[state,rectangle, rounded corners,label={above: $0$}] (n) [left of=m] {$y,a$};
     \node[state,rectangle, rounded corners,label={above: $2$}] (j) [right of=m] {$y,\chi$};
     \node[state,rectangle, rounded corners,label={below: $1$}] (l) [below of=m] {$y,\phi$};

     \node[state,rectangle, rounded corners,label={below: $0$}] (k) [right of=l] {$y,\chi\lor\mondiamond{g}\phi$};
     \node[state,rectangle, rounded corners,label={above: $0$}] (i) [right of=k] {$y,\mondiamond{g}\phi$};

     \path[->] (a) edge node [above] {} (b);
     \path[->] (b) edge node [left] {} (c);
     \path[->] (b) edge node [left] {} (d);
     \path[->] (d) edge node [left] {} (da);
     \path[->] (c) edge node [left] {} (e);
     \path[->] (e) edge node [left] {} (a);
     \path[->] (e) edge node [left] {} (h);
     \path[->] (h) edge node [left] {} (f);
     \path[->] (h) edge node [left] {} (g);
     \path[->] (g) edge [bend left=40] node [left] {} (c);
     \path[->] (f) edge node [left] {} (l);
     \path[->] (g) edge node [left] {} (l);
     \path[->] (l) edge node [left] {} (k);
     \path[->] (k) edge node [left] {} (j);
     \path[->] (k) edge node [left] {} (i);
     \path[->] (j) edge node [left] {} (m);
     \path[->] (m) edge node [left] {} (n);
     \path[->] (m) edge node [left] {} (l);
     \path[->] (i) edge [bend left=30] node [left] {} (g);

  \end{tikzpicture}
\end{tiny}
\end{center}
The game essentially consists of two copies of the automaton
$\autom_\target$ from \autoref{ex:ntrack}; priorities in the game are
also inherited from $\autom_\target$.  The two copies of the automaton
are linked by the positions $(\{y\},\phi)$ and $(\{x,y\},\phi)$ that
encode the evaluation of the modality $\mondiamond{g}\phi$.  By
definition of the model, we have $\xi(x)(g)=\{\{x,y\},\{y\}\}$ and
$\xi(y)(g)=\{\{x,y\}\}$. Recalling the predicate lifting
$\sem{\mondiamond{g}}$ for the monotone diamond from
\autoref{ex:logics}.\ref{item:monotone-mu}, we thus have
$\xi(x)\in \sem{\mondiamond{g}} (D)$ for both $D=\{x,y\}$ and
$D=\{y\}$ but $\xi(y)\in \sem{\mondiamond{g}} (D)$ only for
$D=\{x,y\}$; the moves to the central positions $(\{x,y\},\phi)$ and
$(\{y\},\phi)$ are induced accordingly. The model checking game treats
the propositional atom~$a$ as a nullary modality as per
\autoref{rem:arity}: From positions of shape $(z,a)$, player
$\exists$ can move to~$()$ (a $0$-tuple of set/formula pairs),
provided that~$a$ is satisfied at~$z$. If this holds, then~$\exists$
wins because~$\forall$ has no moves at~$()$; otherwise,~$\exists$
loses, being stuck at $(z,a)$. In the example,~$\exists$
correspondingly wins~$(x,a)$ but loses $(y,a)$.

Player $\exists$ wins the left-most five positions in this game by the strategy that
always moves from position $(x,\chi\lor\mondiamond{g}\phi)$ to position
$(x,\chi)$; this enforces that plays of the game that start in 
one of these positions either get stuck in the position $(\emptyset,a)$ (which belongs to player $\forall$ but has no outgoing transitions and hence is won by player $\exists$),
or forever follow the cycle through position $(x,\chi)$ and thereby visit priority
$2$ infinitely often.
By the correctness of model checking games, this shows that $x$ satisfies all formulae from
$\FLtarget$ except $\mondiamond{g}\phi$.
All other positions in the game are won by player $\forall$ using the strategy
that always moves from $(\{x,y\},\phi)$ to $(y,\phi)$ and from
$(y,a\land\phi)$ to $(y,a)$; with this strategy, player $\forall$ can enforce
that plays either get stuck in position $(y,a)$ (lost by player $\exists$)
or eventually take the bottom-right cycle forever,
seeing priority $1$ infinitely often. This shows that none of the formulae from
$\FLtarget$ are satisfied at $y$.

\end{exa}

\noindent We note that the two versions of the model checking games
are bisimilar, and hence equivalent (\autoref{lem:parity-bisim}):
\begin{lem}\label{lem:sfgame-bounded-mor}
  The assignment $f(x,\psi)=(x,\theta^*(\psi))$,
  $f(D,\psi)=(D,\theta^*(\psi))$ defines a bounded morphism
  $f\colon\sfgame_{\target,(C,\xi)}\to\game_{\target,(C,\xi)}$.
\end{lem}
\begin{proof}
  It is clear that~$f$ preserves priorities and position ownership. We
  check the conditions on moves, restricting attention to the only
  cases where the games differ appreciably, viz., fixpoint literals
  and fixpoint variables. At such positions, however, the outgoing moves
  are uniquely determined in both games, so we just have to show
  that~$f$ preserves these moves.

  So let~$X$ be a fixpoint variable, with $\theta(X)=\eta
  X.\,\psi$. The unique move from both $(x,X)$ and $(x,\eta X.\,\psi)$
  in $\sfgame_{\target,(C,\xi)}$ leads to $(x,\psi)$. Since
  $\theta^*(X)=\theta^*(\eta X.\psi)$, $(x,X)$ and $(x,\eta X.\,\psi)$
  are mapped to the same position under~$f$, and this position has the form
  $(x,\eta X.\psi')$ where~$\psi'$ is obtained from~$\psi$ by
  successively substituting free fixpoint variables~$Y$
  with~$\theta(Y)$ innermost first, but skipping the actual innermost
  variable~$X$ in~$\psi$; so
  $\theta^*(\psi)=\theta^*(\psi[\eta X.\psi/X])=\psi'[\eta
  X.\psi'/X])$.  From $(x,\eta X.\psi')$,~$\exists$'s unique move in
  $\game_{\target,(C,\xi)}$ thus leads to
  $(x,\psi'[\eta X.\psi'/X])=(x,\theta^*(\psi))=f(x,\psi)$, as
  required.
\end{proof}
\noindent The model checking game $\game_{\target,(C,\xi)}$ is very
similar to the one considered by C\^irstea et al.~\cite{CirsteaEA11},
one notable difference being that we do not assume guardedness. On the
other hand, the subformula model checking game
$\sfgame_{\target,(C,\xi)}$ resembles a game that was considered by
Venema~\cite{Venema06} but which, again, assumes guardedness and
moreover works with a version of the coalgebraic $\mu$-calculus based
on the coalgebraic cover modality~\cite{Moss99} instead of on
predicate liftings. Indeed, the proof of correctness given by
C\^irstea et al.\ is largely by reference to Venema's proof, an
argument that is formally justified by our
\autoref{lem:sfgame-bounded-mor}. Due to the mentioned guardedness
issue, we opt to present a full correctness proof of our game, which
largely follows the one given by Venema in that it makes do without
(ordinal) timeouts as frequently used in the literature on the
relational
$\mu$-calculus~\cite{StreettEmerson89,NiwinskiWalukiewicz96,BradfieldWalukiewicz18}.
\begin{thm}[Correctness of the model checking
  game]\label{thm:satisfaction-game}
  Given a coalgebra $(C,\xi)$, a state $x\in C$ and a formula
  $\psi\in\FLtarget$, we have $x\models\psi$ if and only if the
  existential player wins the position $(x,\psi)$ in
  $\game_{\target,(C,\xi)}$.
\end{thm}
\begin{proof}
  We note first that the positions reachable from $(x,\psi)$ in
  $\game_{\target,(C,\xi)}$ are the same as in $\game_{\psi,(C,\xi)}$,
  so we can assume w.l.o.g.\ that~$\psi$ is the target formula, that
  is, $\psi=\chi$. By \autoreflems{lem:parity-bisim}
  and~\ref{lem:sfgame-bounded-mor}, we can then replace
  $\game_{\target,(C,\xi)}$ with $\sfgame_{\target,(C,\xi)}$, since
  $\theta^*(\target)=\target$. We will use structural induction
  on~$\chi$, and hence need to drop, only for purposes of this proof,
  the assumption that~$\target$ is closed. Of course, the game is then
  played over a pair $((C,\xi),i)$ where~$i\colon\Var\pfun \Pow(C)$ is
  a valuation of the fixpoint variables such that $i(X)$ is defined
  for all $X\in\FV(\target)$; we correspondingly write
  $\sfgame_{\target,i}$ for the generalized game, eliding mention
  of~$(C,\xi)$ which remains unchanged throughout. We thus have new
  positions of the shape $(x,X)$ or $(x,\neg X)$, which receive
  priority~$0$ and no outgoing moves. The ownership of these positions
  is defined by letting $(x,X)$ be owned by~$\exists$ if
  $x\notin i(X)$ and by~$\forall$ otherwise, and correspondingly
  letting $(x,\neg X)$ be owned by~$\exists$ if $x\in i(X)$ and
  by~$\forall$ otherwise. Other than this, the game remains unchanged.

  Next, for purposes of economizing one direction of the proof, we
  modify $\sfgame_{\target,i}$ to ensure symmetry between the
  players; we write $\sfsymgame_{\target,i}$ for this
  symmetrized game. First, we reassign positions of the form
  $(x,\nu X.\,\psi)$ or $(x,X)$, with~$X$ a $\nu$-variable, to the
  universal player -- since these positions have precisely one
  outgoing move, this change is clearly immaterial to how the game is
  played. Second, out of every pair $\hearts,\overline\hearts$ of dual
  modalities, we arbitrarily assign one to~$\exists$ and the other
  to~$\forall$, and we write~$\hearts$ for modalities assigned
  to~$\exists$, and~$\overline\hearts$ for modalities assigned
  to~$\forall$. Moreover, we rename the previous positions of the form
  $(D,\psi)$ into $(D,\psi,\forall)$, and introduce additional
  positions $(D,\psi,\exists)$. Positions of the form
  $(x,\hearts\psi)$ still belong to~$\exists$, with moves like before,
  into the renamed positions of the form~$(D,\psi,\forall)$, which
  still belong to~$\forall$.  On the other hand, positions of the form
  $(x,\overline\hearts\psi)$ now belong to~$\forall$, and~$\forall$
  can move to $(D,\psi,\exists)$ if
   $\xi(x)\notin\Sem{\overline{\hearts}}_C(\overline{D})$,
  where~$\overline D$ denotes the complement of~$D$ in~$C$. The new
  positions $(D,\psi,\exists)$ belong to~$\exists$, who can move to
  $(y,\psi)$ such that $y\in D$. The new sequences of moves
  $(x,\overline\hearts\psi)\xrightarrow{\forall}
  (D,\psi,\exists)\xrightarrow{\exists} (y,\psi)$ are, for purposes
  of~$\exists$ winning the game, equivalent to the previous sequences
  $(x,\overline\hearts\psi)\xrightarrow{\exists}
  (D,\psi)\xrightarrow{\forall} (y,\psi)$: In either case, $\exists$
  can force the game into a set~$U$ of positions of the form
  $(y,\psi)$, necessarily of the form $U=D\times\{\psi\}$, iff
  $\xi(x)\in\Sem{\overline\hearts}_C(D)$. To see this for the new
  moves
  $(x,\overline\hearts\psi)\xrightarrow{\forall}
  (D,\psi,\exists)\xrightarrow{\exists} (y,\psi)$, we reason as
  follows: $\exists$ can force the game into~$D\times\{\psi\}$
  iff~$\forall$ cannot move to $(D',\psi,\exists)$ for any subset
  $D'\subseteq\overline D$ iff
  $\xi(x)\in\Sem{\overline\hearts}_C(\overline{D'})$ for all
  $D'\subseteq\overline D$ iff $\xi(x)\in\Sem{\overline\hearts}_C(D)$,
  where the last step uses monotonicity of~$\overline\hearts$. Thus,
  $\sfgame_{\target,i}$ and $\sfsymgame_{\target,i}$
  are equivalent in the sense that positions of the form $(x,\psi)$
  are won by the same player in either game.

  By now, we have reduced the claim of the lemma to showing that in
  $\sfsymgame_{\target,i}$,
  \begin{enumerate}
  \item\label{item:e-wins} if $x\models\target$, then~$\exists$ wins
    $(x,\target)$; and
  \item\label{item:a-wins} if $x\not\models\target$, then~$\forall$ wins
    $(x,\target)$.
  \end{enumerate}
  The symmetry of $\sfsymgame_{\target,i}$ allows us to
  conclude \eqref{item:a-wins} from \eqref{item:e-wins}, as follows:
  If $x\not\models\target$, then $x\models\neg\target$
  (with~$\neg\target$ defined by taking negation normal forms as
  indicated in
  \autoref{sec:prelims}). By~\eqref{item:e-wins},~$\exists$ wins
  the position $(x,\neg\target)$ in the model checking game for
  $\neg\target$. This game is now dual to the model checking game
  for~$\target$ in the sense that one is obtained from the other by
  swapping the positions of the players and dualizing the priorities
  (i.e.\ swapping priorities $2n$ and $2n-1$); of course, positions
  $(x,\neg\psi)$ in the game for~$\neg\target$ correspond to positions
  $(x,\psi)$ in the game for~$\target$. We thus immediately obtain
  that~$\forall$ wins $(x,\target)$.
  
  It remains to prove~\eqref{item:e-wins}.  We switch back to using
  the simpler game $\sfgame_{\target,i}$. As indicated above, we
  proceed by induction on~$\target$; we treat~$i$ as universally
  quantified in the inductive claim. The cases for free fixpoint
  variables and Boolean operators ($\land$, $\lor$, $\top$, $\bot$)
  are trivial; we illustrate this on the case where
  $\target=\target_1\land\target_2$: If
  $x\in\Sem{\target_1\land\target_2}_{i}$, then
  $x\in\Sem{\target_1}_{i}$ and $x\in\Sem{\target_2}_{i}$. By
  induction,~$\exists$ wins $(x,\target_j)$ in $\sfgame_{\target_j,i}$
  for $j=1,2$. In $\sfgame_{\target_1\land\target_2,i}$, $\forall$ can
  move from $(x,\target)$ to either $(x,\target_1)$ or
  $(x,\target_2)$, so~$\exists$ wins by playing like in
  $\sfgame_{\target_1,,i}$ or $\sfgame_{\target_2,i}$,
  respectively. The remaining cases are as follows.
  \begin{itemize}
  \item $\target=\hearts\target_1$: By induction, $\exists$ wins by
    playing $(D,\target_1)$ where $D=\Sem{\target_1}_i$.
  \item $\target=\mu X.\,\target_1$: It suffices to show that the set
    \begin{equation*}
      W_1:=\{x\in C\mid \text{$\exists$ wins $(x,\target)$ in $\sfgame_{\target,i}$}\}=\{x\in C\mid \text{$\exists$ wins $(x,\target_1)$ in $\sfgame_{\target,i}$}\}
    \end{equation*}
    (where the second equality is immediate from the game rules) is a
    prefixpoint of the function defining
    $\Sem{\target}_i=\Sem{\mu X.\,\target_1}_i$ as a least
    prefixpoint, i.e.\ that
    \begin{equation}\label{eq:prefixpoint}
      \Sem{\target_1}_{i'}\subseteq W_1\quad\text{where $i'=i[X\mapsto W_1]$}.
    \end{equation}
    So let $x\in\Sem{\target_1}_{i'}$. By induction,~$\exists$ has a
    strategy~$s'$ in $\sfgame_{\target_1,i'}$ that wins
    $(x,\target_1)$. We have to show that~$\exists$ wins
    $(x,\target_1)$ in $\sfgame_{\target,i}$, which differs
    from~$\sfgame_{\target_1,i'}$ only in that it treats the fixpoint
    variable~$X$ as bound. The winning strategy works as follows:
  \begin{itemize}
  \item From~$(x,\target_1)$, play according to~$s'$ until a position
    of the form $(y,X)$ is encountered (if ever). This is possible
    because up to that point, there is no difference
    between~$\sfgame_{\target,i}$ and~$\sfgame_{\target_1,i'}$. If no
    position $(y,X)$ is ever reached, then the play effectively takes
    place in~$\sfgame_{\target_1,i'}$, and as such follows~$s'$. It is
    thus won by~$\exists$, since~$s'$ is winning.
  \item If a position of the form $(y,X)$ is reached, then
    $y\in i'(X)=W_1$, since~$s'$ is winning
    in~$\sfgame_{\target_1,i'}$. The next position reached is
    $(y,\target_1)$, so~$\exists$ wins in~$\sfgame_{\target_1,i}$
    because~$y\in W_1$.
  \end{itemize}
\item $\target=\nu X.\,\target_1$: Let~$x\in\Sem{\target}_i$.  We
  construct a strategy~$s$ that wins $(x,\target)$ in
  $\sfgame_{\target,i}$ as follows. From $(x,\target)$, the game
  proceeds to $(x,\target_1)$. By fixpoint unfolding, we have
  \begin{equation}\label{eq:game-unfolding}
    \Sem{\target}_i=\sem{\target_1}_{i'}\quad\text{where $i'=i[x\mapsto\Sem{\target}_i]$}.
  \end{equation}
  By induction, $\exists$ thus has a strategy~$s'$ that wins
  $(x,\target_1)$ in~$\sfgame_{\target_1,i'}$. We let~$s$ follow~$s'$
  in~$\sfgame_{\target,i}$ until a position of the form $(y,X)$ is
  reached (exploiting like in the previous case that up to that point,
  the games~$\sfgame_{\target_1,i'}$ and~$\sfgame_{\target,i}$ do not
  differ). Since~$s'$ is winning in~$\sfgame_{\target_1,i'}$, we then
  have $y\in i'(X)=\Sem{\target}_i=\Sem{\target_1}_{i'}$, so we can
  continue in the same manner after the game~$\sfgame_{\target,i}$
  automatically proceeds to $(y,\target_1)$.  To see that~$s$ is
  winning, we distinguish cases on a play~$\pi$ that follows~$s$:
  \begin{itemize}
  \item If from some point on,~$\pi$ no longer reaches positions of
    the form $(y,X)$, then~$\pi$ has a suffix that is a winning play
    for~$\exists$ from a position of the form $(y,\target_1)$
    in~$\sfgame_{\target_1,i'}$, so~$\exists$ wins~$\pi$.
  \item Otherwise,~$\pi$ infinitely often visits positions of the form
    $(y,X)$. Thus,~$X$ is unfolded infinitely often. Intuitively
    speaking, since~$X$ is a $\nu$-variable and $\nu X.\target_1$ is
    the outermost fixpoint in the target formula (being the target
    formula itself),~$\exists$ should therefore win the play according
    to the intention of the game, as long as the mechanism that
    replaces the direct comparison of inner vs.\ outer fixpoints (as
    used in the winning condition of the game considered by
    Venema~\cite{Venema06}) with the comparison of alternation depth
    works. Formally, we proceed as follows. We have to show that
    $\ad(Z)<\ad(X)$ for every $\mu$-variable~$Z$ that is unfolded
    on~$\pi$ between two unfoldings of~$X$ (this implies
    $\alpha(\theta^*(Z))<\alpha(\theta^*(X))$, so positions where~$Z$
    is unfolded have lower priority than positions where~$X$ is
    unfolded). Let $Y_n,\dots,Y_1,Y_0=X$ be the sequence of variables
    unfolded between two unfoldings of~$X$, including the second (but
    not the first) unfolding of~$X$ itself. We show by induction on
    $j\in\{0,\dots,n\}$ that for every~$j$, there is a dependency
    chain, possibly of length~$0$, from~$Y_j$ to~$X$ (implying
    that~$\ad(X)\ge\ad(Y_j)$, and that $\ad(X)>\ad(Y_j)$ if~$Y_j$ is a
    $\mu$-variable). The induction base $j=0$ is trivial. In the
    induction step for $j>0$, we can assume that $Y_j\neq Y_k$ for
    $j>k$, since otherwise we are done by induction. The unfolding
    step from $Y_j$ leads to a position of the form $(y,\psi)$ where
    $\theta(Y_j)=\eta X.\psi$. Since a position $(z,Y_{j-1})$ is
    reached from $(y,\psi)$ without interceding unfolding
    steps,~$Y_{j-1}$ is a subformula of $\psi$. If
    $Y_{j-1}\in\FV(\psi)$, then $Y_{j-1}\in\FV(\theta(Y_j))$ since
    $Y_{j-1}\neq Y_j$; that is, $Y_j\depord Y_{j-1}$, and we are done
    by induction. Otherwise, $\theta(Y_{j-1})$ is a subformula
    of~$\psi$. By induction, there is a dependency chain
    $Y_{j-1}\depord Z\depord\dots\depord X$; in particular,
    $Z\in\FV(\theta(Y_{j-1}))$. If $Z=Y_j$, then we are
    done. Otherwise, $Z\in\FV(\theta(Y_j))$, i.e.\ $Y_j\depord Z$, and
    we are done.\qedhere
  \end{itemize}
\end{itemize}

\end{proof}

\begin{rem}[Fixpoint games]\label{rem:fp-games}
  By instantiation of the model checking game to a generalization of
  the monotone $\mu$-calculus
  (\autoref{ex:logics}.\ref{item:monotone-mu}), we obtain a
  general form of fixpoint games for monotone functions on powerset
  lattices, which in turn are an instance of fixpoint games over
  continuous lattices~\cite{BaldanKMP19}. Details are as follows.
  
  We need only the case without propositional atoms, whose mention we
  therefore elide in the following, and with only one atomic program
  that we keep implicit. On the other hand, we generalize to
  higher-arity neighbourhood frames and modalities: For $n\ge 0$, we
  define the \emph{$n$-ary monotone neighbourhood functor}~$\Mon_n$
  (e.g.~\cite{SchroderPattinson10,MartiVenema15}) as taking a set~$X$
  to the set of subsets of~$(\contrapow X)^n$ that are upwards closed
  under componentwise subset inclusion. We use an $n$-ary
  modality~$\Diamond$, which we interpret over~$\Mon_n$ by the
  predicate lifting given by
  \begin{equation*}
    \Sem{\Diamond}_X(A_1,\dots,A_n)=\{\alpha\in\Mon_n X\mid (A_1,\dots,A_n)\in\alpha\}.
  \end{equation*}
  By transposition of arguments, a coalgebra
  $C\to\Mon_n C\subseteq\contrapow((\contrapow C)^n)$ can
  alternatively be seen as a map
  $g\colon(\contrapow C)^n\to\contrapow C$ that is monotone w.r.t.\
  (componentwise) subset inclusion. Recall here that $\contrapow$ is
  the contravariant powerset functor; as we do not actually need the
  action on maps in the following, we will just write $\Pow(C)$ in
  place of $\contrapow C$. The semantics of a formula
  $\Diamond(\phi_1,\dots,\phi_n)$ in~$C$ under a valuation~$i$ is then
  equivalently given by
  $\Sem{\Diamond\phi}_i=g(\Sem{\phi_1}_i,\dots,\Sem{\phi_1}_i)$. That
  is, we can just see the monotone $\mu$-calculus as an expression
  language for nested fixpoints over (higher-arity) monotone functions
  on~$\Pow(C)$. In the corresponding instance of the model checking
  game on~$C$,~$\exists$ can move from a position
  $(x,\Diamond(\psi_1,\dots,\psi_n))$ to any tuple
  $((D_1,\psi_1),\dots,(D_n,\psi_n))$ such that
  $x\in g(D_1,\dots,D_n)$. We use these games in the fixpoint
  characterization of the satisfiability game
  (\autoref{lem:sat-game-fp}). 
\end{rem}

\section{One-Step Satisfiability and Tableaux}\label{sec:tableaux}

In this section, we identify an embodiment of a model for $\target$ in
the shape of a subautomaton of the co-determinized tracking
automaton~$\mathsf{B}_\target$ that satisfies certain additional
properties; we will use this concept as a stepping stone in the
reduction of satisfiability checking to game solving, and, as usual,
call such a witness for formula satisfaction a \emph{tableau}.
Specifically, such a subautomaton consists of those automaton nodes
$q$ for which there are states in the model that jointly satisfy all
formulae from $l(q)$, and the automaton transitions in a tableau are
required to witness satisfaction of those formulae; we formalize the
structural property required for the satisfaction of modalities using
the concept of \emph{one-step satisfiability}. Then we show that every
tableau carries a coalgebra structure that is \emph{coherent} with its
transitional structure and its labels; such coalgebras then satisfy a
truth lemma implying satisfaction of the target formula.  The proof of
the truth lemma relies on the model checking game, and exploits that
the latter relates closely to the nondeterministic tracking automaton
$\autom_\target$.

We begin with considerations on the above-mentioned problem of
\emph{one-step satisfiability checking}, a functor-specific problem
that in many instances can be solved in time singly exponential in
$\size(\target)$.

\begin{defi}[One-step satisfiability
  problem~\cite{Schroder07,SchroderPattinson08,MyersEA09}]\label{def:oss}
  Let $V$ be a finite set of \emph{propositional variables}.  A
  \emph{one-step pair} $(\osOne,\osZero)$ (over~$V$) consists of a set
  $\osZero\subseteq\Pow(V)$ (understood as a disjunctive normal form
  over~$V$, cf.\ \autoref{rem:os-logic}) and a set
  $\osOne\subseteq \Lambda(V)$, understood conjunctively, where we
  require that~$\osOne$ mentions every element of~$V$ precisely
  once. Correspondingly, we interpret $a\in V$ and~$\osOne$
  over~$\osZero$ by
  \begin{align*}
    \sem{a}^\osZero_0 & = \{u\in \osZero\mid a\in u\}\:  \subseteq\Theta\\
    \sem{\osOne}^\osZero_1&
                            =\textstyle\bigcap_{\hearts
                            a\in
                            \osOne}\sem{\hearts}_\osZero\sem{a}^\osZero_0  \subseteq F\Theta.
  \end{align*}
  We say that $(\osOne,\osZero)$ is \emph{satisfiable} (over the
  functor~$F$) if $\sem{\osOne}^\osZero_1\neq\emptyset$. The
  \emph{strict one-step satisfiability problem} is to decide whether a
  given one-step pair $(\osOne,\osZero)$ is satisfiable; here, the
  qualification `strict' refers to the measure of input size of the
  problem, which we take to be
  \begin{equation*}
    \size(\osOne):=\sum_{\hearts a\in \osOne}(1+\size(\hearts))
  \end{equation*}
  (in particular,~$|\osZero|$, which may be exponential
  in~$\size(\osOne)$, does not count towards the input size).
\end{defi}
  
\begin{rem}[One-step logic]\label{rem:os-logic}
  We keep the definition of the actual one-step logic as mentioned in
  the introduction somewhat implicit in the above definition of the
  one-step satisfiability problem. In a more explicitly syntactic
  view, one will regard~$V$ as a set of propositional variables. One
  then sees that a one-step pair $(\gamma,\Theta)$ as above contains
  two layers: a purely propositional layer embodied in~$\osZero$,
  which postulates which propositional formulae over~$V$ are
  satisfiable (that is, we see $\osZero\subseteq\Pow (V)$ as a
  disjunctive normal form
  $\Lor_{u\in \osZero}(\Land_{a\in u}a\land\Land_{a\in V\setminus
    u}\neg a)$); and a modal layer with nesting depth of modalities
  uniformly equal to~$1$, embodied in the set~$\osOne$ of modal
  literals, which specifies a constraint
  $\Land_{\hearts a\in \osOne}\hearts a$ on an element of
  $F\osZero$. Under this perspective (and otherwise), it is trivial to
  note that satisfiability of a one-step pair $(\gamma,\Theta)$ is
  preserved under enlarging~$\Theta$, as this corresponds to weakening
  the propositional formula represented by~$\Theta$.
\end{rem}

\begin{exa}[One-step satifiability]\label{ex:onestepsat}
  We consider the one-step satisfiability problem for the logics from 
  \autoref{ex:logics}, omitting details on the (trivial) treatment
  of propositional atoms.
  \begin{enumerate}[wide]
  \item\label{item:oss-rel} For the \emph{relational modal
      $\mu$-calculus} (\autoref{ex:logics}.1.), where
    $\Lambda=\{\Diamond,\Box\}$, the one-step satisfiability problem
    is to decide, for a given one-step pair $(\osOne,\osZero)$
    over~$V$, whether there is $A\in \sem{\osOne}^\osZero_1$, that is,
    a subset $A\in\Pow \osZero$ such that for each
    $\Diamond a\in \osOne$, there is $u\in A$ such that $a\in u$, and
    for each $\Box b\in \osOne$ and each $u\in A$, $b\in u$.
    Equivalently, one needs to check that for
    each~$\Diamond a\in \osOne$ there is $u\in \osZero$ such that
    $a\in u$ and moreover $b\in u$ for all $\Box b\in \osOne$.  To
    avoid quadratic complexity in~$\size(\gamma)$, implement this
    check in two passes: In the first pass, go through all
    $\Box b\in\gamma$ and remove from~$\Theta$ all~$u$ such that
    $b\notin u$; in the second pass, go through all
    $\Diamond a\in\gamma$ and check that there remains some~$u$
    in~$\Theta$ such that $a\in u$. Both passes can be done in time
    $\lO(\size(\osOne)\cdot|V|\cdot
    |\osZero|)=\lO(\size(\osOne)\cdot|V|\cdot 2^{|V|})$, showing that
    in this case the strict one-step satisfiability problem is in
    \ExpTime. We note that this is all the work that will be required
    to instantiate our generic complexity bound
    (Theorem~\ref{thm:exptime} below) to the relational
    $\mu$-calculus, obtaining the known upper bound \ExpTime for
    satisfiability checking~\cite{EmersonJutla99}.

  \item\label{item:oss-mon} For the \emph{monotone modal
      $\mu$-calculus} (\autoref{ex:logics}.2.) with set $\AtGames$ of
    atomic games, we have
    $\Lambda=\{\mondiamond{g},\monbox{g}\mid g\in \AtGames\}$, again
    eliding propositional atoms for the sake of readability.  It is an
    immediate property of the semantics of monotone modalities that in
    order to check that $\sem{\osOne}^\osZero_1\neq\emptyset$ for a
    given one-step pair $(\osOne,\osZero)$ over~$V$, it suffices to
    check that whenever $\mondiamond{g}a,\monbox{g}b\in\osOne$, then
    there is $u\in\osZero$ such that
    $a,b\in u$~\cite[Proposition~3.8]{Vardi89}. (Indeed, this
    criterion corresponds to the usual monotonicity rule -- cf.\
    \cite{SchroderPattinson09b,CirsteaEA11a} -- under the
    correspondence between modal tableau rules and one-step
    satisfiability checking discussed in \autoref{rem:rules} below.)
    This can clearly be done in time
    $\lO(\size(\osOne)^2\cdot|V|\cdot
    |\osZero|)=\lO(\size(\osOne)^2\cdot|V|\cdot 2^{|V|})$, showing
    that the strict one-step satisfiability problem for the monotone
    case is in \ExpTime. We note that again this is all the work that
    is required to instantiate our generic complexity bound and obtain
    the known upper \ExpTime bounds on satisfiability checking for
    game logic~\cite{PaulyParikh03} and the monotone
    $\mu$-calculus~\cite{CirsteaEA11a}; in fact, it appears that for
    the latter case, the result for the full unguarded logic is
    formally new (however, we note that it could alternatively be
    obtained by encoding the monotone $\mu$-calculus into the
    relational $\mu$-calculus in the same way as for game
    logic~\cite{Parikh83,PaulyParikh03}).

  \item For the \emph{graded $\mu$-calculus}
    (\autoref{ex:logics}.\ref{item:graded}.), the one-step
    satisfiability problem is to decide, for a one-step
    pair~$(\osOne,\osZero)$, whether there is a multiset
    $\beta\in\Bag \osZero$ such that
    $\sum_{u\in \osZero\mid a\in u}\beta(u)>m$ for each
    $\langle m\rangle a\in \osOne$ and
    $\sum_{u\in \osZero\mid a\notin u}\beta(u)\leq m$ for each
    $[m]a\in \osOne$. The easiest way to see that the strict one-step
    satisfiability problem is in \ExpTime is via a nondeterministic
    polynomial-space algorithm that goes through all $u\in \osZero$,
    guessing multiplicities $\beta(u)\in\{0,\dots,m+1\}$ where~$m$ is
    the greatest index of any diamond modality $\langle m\rangle$ that
    occurs in~$\osOne$. This multiplicity is used to update~$|V|$
    counters that keep track of the total
    measure~$\beta(\Sem{a}^\osZero_0)$ for $a\in V$; after updating
    the counters, the multiplicity is discarded, so that only
    polynomial space is used (for the counters). Once all
    multiplicities have been guessed, the algorithm verifies that
    $\beta\in\Sem{\osOne}^\osZero_1$, using only the final counter
    values~\cite[Lemma~1]{KupfermanEA02}. Essentially the same method
    works also for the \emph{graded $\mu$-calculus with polynomial
      inequalities} (\autoref{ex:logics}.\ref{item:graded-poly})
    (and similar ideas have been used in work on Presburger modal
    logic~\cite{DemriLugiez06}): In this setting, it still suffices to
    guess multiplicities up to~$b+1$ where $b$ is the largest index of
    any diamond modality~$\langle p\rangle$ occurring in~$\osOne$
    (indeed, if anything the bounds become smaller; e.g.~for
    $\langle X_1^2+X_2^2-b\rangle$, it suffices to explore
    multiplicities up to $\lceil \sqrt{b}\rceil+1$). Note that this
    argument does rely on the assumption that all coefficients of
    non-constant monomials are non-negative.
  \item\label{item:oss-prob} For the \emph{probabilistic
      $\mu$-calculus with polynomial inequalities}
    (\autoref{ex:logics}.\ref{item:prob-poly}), we first observe,
    following~\cite{GutierrezBasultoEA17,KupkeEA22}, that a small
    model property holds for the one-step logic: If a one-step pair
    $(\osOne,\osZero)$ over~$V$ is satisfiable, then
    $\Sem{\osOne}^\osZero_1$ contains an element
    $(d,Q)\in\Dist \osZero\times\Pow(\PropAts)$ such that $d(u)>0$ for only
    $|V|$-many $u\in \osZero$. This is seen as follows: Suppose that
    $(d_0,Q)\in\Sem{\osOne}^\osZero_1$. Then for any
    $d\in\Dist \osZero$, to have $(d,Q)\in\Sem{\osOne}^\osZero_1$ it
    suffices that
     $d(\Sem{a}^\osZero_0)=d_0(\Sem{a}^\osZero_0)$
    for each $a\in V$. Since
    $d(\Sem{a}^\osZero_0)=\sum_{u\in\Sem{a}^\osZero_0}d(u)$, this means that the
    numbers $y_u=d(u)$, for $u\in \osZero$, form a non-negative solution to
    the system of linear equations
    \begin{equation*}
      \sum_{u\in\Sem{a}^\osZero_0}y_u=d_0(\sem{a}^u_0) \qquad\text{for $a\in V$}.
    \end{equation*}
    Since this system has a solution induced by~$d_0$ itself, we
    obtain by the Carath\'eodory theorem (e.g.~\cite{schrijver86})
    that there is a solution with at most~$|V|$ non-zero
    components. This allows us to solve the strict one-step
    satisfiability problem in non-deterministic polynomial space, and
    hence in exponential time, as follows: Guess $|V|$ elements
    $u\in \osZero$ that receive positive weight~$d(u)$ in a
    solution~$d$, and then check for satisfiability of the constraint
    on these~$|V|$ real numbers that is embodied in~$\osOne$. This
    constraint is a polynomially-sized system of polynomial
    inequalities, whose satisfiability can, by results of
    Canny~\cite{Canny88}, be checked in polynomial space.
  \end{enumerate}
\end{exa}
\begin{rem}[One-step polysize model property]\label{rem:ospmp}
  We say that the logic has the \emph{one-step polysize model property
    (OSPMP)} if there is a polynomial~$p$ such that whenever a
  one-step pair $(\osOne,\osZero)$ over~$V$ is satisfiable, then
  $\sem{\osOne}^\osZero_1$ has an element of the form $Fi(t)$ where
  $i\colon \osZero_0\to\osZero$ is the inclusion of a
  subset~$\osZero_0\subseteq\osZero$ such that
  $|\osZero_0|\le p(|V|)$~\cite{KupkeEA22}. For instance, the
  arguments in \autorefexas{ex:onestepsat}.\ref{item:oss-rel}
  and~\ref{ex:onestepsat}.\ref{item:oss-prob} show that the relational
  $\mu$-calculus and the probabilistic $\mu$-calculus (even with
  polynomial inequalities) both have the OSPMP. Similar arguments as
  for the probabilistic case show that the graded $\mu$-calculus (even
  with polynomial inequalities) also has the OSPMP, using the integer
  Carath\'eodory theorem~\cite{EisenbrandShmonin06};
  cf.~\cite{KupkeEA22} for details. Our model construction below 
  will establish that the OSPMP implies a polynomially branching model property
  (\autoref{rem:poly-branching}).
\end{rem}

Next, we present our notion of tableaux, which are partial subautomata
of the co-determinized tracking automaton
$\mathsf{B}_\target=(\detcarrier, \Sigma,\delta,\initstate,\detprio)$.
We first fix some notation: To $\Xi\subseteq\mathsf{selections}$ and a
node $q\in\detcarrier$, we associate a one-step pair
$(\gamma_q,\Theta_q^\Xi)$ over a set~$V_q$ of propositional variables,
given by
\begin{align}
  V_q & =\{a_{\hearts\psi}\mid\hearts\psi\in l^A(q)\}\nonumber\\
  \label{eq:node-one-step-pair}
  \gamma_q&=\{\hearts a_{\hearts\psi}\mid \hearts\psi\in l^A(q)\}\\
u_q^\kappa&=\{a_{\hearts\psi}\mid \hearts\psi\in \kappa\cap l^A(q)\}\quad\text{for $\kappa\in\mathsf{selections}$}\nonumber\\
\Theta_q^\Xi&=\{u_q^\kappa\mid \kappa\in \Xi\} \nonumber
\end{align}
Thus, $\gamma_q$ abstracts modalized formulae~$\hearts\psi$ found in
the label of~$q$ into~$\hearts a_{\hearts\psi}$, and $\Theta_q^\Xi$
contains, for each~$\kappa\in\Xi$, the set~$u^\kappa_q$ of variables
$a_{\hearts\psi}$ such that $\hearts\psi\in l^A(q)$ and
$\psi\in\Delta(\hearts\psi,\kappa)$, so that the non-deterministic
tracking automaton~$\autom_\target$ tracks $\hearts\psi$ to $\psi$
under~$\kappa$. In the reading of~$\Theta_q^\Xi$ suggested in
\autoref{rem:os-logic},~$\Theta_q^\Xi$ is understood as the
disjunction of the~$u_q^\kappa$ over all $\kappa\in\Xi$,
with~$u_q^\kappa$ read as the conjunctive clause
$\Land_{\hearts\psi\in\kappa\cap
  l^A(q)}a_{\hearts\psi}\land\Land_{\hearts\psi\in
  l^A(q)\setminus\kappa}\neg a_{\hearts\psi}$. We can similarly
understand~$\Xi$ as a disjunctive normal form over atoms of the form
$\hearts\psi\in l^A(q)$, and~$\Theta^\Xi_q$ arises from~$\Xi$ by
simply renaming these atoms into $a_{\hearts\psi}$. Notice also that
the interpretation of $a_{\hearts\psi}\in V_q$ over~$\Theta^\Xi_q$ as
per Definition~\ref{def:oss} is
$\sem{a}^{\Theta^\Xi_q}_0=\{u^\kappa_q\mid \kappa\in\Xi,
\hearts\psi\in\kappa\}$, a set which depends monotonically
on~$\Xi$. There is thus a balance to strike in selecting the
set~$\Xi$, which needs to be large enough to ensure satisfiability of
the one-step pair $(\gamma_q,\Theta_q^\Xi)$, but on the other hand
enlarging~$\Xi$ implies having to track more formulae.

Furthermore, given $\hearts\psi\in\FLtarget$ and
$\Xi\subseteq\mathsf{selections}$, we write
\begin{equation}\label{eq:Xi-hearts}
  \Xi/{\hearts\psi}=\{\kappa\in\Xi \mid \hearts\psi\in\kappa\}.
\end{equation}

\begin{defi}[Pre-tableaux and tableaux]\label{defn:pretab}
  A \emph{pre-tableau} $(W,\Sigma,\delta',\initstate,\detprio')$, or
  just $(W,\delta')$, \emph{for $\target$} consists of a set
  $W\subseteq\detcarrier$ of nodes, a partial transition map
  $\delta'\colon W\times \Sigma\pfun W$, and a priority map
  $\detprio'\colon W\to\mathbb N$ such that the following conditions
  hold:
  \begin{enumerate}
  \item\label{item:L-delta} $(W,\Sigma,\delta',\initstate,\detprio')$ is a \emph{partial
      subautomaton} of~$\mathsf{B}_\target$. That is, the initial
    node~$\initstate$ of $\mathsf{B}_\target$ is in~$W$;
    $\delta'(q,\sigma)=\delta(q,\sigma)$ whenever $\delta'(q,\sigma)$
    is defined for $q\in W,\sigma\in\Sigma$; and
    $\detprio'(q)=\detprio(q)$ for all $q\in W$. (Note that
    $\delta'(q,\sigma)$ may be undefined even when
    $\delta(q,\sigma)\in W$.)
  \item\label{item:tableau-tau} For all $q\in W$, we have
    $\bot\notin l^A(q)$, and there is a unique
    $\tau\in\mathsf{choices}$ such that $\delta'(q,\tau)$ is defined
    (it then equals $\delta(q,\tau)$ by (\ref{item:L-delta})).
  \item\label{item:tableau-Xi} For all $q\in W$, the one-step pair
    $(\gamma_q,\Theta^{\Xi(q)}_q)$ (notation as
    per~\eqref{eq:node-one-step-pair}) is one-step satisfiable, where
    \begin{equation}\label{eq:Xi-q}
      \Xi(q)=\{\kappa\in\mathsf{selections}\mid \kappa\subseteq l^A(q)\text{ and
        $\delta'(q,\kappa)$ is defined}\}
  \end{equation}
  (i.e.\ the modal label of $q$ is one-step satisfiable over the
  labels of the modal $\delta'$-successors of $q$).
  \end{enumerate}
  We refer to transitions in $(W,\delta')$ under letters in
  $\mathsf{choices}$ as \emph{local transitions}, and under letters in
  $\mathsf{selections}$ as \emph{modal transitions}. By
  (\ref{item:tableau-tau}), there is, from every~$q\in W$, a unique
  \emph{local} run, i.e.\ one consisting of only local transitions,
  which we denote by $\rho(q)$.  A \emph{tableau} is a pre-tableau in
  which every (infinite) run starting at $q_0=\initstate$ is
  accepting.
\end{defi}
\noindent Thus, a pre-tableau $(W,\delta')$ is obtained from
$\mathsf{B}_\target$ by keeping just a single outgoing local
transition at each node, and by removing some of the modal transitions
in such a way that the remaining modal transitions still suffice for
satisfaction of the modal literals in the label. In order for
$(W,\delta')$ to be a tableau, we additionally require that all runs
of this automaton are accepting (however, there may be infinite words
over~$\Sigma$ on which $(W,\delta')$ does not have a run, and which
are thus not accepted). Since the definition of~$u^\kappa_q$ as
per~\eqref{eq:node-one-step-pair} only deems a modal argument~$\psi$
to be satisfied in a~$\kappa$-successor of~$q$ if~$\psi$ is tracked
under~$\kappa$, there is a balance to be struck in choosing the modal
transitions to keep -- as indicated, these need to suffice to satisfy
all modal literals in the label, but every modal transition that is
kept induces more tracking that may expose infinite deferral.

\begin{exa}[Tableaux]\label{ex:tab}
  Recall the co-determinized tracking automaton $\mathsf{B}_\target$
  for the monotone $\mu$-calculus formula
  $\chi=\nu X.\,(a\land\mu Y.\,(X \lor\mondiamond{g}Y))$ from
  \autoref{ex:dtrack}. To avoid triviality, we slightly tweak the
  semantics to work with \emph{serial} monotone neighbourhood frames,
  i.e.\ coalgebras for the functor
  $\Mon_s^\AtGames\times\Pow(\PropAts)$ where~$\Mon_s$ is the
  \emph{serial monotone neighbourhood functor}~$\Mon_s$ given by
  $\Mon_s(X)=\{N\in\Mon X\mid \emptyset\notin N\neq\emptyset\}$. (In a
  serial monotone neighbourhood frame, we thus cannot satisfy a
  formula $\gldiamond{g}\phi$ at a state~$c$ by just making the empty
  set a neighbourhood of~$c$.) Below we show a partial automaton
  (obtained from $\mathsf{B}_\target$ by removing various transitions
  and nodes) that is a tableau for $\target$; for better comparison
  with the original automaton $\mathsf{B}_\target$, we use dotted
  transitions and nodes to depict the parts of~$\mathsf{B}_\target$
  that have been removed (and are not considered to belong to the
  tableau). Recall that
  $\kappa_{\mondiamond{g}\phi}=\{\mondiamond{g}\phi\}$.
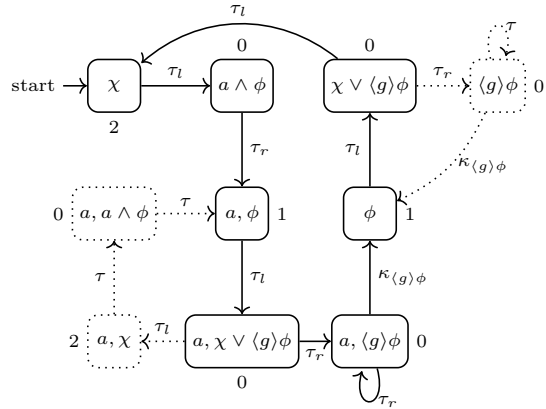
\begin{center}
\begin{tiny}
  \begin{tikzpicture}[
		auto,
    node distance=1.7cm,
    semithick
    ]
     \node[state,rectangle,rounded corners,initial left,label={below: $2$}] (e) {$\chi$};
     \node[state,rectangle,rounded corners,label={above: $0$}] (f) [right of=e] {$a\wedge\phi$};
     \node[state,rectangle,rounded corners,label={right: $1$}] (g) [below of=f] {{$a,\phi$}};
     \node[state,rectangle,rounded corners,label={below: $0$}] (h) [below of=g] {$a,\chi\vee\mondiamond{g}\phi$};
     \node[state,rectangle,rounded corners,label={right: $0$}] (i) [right of=h] {$a,\mondiamond{g}\phi$};
     \node[state,rectangle,rounded corners,label={right: $1$}] (j) [above of=i] {$\phi$};
     \node[state,rectangle,rounded corners,label={above: $0$}] (k) [above of=j] {$\chi\vee\mondiamond{g}\phi$};
     \node[state,dotted,rectangle,rounded corners,label={right: $0$}] (l) [right of=k] {$\mondiamond{g}\phi$};

     \node[state,dotted,rectangle,rounded corners,label={left: $2$}] (a) [left of=h] {$a,\chi$};
     \node[state,dotted,rectangle,rounded corners,label={left: $0$}] (b) [above of=a] {$a,a\land\phi$};
     \path[->] (e) edge node [above] {$\tau_l$} (f);
     \path[->] (f) edge node [right] {$\tau_r$} (g);
     \path[->] (g) edge node [right] {$\tau_l$} (h);
     \path[->] (h) edge node [below] {$\tau_r$} (i);
     \path[->] (i) edge node [right] {$\kappa_{\mondiamond{g}\phi}$} (j);
     \path[->] (i) edge [loop below] node [right] {$\tau_r$} (i);
     \path[->] (j) edge node [left] {$\tau_l$} (k);
     \path[->] (k) edge [bend right=40] node [above] {$\tau_l$} (e);

     \path[->] (l) edge [dotted, loop above] node [right] {$\tau$} (l);
     \path[->] (h) edge [dotted] node [above] {$\tau_l$} (a);
     \path[->] (a) edge [dotted] node [left] {$\tau$} (b);
     \path[->] (b) edge [dotted] node [above] {$\tau$} (g);
     \path[->] (k) edge [dotted] node [above] {$\tau_r$} (l);
     \path[->] (l) edge [dotted, bend left=20] node [right] {$\kappa_{\mondiamond{g}\phi}$} (j);

  \end{tikzpicture}
\end{tiny}
\end{center}
One easily verifies that this structure is indeed a tableau: First, it
is clearly a subautomaton of $\mathsf{B}_\target$.  Moreover, every
(reachable) node has exactly one outgoing local transition, and no
node contains $\bot$ in its label. Again we skip the treatment of
propositional atoms as modalities, and concentrate instead on one-step
satisfiability in the bottom-right node, which for purposes of the
subsequent discussion we denote as~$q$. In~$q$, the notation
introduced in~\eqref{eq:node-one-step-pair} and~\eqref{eq:Xi-q}
instantiates as follows. We have~$V_q=\{a_{\mondiamond{g}\phi}\}$,
$\gamma_q=\{\mondiamond{g}a_{\mondiamond{g}\phi}\}$,
$\Xi(q)=\{\kappa_{\mondiamond{g}\phi}\}$,
and~$\Theta^{\Xi(q)}_q=\{u^{\kappa_{\mondiamond{g}\phi}}_q\}=\{\{a_{\mondiamond{g}\phi}\}\}$.
One-step satisfiability of the one-step pair
$(\gamma_q,\Theta^{\Xi(q)}_q)=
(\{\mondiamond{g}a_{\mondiamond{g}\phi}\},\{\{a_{\mondiamond{g}\phi}\}\})$
is obvious. Finally, every infinite run
of this subautomaton either loops at the bottom-right node forever, or
visits the initial node infinitely often; in either case, the maximal
priority that is visited infinitely often is even, so the run is
accepting.
\end{exa}

We will see that there is a tableau for~$\target$ if and only
if~$\target$ is satisfiable. We go on to show one direction of this
statement (`only if') now; the other direction is a consequence of the
results of \autoref{sec:satgames} below. That is, we will
construct a coalgebraic model of the target formula~$\target$ on a
tableau for $\target$. As indicated above, the key property of such a
coalgebra is \emph{coherence} w.r.t.\ the tableau. We will build the
coalgebraic model using only so-called \emph{state nodes} of the
tableau, defined next.

\begin{defi}[Local runs, pre-tableau states]
  Let $(W,\delta')$ be a pre-tableau.
  A node $q\in W$ is a \emph{state node} if the local run $\rho(q)$
  that starts at $q$ is a cycle.
  We denote the set of state nodes of $(W,\delta')$
  by $\mathsf{states}(W,\delta')\subseteq W$.  For $q\in W$, we
  let~$\lceil q\rceil$ denote the first state node (for definiteness)
  on $\rho(q)$ (possibly $\lceil q\rceil =q$), and extend this
  notation to sets of nodes, putting
  $\lceil V\rceil=\{\lceil q\rceil\mid q\in V\}\subseteq
  \mathsf{states}(W,\delta')$ for $V\subseteq W$.
\end{defi}

\begin{exa}[Local runs, pre-tableau states]
  The tableau from \autoref{ex:tab} provides just a single way to
  construct local runs (namely, by eventually staying forever in the
  bottom right node) and, consequently, contains just a single state
  node. Exploiting the fact that labels happen to be unique in the
  example, we refer to states by their labels. Then we have
  $\lceil W\rceil=\mathsf{states}(W,\delta')=\{\{a, \langle
  g\rangle\phi\}\}$. We can however, modify the tableau in this
  example and obtain an alternative tableau by taking the left
  ($\tau_l$-)transition at $\{a,\chi\vee\mondiamond{g}\phi\}$, instead
  of the right ($\tau_r$-)transition. In this alternative tableau, we
  then have state nodes $\{a,\phi\}$,
  $\{a,\chi\vee\mondiamond{g}\phi\}$, $\{a,\chi\}$ and
  $\{a,a\land \phi\}$, all of which are part of the local run that
  loops through the bottom left cycle of the automaton. Then we have,
  e.g., $\lceil\{a\land\phi\}\rceil=\{a,\phi\}$ but
  $\lceil\{a,\chi\}\rceil=\{a,\chi\}$.
\end{exa}

\begin{rem}[Local runs, pre-tableau states]
  Observe that the labels of nodes along the local run $\rho(q)$
  become semantically stronger through the choice of disjuncts in
  disjunctions. In particular, the set of modal literals contained in
  the label grows monotonically along $\rho(q)$. At the same time,
  formulae may be syntactically lost from the label along steps of the
  local run; e.g.\ a disjunction may be replaced with a disjunct, a
  conjunction with both its conjuncts, and fixpoint literals may be
  unfolded. All (state) nodes on the local run of a state node are
  thus semantically equivalent, and contain the same modal
  literals, but otherwise may differ syntactically. We use the
  mechanism of local runs to avoid introducing a notion of non-modal
  entailment that combines propositional entailment and fixpoint
  unfolding.
\end{rem}

\begin{defi}[Coherence]\label{defn:strongcoherence}
  Let $(W,\delta')$ be a pre-tableau.  A coalgebra structure~$\xi$ on
  $\mathsf{states}(W,\delta')$ is \emph{coherent} (over $(W,\delta')$) if for all
  $q\in \mathsf{states}(W,\delta')$ and all $\hearts\psi\in \FLtarget$,
\begin{align*}
\hearts\psi\in l^A(q)\text{ implies } 
\xi(q)\in\sem{\hearts}\lceil \delta'(q,\mathsf{selections}/{\hearts\psi})\rceil.
\end{align*}
\end{defi}
\noindent Note that~$\psi\in l^A(\delta'(q,\kappa))$ for every
$\kappa\in\mathsf{selections}/{\hearts\psi}$, so~$\psi$ is
semantically entailed by $l^A(\lceil\delta'(q,\kappa)\rceil)$. The
converse, however, does not hold, i.e.\ even if $l^A(q')$
entails~$\psi$, it need not be the case that~$q'$ is on the local run
of some node in
$\delta'(q,\mathsf{selections}/{\hearts\psi})$. Requiring that
$\xi(q)\in\sem{\hearts}(\lceil
\delta'(q,\mathsf{selections}/{\hearts\psi})\rceil)$ in the above
definition thus means that we insist that~$\hearts\psi$ is satisfied
considering only those successors of~$q$ to which~$\psi$ is tracked.

Due to property ($3$) of pre-tableaux, coherent coalgebra structures
always exist:

\begin{lem}[Existence lemma]\label{lem:existence}
  Let $(W,\delta')$ be a pre-tableau. Then there is a coherent coalgebra
  structure on $\mathsf{states}(W,\delta')$.
\end{lem}
\begin{proof}
  Let $q\in\mathsf{states}(W,\delta')$ be a state node. Since
  $(W,\delta')$ is a pre-tableau, we can pick
  \begin{equation*}
    t\in\sem{\gamma_q}_1^{\Theta_q^{\Xi(q)}},
  \end{equation*}
  in notation as per~\eqref{eq:node-one-step-pair}
  and~\eqref{eq:Xi-q}; in particular,
  $\Theta_q^{\Xi(q)}=\{u^\kappa_q\mid\kappa\subseteq l^A(q) \text{
    and}$ $\delta'(q,\kappa)$ $\text{is defined}\}$ where
  $u^\kappa_q=\{a_{\hearts\psi}\mid\hearts\psi\in\kappa\}$; so
  $\Theta_q^{\Xi(q)}$ abstracts the labels of the successors of~$q$ in
  the pretableau $(W,\delta')$.
  Let~$h\colon \Theta_q^{\Xi(q)}\to\Xi(q)$ be a section of the
  surjective map~$\Xi(q)\to\Theta_q^{\Xi(q)}$,
  $\kappa\mapsto u^\kappa_q$ (so $u=u^{h(u)}_q$ for
  $u\in\Theta_q^{\Xi(q)}$), and put $\xi(q)=(Fg)(t)$ where
  $g\colon \Theta_q^{\Xi(q)}\to\mathsf{states}(W,\delta')$ is defined
  by $g(u)=\lceil \delta'(q,h(u))\rceil$.  We show that thus
  defined,~$\xi$ is coherent: Let $\hearts\psi\in l^A(q)$. Then
  $\hearts a_{\hearts\psi}\in\gamma_q$, so
  $t\in\sem{\hearts} \sem{a_{\hearts\psi}}^{\Theta_q^{\Xi(q)}}_0$. By
  naturality and monotonicity of $\sem{\hearts}$,
  $\xi(q)\in\sem{\hearts}\lceil
  \delta'(q,\mathsf{selections}/{\hearts\psi})\rceil$ follows once we
  show that
  \begin{equation*}
    \sem{a_{\hearts\psi}}^{\Theta_q^{\Xi(q)}}_0\subseteq g^{-1}[\lceil \delta'(q,\mathsf{selections}/{\hearts\psi})\rceil].
  \end{equation*}
  So let $u\in\sem{a_{\hearts\psi}}^{\Theta_q^{\Xi(q)}}_0$, that is,
  $a_{\hearts\psi}\in u$. By~\eqref{eq:Xi-q}
  and~\eqref{eq:node-one-step-pair}, $\delta'(q,h(u))$ is defined and
  $\hearts\psi\in h(u)$, so
  $\delta'(q,h(u))\in\delta'(q,\mathsf{selections}/{\hearts\psi})$
  and hence
  $g(u)\in \lceil
  \delta'(q,\mathsf{selections}/{\hearts\psi})\rceil$ as required.
\end{proof}

\noindent We finally show that a coherent coalgebra structure is
indeed a model of~$\target$:

\begin{lem}[Truth lemma]\label{lem:truth}
  Let $(W,\delta')$ be a tableau, and let~$\xi$ be a coherent
  coalgebra structure on $V:=\mathsf{states}(W,\delta')$. Then
  $\lceil\initstate\rceil\in\sem{\target}$ in $(V,\xi)$.
\end{lem}
\begin{proof}
  By \autoref{thm:satisfaction-game}, it suffices to show
  that~$\exists$ wins the position $(\lceil \initstate\rceil,\target)$
  in the model checking game $\game_{{\target,(V,\xi)}}$. We define a
  history-dependent $\exists$-strategy~$s$ in
  $\game_{{\target,(V,\xi)}}$, maintaining the invariant that if
  $(q_n,\psi_n)$ is the $n$-th position of the shape $(q',\psi')$
  visited in the play, then
  \begin{quote}
    there is $u_n\in W$ such that $\psi_n\in l^A(u_n)$ and $q_n$ lies
    on the local run~$\rho(u_n)$, and for $n>0$ there is a word~$w_n$
    such that $u_n=\delta'(u_{n-1},w_n)$ and
    $\psi_n\in\Delta(\psi_{n-1},w_n)$. Moreover,
    \begin{itemize}
    \item If $\psi_{n-1}$ has the form $\psi_{n-1}=\hearts\phi$, then
      $w_{n-1}$ has the form $w_{n-1}=\tau_1,\dots,\tau_m,\kappa$
      where $m\ge 0$, $\tau_1,\dots,\tau_m\in\mathsf{choices}$ and
      $\kappa\in\mathsf{selections}/{\hearts\phi}$. (Notice that this
      description of~$w_{n-1}$ already
      implies~$\Delta(\psi_{n-1},w_{n-1})=\Delta(\hearts\phi,w_{n-1})=\{\phi\}=\{\psi_n\}$.)
    \item Otherwise, $w_{n-1}$ has the form $w_n=\tau$ where
      $\tau\in\mathsf{choices}$.
    \end{itemize}
  \end{quote}
  The history-dependence of~$s$ is caused by keeping the tableau node
  $u_{n}$ in memory. The invariant holds initially, i.e.\ at
  $(q_0,\psi_0)=(\lceil \initstate\rceil,\target)$, for
  $u_0=\initstate$. We show next that~$\exists$ can enforce the
  invariant at $(q_{n+1},\psi_{n+1})$ if it holds at
  $(q_n,\psi_n)$. Since $(W,\delta')$ is a pre-tableau, we have a
  unique $\tau\in\mathsf{choices}$ such that $\delta'(u_i,\tau)$ is
  defined. We distinguish cases on~$\psi_n$:
  \begin{enumerate}
  \item $\psi_n=\bot$: By the invariant and the definition of
    pre-tableaux, this case does not occur.
  \item $\psi_n=\top$: $\exists$ wins immediately.
  \item $\psi_n=\phi_1\land\phi_2$: Then $(q_n,\psi_n)$ belongs
    to~$\forall$, who moves to $(q_{n+1},\psi_{n+1})$ where
    $q_{n+1}=q_n$ and $\psi_{n+1}\in\{\phi_1,\phi_2\}$. The invariant
    is preserved by taking $u_{n+1}=\delta'(u_i,\tau)$.
  \pagebreak
  \item $\psi_n=\phi_1\vee\phi_2$: We define~$s$ by letting~$\exists$
    move to $(q_{n+1},\psi_{n+1})=(q_n,\tau(\psi_n))$. Again, the
    invariant is preserved by taking $u_{n+1}=\delta'(u_i,\tau)$.
  \item $\psi_n=\eta x.\phi$: We define~$s$ by letting~$\exists$ play
    the only available move, to
    $(q_{n+1},\psi_{n+1})=(q_n,\phi[\eta x.\phi/x])$; again, the
    invariant is preserved by taking $u_{n+1}=\delta'(u_i,\tau)$.
  \item $\psi_{n}= \hearts\phi$: It follows from the invariant that
    $\hearts\phi\in l^A(q_n)$, since $\hearts\phi$ is never processed
    along $\rho(u_n)$. Since $\xi$ is coherent, we thus have
    $\xi(q_n)\in\sem{\hearts}(D)$ for
    $D= \lceil \delta'(q_n,\mathsf{selections}/{\hearts\phi})\rceil$,
    so we can define~$s$ by letting~$\exists$ move to $(D,\phi)$. If
    $D=\emptyset$, then~$\exists$ wins
    immediately. Otherwise,~$\forall$ moves to some position
    $(q_{n+1},\phi)$ (so $\psi_{n+1}=\phi$) such that $q_{n+1}\in D$,
    i.e.\ there is $\kappa\in\mathsf{selections}/{\hearts\phi}$ such
    that $q_{n+1}=\lceil\delta(q_n,\kappa)\rceil$. Since by the
    invariant,~$q_n$ lies on the local path~ $\rho(u_n)$, we have
    $\delta(q_n,\kappa)=\delta(u_n,w_n)$ for $w_n\in\Sigma^*$ of the
    required form $w_n=\tau_1,\dots,\tau_m,\kappa$ where
    $\tau_1,\dots,\tau_m\in\mathsf{choices}$, so the invariant is
    preserved by taking $u_{n+1}=\delta(q_{n},\kappa)$. In particular,
    $\psi_{n+1}=\phi$ is in $l^A(u_{n+1})$ and in $\Delta(\psi_n,w_n)$
    because $\hearts\phi\in\kappa$, and $q_{n+1}$ is on the local
    path~$\rho(u_{n+1})$.\medskip
  \end{enumerate} 
  We have to show that $s$ is a winning strategy; since, as we have
  noted, the game never reaches a position of the form
  $(q,\bot)$,~$\exists$ wins all finite plays, so it remains only to
  show that~$\exists$ wins every infinite play~$\pi$ that
  follows~$s$. By the invariant,~$\pi$ induces a word
  $w=w_0w_1\ldots\in\Sigma^\omega$, a run $\bar\pi=u_0,u_1,\ldots$
  (with $u_0=\initstate$) of the tableau $(W,\delta')$ on $w$, and a
  run~$\rho$ of the non-deterministic tracking
  automaton~$\mathsf{A}_\target$ on~$w$. Since $(W,\delta')$ is a
  tableau and~$w$ has an infinite run,~$w$ is accepted
  by~$(W,\delta')$ and hence also by~$\mathsf{B}_\target$. Thus,~$w$
  is rejected by~$\autom_\target$; in particular,~$\rho$ is a
  non-accepting run of~$\autom_\target$. Now~$\rho$ differs from the
  sequence of formulae occurring in~$\pi$ only by possible finite
  repetition of formulae of the form~$\psi_n=\hearts\phi$, caused
  by~$\mathsf{A}_\target$ looping on the choice functions occurring
  in~$w_n$. As the winning objective in $\game_{{\target,(V,\xi)}}$ is
  dual to acceptance in $\mathsf{A}_\target$, which is unaffected by
  finite repetition of letters (`stuttering'),~$\exists$ thus wins the
  play~$\pi$.
\end{proof}

\section{Satisfiability Games}\label{sec:satgames}

We now introduce a generic game characterization of satisfiability
in the coalgebraic $\mu$-calculus (Definition~\ref{defn:satgames}),
rooted in classical algorithmic treatments of the relational
$\mu$-calculus as well as in previous work on the coalgebraic
$\mu$-calculus~\cite{FontaineEA10,CirsteaEA11a} (see
\autoref{sec:intro} and \autoref{rem:flv} for a detailed
discussion). In the game, the existential player effectively attempts
to establish existence of a tableau (\autoref{sec:tableaux}) for
the target formula~$\target$. Like tableaux, the game thus involves
the notion of one-step satisfiability (Definition~\ref{def:oss}). (We
note that a similar condition appears in a previous notion of
satisfiability game for the coalgebraic
$\mu$-calculus~\cite{FontaineEA10}, which however is otherwise
markedly different from ours, cf.\ \autoref{rem:flv}. The notion of
tableau game used in the algorithm based on complete sets of modal
tableau rules~\cite{CirsteaEA11a} has a more similar shape to ours but
appears slightly larger in that automata nodes that go into the model
construction are additionally annotated with tableau sequents.) We
prove correctness of the game by showing on the one hand that a
tableau may indeed be extracted from a winning strategy of the
existential player (\emph{completeness}), and on the other hand that a
winning strategy of the existential player can be extracted from a
given model of the target formula~$\target$ (\emph{soundness}), and
indeed may be obtained from a winning strategy of the existential
player in the corresponding model checking game.
\pagebreak

We first present the definition of the satisfiability game, which like
the model checking game (\autoref{sec:tracking}) takes the shape of
a standard parity game. Recall that
$\mathsf{B}_\target=(\detcarrier,\Sigma,\delta,\initstate,\Omega)$ is
the co-determinized tracking automaton for the target
formula~$\target$, and comes with the labelling function
$l^A:\detcarrier\to\Pow(\FLtarget)$.

\begin{defi}[Satisfiability game]\label{defn:satgames}

The \emph{satisfiability game} $\game_\target=(V_\forall,V_\exists,E,v_0,\Omega')$ for $\target$ is a parity game
with sets
\begin{align*}
V_\exists&=\detcarrier\times\{0,1\} &
V_\forall&=\detcarrier\cup(\detcarrier\times\Pow(\mathsf{selections}))
\end{align*}
of positions and $v_0=q_\mathsf{init}$. 
The moves and the priorities in $\game_\target$ are
defined by the following table, where $q\in\detcarrier$ and
$\Xi\in \Pow(\mathsf{selections})$, and we write
\begin{align*}
\mathsf{selections}(q)=\{\kappa\in\mathsf{selections}\mid \kappa\subseteq l^A(q)\}.
\end{align*}\medskip
\vspace*{-0.5cm}
\begin{center}
\begin{tabular}{|c|c|c|c|}
\hline
position & owner & set of allowed moves & priority\\
\hline
$q$ & $\forall$ & $\{(q,0),(q,1)\}$ & $\Omega(q)$\\
$(q,0)$ & $\exists$ & $\{\delta(q,\tau)\in\detcarrier\mid\tau\in\mathsf{choices},\bot\notin l^A(q) \}$ & $0$\\
$(q,1)$ & $\exists$ & $\{(q,\Xi)\mid\Xi\in\Pow(\mathsf{selections}(q)),
(\gamma_q,\Theta_q^\Xi) \text{ is satisfiable}\}$ & $0$\\
$(q,\Xi)$ & $\forall$ & $\{\delta(q,\kappa)\mid \kappa\in\Xi\}$ & $0$\\
\hline
\end{tabular}
\end{center}
\end{defi}\medskip
\noindent Recall here that the components of the one-step pair
$(\gamma_q,\Theta_q^\Xi)$ are
\begin{align*}
  \gamma_q&=\{\hearts a_{\hearts\psi}\mid \hearts\psi\in l^A(q)\} \hspace{10em}\text{and}\\
  \Theta_q^\Xi&=\{u^\kappa_q=\{a_{\hearts\psi}\mid\hearts\psi\in
                l^A(q)\cap\kappa\}\mid
                \kappa\in \Xi\}.
\end{align*}
Thus, $\game_\target$ is a parity game with $\prios$ priorities
(inherited from~$\mathsf{B}_\target$, cf.\ \autoref{sec:tracking})
and $|\detcarrier|(3+2^{2^{n_0}})$ positions.  As indicated above, the
game is aimed at determining whether there exists a tableau
for~$\target$, and thus closely follows
Definition~\ref{defn:pretab}. Specifically, when the play reaches a
node~$q\in\detcarrier$, this indicates that the node needs to be
included in the tableau, so~$\forall$ may challenge either the
propositional clause~\eqref{item:tableau-tau} or the modal
clause~\eqref{item:tableau-Xi} of the definition of pre-tableaux by
moving to $(q,0)$ or to $(q,1)$, respectively. The admissibility
conditions for the respective subsequent $\exists$-moves match the
conditions given in the relevant clauses of
Definition~\ref{defn:pretab}; in particular,~$\exists$ loses $(q,0)$
if $\bot\in l^A(q)$. The winning condition of~$\game_\target$ ensures
the tableau property from Definition~\ref{defn:pretab}; that
is,~$\exists$ wins a play~$\pi$ iff the sequence of positions of the
form $q\in\detcarrier$ encountered on~$\pi$ is an accepting run
of~$\mathsf{B}_\target$. We formally establish in \autoref{lem:games}
that~$\exists$ winning the game really does guarantee existence of a
tableau for~$\target$.

\begin{exa}[Satisfiability game]\label{ex:satgame}
  Recall the co-determinized tracking automaton $\mathsf{B}_\target$
  for the monotone $\mu$-calculus formula
  $\chi=\nu X.\,(a\land\mu Y.\,(X \lor\mondiamond{g}Y))$ from
  \autoref{ex:dtrack} (with $\phi$ abbreviating
  $\mu Y.\,(X \lor\mondiamond{g}Y)$). Like already in
  \autoref{ex:tab}, we restrict the semantics to serial monotone
  neighbourhood frames. Below we show the satisfiability game
  $\mathsf{G}_\target$ for $\target$, constructed over
  $\mathsf{B}_\target$. Again, rounded boxes indicate
  $\exists$-positions.

\begin{center}
\begin{tiny}
  \begin{tikzpicture}[
		auto,
    node distance=1.7cm,
    semithick
    ]
     \node[state,rectangle,initial left,label={above: $2$}] (e) {$\chi$};
     \node[state,rectangle,rounded corners,label={above: $0$}] (ea) [right of=e] {$\chi,0$};
     \node[state,rectangle,label={above: $0$}] (f) [right of=ea] {$a\wedge\phi$};
     \node[state,rectangle,rounded corners,label={right: $0$}] (fa) [below of=f] {$a\wedge\phi,0$};
     \node[state,rectangle,label={right: $1$}] (g) [below of=fa] {{$a,\phi$}};
     \node[state,rectangle,rounded corners,label={right: $0$}] (ga) [below of=g] {{$a,\phi,0$}};
     \node[state,rectangle,label={right: $0$}] (h) [below of=ga] {$a,\chi\vee\mondiamond{g}\phi$};
     \node[state,rectangle,rounded corners,label={below: $0$}] (ha) [below of=h] {$a,\chi\vee\mondiamond{g}\phi,0$};
     \node (wa) [right of=ha] {};
     \node[state,rectangle,label={below: $0$}] (i) [right of=wa] {$a,\mondiamond{g}\phi$};
     \node (wu) [right of=i] {};
     \node[state,rectangle,rounded corners,label={right: $0$}] (ya) [right of=wu] {$a,\mondiamond{g}\phi,0$};
     \node[state,rectangle,rounded corners,label={above: $0$}] (ia) [above of=i] {$a,\mondiamond{g}\phi,1$};
     \node (yo) [right of=ia] {};
     \node[state,rectangle,label={right: $0$}] (za) [right of=yo] {$a,\mondiamond{g}\phi,\{\{\mondiamond{g}\phi\}\}$};
     \node (zu) [above of=ia] {};
     \node[state,rectangle,label={right: $1$}] (j) [right of=zu] {$\phi$};
     \node[state,rectangle,rounded corners,label={right: $0$}] (ja) [above of=zu] {$\phi,0$};
     \node[state,rectangle,label={right: $0$}] (k) [above of=ja] {$\chi\vee\mondiamond{g}\phi$};
     \node[state,rectangle,rounded corners,label={above: $0$}] (ka) [above of=k] {$\chi\vee\mondiamond{g}\phi,0$};
     \node (yu) [right of=ka] {};
     \node[state,rectangle,label={above: $0$}] (l) [right of=yu] {$\mondiamond{g}\phi$};
     \node[state,rectangle,rounded corners,label={right: $0$}] (t) [right of=l] {$\mondiamond{g}\phi,0$};
     \node[state,rectangle,rounded corners,label={right: $0$}] (la) [below of=l] {$\mondiamond{g}\phi,1$};

     \node (zo) [above of=za] {};
     \node[state,rectangle,label={right: $0$}] (xa) [above of=zo] {$\mondiamond{g}\phi,\{\{\mondiamond{g}\phi\}\}$};
     \node (wo) [left of=ha] {};
     \node[state,rectangle,label={left: $2$}] (a) [left of=wo] {$a,\chi$};
     \node[state,rectangle,rounded corners,label={left: $0$}] (aa) [above of=a] {$a,\chi,0$};
     \node[state,rectangle,label={left: $0$}] (b) [above of=aa] {$a,a\land\phi$};
     \node[state,rectangle,rounded corners,label={left: $0$}] (ba) [above of=b] {$a,a\land\phi,0$};
     \path[->] (e) edge node [above] {} (ea);
     \path[->] (ea) edge node [above] {} (f);
     \path[->] (f) edge node [right] {} (fa);
     \path[->] (fa) edge node [right] {} (g);
     \path[->] (g) edge node [right] {} (ga);
     \path[->] (ga) edge node [right] {} (h);
     \path[->] (h) edge node [below] {} (ha);
     \path[->] (ha) edge [very thick] node [below] {} (i);
     \path[->] (i) edge node [right] {} (ia);
     \path[->] (ia) edge node [right] {} (za);
     \path[->] (za) edge node [right] {} (j);
     \path[->] (i) edge [bend left=20] node [right] {} (ya);
     \path[->] (ya) edge [bend left=20] node [right] {} (i);
     \path[->] (j) edge node [left] {} (ja);
     \path[->] (ja) edge node [left] {} (k);
     \path[->] (k)  edge node [left] {} (ka);
     \path[->] (ka) edge [very thick, bend right=30] node [above] {} (e);
     \path[->] (ka) edge node [above] {} (l);
     \path[->] (ha) edge node [above] {} (a);
     \path[->] (a) edge node [left] {} (aa);
     \path[->] (aa) edge node [left] {} (b);
     \path[->] (b) edge node [above] {} (ba);
     \path[->] (ba) edge node [above] {} (g);
     \path[->] (l)  edge node [left] {} (la);
     \path[->] (la) edge node [right] {} (xa);
     \path[->] (xa) edge node [right] {} (j);
     \path[->] (l) edge [bend left=20] node [right] {} (t);
     \path[->] (t) edge [bend left=20] node [right] {} (l);

  \end{tikzpicture}
\end{tiny}
\end{center}
(In the figure, we have omitted positions of the form $(q,1)$ where
$l^A(q)$ does not contain any modal literal, such as $(\chi,1)$ in the
present example, which~$\exists$ wins immediately by moving to
$(q,\emptyset)$.) The bold arrows indicate an $\exists$-strategy~$s$
that wins the initial position $\chi$ in the game. As indicated above,
such strategies induce tableaux; in this case,~$s$ induces
(essentially) the tableau considered in \autoref{ex:tab} above.

\end{exa}

\noindent As indicated above, we prove completeness of the game by
showing that winning strategies for the existential player in the
satisfiability game induce tableaux:

\begin{lem}\label{lem:games}
  If the existential player wins the satisfiability game $\game_\target$,
  then there is a tableau for~$\target$.
\end{lem}
\begin{proof}
  Let~$s\colon V_\exists\to V$ be a history-free $\exists$-strategy
  that wins $\initstate$ in~$\game_\target$, and let~$W$ be the set of
  positions of the form $q\in\detcarrier$ that are reachable in plays
  that follow~$s$.
  We then define $\delta'\colon W\times\Sigma\pfun W$ as follows. Let
  $q\in W$. By the definition of the satisfiability game, there is
  $\tau\in\mathsf{choices}$ such that $s(q,0)=\delta(q,\tau)\in W$ (in
  particular, $\bot\notin l^A(q)$); we put
  $\delta'(q,\tau)=\delta(q,\tau)$, and let $\delta'(q,\tau')$ be
  undefined for all other $\tau'\in\mathsf{choices}$.  Similarly,
  $s(q,1)$ has the form $s(q,1)=(q,\Xi)$ where
  $\Xi\in\Pow(\mathsf{selections}(q))$, and we put
  $\delta'(q,\kappa)=\delta(q,\kappa)$ if $\kappa\in \Xi$, and let
  $\delta'(q,\kappa)$ be undefined otherwise, noting that
  $\delta(q,\kappa)\in W$ for $\kappa\in \Xi$ because~$\forall$ can
  move to $\delta(q,\kappa)$ from $(q,\Xi)=s(q,1)$. Then $(W,\delta')$
  is a pre-tableau by construction. To show that $(W,\delta')$ is a
  tableau, let $\initstate=q_0,q_1,\dots$ be a run of $(W,\delta')$ on
  a word $w=\sigma_0,\sigma_1.\dots$ By construction of $(W,\delta')$,
  this run induces a play~$\pi$ in $\game_\target$ that follows~$s$
  (explicitly, if $\sigma_i\in\mathsf{choices}$, then the play has one
  intermediate position~$(q_i,0)$ between~$q_i$
  and~$q_{i+1}=s(q_i,0)$, and if $\sigma_i\in\mathsf{selections}$,
  then the play has two intermediate positions $(q_i,1)$ and
  $s(q_i,1)$ between~$q_i$ and~$q_{i+1}$).  Since $s$ is a winning
  strategy,~$\pi$ is won by~$\exists$, which by the comments after
  Definition~\ref{defn:satgames} means that $q_0,q_1,\dots$ is an
  accepting run of $\mathsf{B}_\target$, showing that $(W,\delta')$ is
  a tableau.
\end{proof}

\noindent It remains to prove soundness. As indicated above, we
proceed by transforming a winning strategy in the model checking game
into one in the satisfiability game, exploiting that the former is
based on non-acceptance in the tracking automaton~$\mathsf{A}_\target$
and the latter on acceptance in the co-determinized tracking
automaton~$\mathsf{B}_\target$.

\begin{lem}[Soundness]\label{lem:modeltogame} Let $\target$ be satisfiable. Then
the existential player wins~$\game_\target$.
\end{lem}
\begin{proof}
  Fix a coalgebra $(C,\xi)$ and a state $x_0\in C$ such that
  $x_0\models\target$.  By
  \autoref{thm:satisfaction-game},~$\exists$ has a history-free
  strategy~$s'$ that wins~$(x_0,\target)$ in
  $\game_{\target,(C,\xi)}$. We construct a history-dependent winning
  strategy~$s$ for~$\exists$ in the satisfiability game
  $\game_\target$ that maintains the invariant that in positions
  $(q,b)\in V_\exists=\detcarrier\times\{0,1\}$, there is $x\in C$
  such that
  \begin{equation*}
    \text{for all }\psi\in l^A(q), \text{ $s'$ wins }(x,\psi) \text{ in } \game_{{\target,(C,\xi)}};
  \end{equation*}
  we call such an~$x$ a \emph{realizer} of~$q$. More precisely, we
  keep the realizer in memory, and in each step, we construct the new
  realizer from the previous one; this is why the strategy we
  construct is not history-free.  We write $s((q,b),x)$ for the move
  recommended by~$s$ when in a position $(q,b)\in V_\exists$ with
  realizer~$x$.

  We can pick~$x_0$ as the realizer of~$\initstate$, ensuring that the
  invariant holds initially since $l^A(\initstate)=\{\target\}$ and
  $s'$ wins $(x_0,\target)$. To see that the existential player can
  maintain the invariant, let
  $(q,b)\in V_\exists=\detcarrier\times\{0,1\}$ and let $x\in C$ be a
  realizer of $q$. We distinguish cases on~$b$.

  Case $b=0$: We define $\tau\in\mathsf{choices}$ as follows.  For each
  disjunction $\psi=\psi_1\vee\psi_2\in l^A(q)$,~$s'$ wins $(x,\psi)$
  by the invariant, and we have $s'(x,\psi)=(x,\psi_i)$ for some
  $i\in\{1,2\}$; of course,~$s'$ then wins $(x,\psi_i)$. We put
  $\tau(\psi)=\psi_i$.  For disjunctions $\psi\in\FLtarget$ not
  contained in $l^A(q)$, define $\tau(\psi)$ arbitrarily. We put
  $s((q,0),x)=\delta(q,\tau)$: since the invariant implies in
  particular that $\bot\notin l^A(q))$, this is a valid move.

  To establish the invariant at the new position $\delta(q,\tau)$, we
  pick the original realizer~$x$ of~$q$ as the new realizer of
  $\delta(q,\tau)$. We have to show that for all
  $\psi\in l^A(\delta(q,\tau))$,~$s'$ wins $(x,\psi)$ in in
  $\game_{{\target,(C,\xi)}}$.  By the definition of the tracking
  automaton (Definition~\ref{def:tracking}), such a~$\psi$ arises in
  one of the following ways:
  \begin{itemize}
  \item $\psi$ is a disjunct of disjunction in $l^A(q)$. It was shown
    above that~$s'$ wins $(x,\psi)$ in this case.
  \item $\psi$ is a conjunct (w.l.o.g., the left one) of a formula
    $\psi\land\phi\in l^A(q)$. Then the universal player can move from
    $(x,\psi\land\phi)$ to $(x,\psi)$ in
    $\game_{{\target,(C,\xi)}}$. Since~$s'$ wins $(x,\psi\land\phi)$
    by the invariant,~$s'$ also wins $(x,\psi)$.
  \item $\psi=\psi_1[\eta X.\,\psi_1/X]$ for some
    $\eta X.\,\psi_1\in l^A(q)$. By the invariant,~$s'$ wins
    $(x,\eta X.\,\psi_1)$ in $\game_{{\target,(C,\xi)}}$, so~$s'$
    also wins the unique next position $(x,\psi)$.
  \item $\psi=\hearts\psi_1\in l^A(q)$, Then~$s'$ wins $(x,\psi)$ by
    the invariant.
  \end{itemize}

  Case $b=1$: For each $\hearts\psi\in l^A(q)$,~$s'$ wins
  $(x,\hearts\psi)$ in $\game_{{\target,(C,\xi)}}$ by the invariant.
  Then $s'(x,\hearts\psi)$ has the form
  \begin{equation}\label{eq:U-psi}
    s'(x,\hearts\psi)=(D_{\hearts\psi},\psi)\quad\text{where}\quad \xi(x)\in\sem{\hearts}_C(D_{\hearts\psi}).
  \end{equation}
  Put
  \begin{equation*}\textstyle
    \Xi=\{\kappa\in\mathsf{selections}\mid
     \kappa\subseteq l^A(q)\text{ and }
    \bigcap_{\hearts\psi\in\kappa} D_{\hearts\psi}\neq\emptyset\}.
  \end{equation*}
  For $\kappa\in\Xi$, fix an element
  $y_\kappa\in\bigcap_{\hearts\psi\in\kappa} D_{\hearts\psi}$.  We put
  $s((q,1),x)=(q,\Xi)$. The ensuing moves of the universal player then
  reach a position in~$V_\exists$ of the form $(\delta(q,\kappa),b')$
  where $\kappa\in\Xi$. We pick~$y_\kappa$ as the realizer of
  $\delta(q,\kappa)$. To show that this ensures the invariant, let
  $\psi\in l^A(\delta(q,\kappa))$, that is, $\hearts\psi\in\kappa$ for
  some $\hearts\in\Lambda$; we have to show that~$s'$ wins
  $(y_\kappa,\psi)$. But this follows from the fact that~$s'$ is a
  winning strategy and the universal player can move from
  $s'(x,\hearts\psi)=(D_{\hearts\psi},\psi)$ to
  $(y_\kappa,\psi)$, since $y_\kappa\in D_{\hearts\psi}$.

  It remains to show that
  $((q,1),(q,\Xi))$ is a valid move in $\game_\target$. Recalling that
  \begin{align*}
V_q&=\{a_{\hearts\psi}\mid \hearts\psi\in l^A(q)\}\\
\gamma_q&=\{\hearts a_{\hearts\psi}\mid \hearts\psi\in l^A(q)\}\\
u_q^\kappa&=\{a_{\hearts\psi}\mid \hearts\psi\in \kappa\cap l^A(q)\}\\
\Theta_q^\Xi&=\{u_q^\kappa\mid \kappa\in \Xi\},
  \end{align*} we  have to show that, as discussed
  after~\eqref{eq:node-one-step-pair},~$\Xi$ is large enough to ensure that
    \begin{equation}\label{eq:oss-automaton}
      \sem{\gamma_q}^{H}_1
      \neq\emptyset
    \end{equation}
    where $H=\Theta_q^\Xi$; recall here that by definition,
    $\sem{\gamma_q}^{H}_1=\bigcap_{\hearts
      a_{\hearts\psi}\in\gamma_q}\sem{\hearts}_H\sem{a_{\hearts\psi}}^H_0$.
    We define a labelling $l\colon C\to H$ by
    \begin{equation}\label{eq:lx}
      l(y)=\{a_{\hearts\psi}\in V_q\mid y\in D_{\hearts\psi}\}.
    \end{equation}
    Here, we have to show that for $y\in C$, we indeed have
    $l(y)\in H$, i.e.\ we have to find $\kappa\in\Xi$ such that
    $l(y)=u_q^\kappa$.  This will hold by definition for
    $\kappa:=\{\hearts\psi\mid a_{\hearts\psi}\in l(y)\}$,
    once we show that $\kappa\in\Xi$. The latter means that
    $D:=\bigcap_{a_{\hearts\psi}\in
      l(y)}D_{\hearts\psi}\neq\emptyset$; but $y\in D$ by definition
    of $l(y)$.  We now establish~\eqref{eq:oss-automaton} by showing
    that
    \begin{equation*}
      Fl(\xi(x))\in\sem{\gamma_q}^H_1.
    \end{equation*}
    So let $\hearts a_{\hearts\psi}\in\gamma_q$, that is,
    $\hearts\psi\in l^A(q)$; we have to show
    $Fl(\xi(x))\in\sem{\hearts}_H\sem{a_{\hearts\psi}}^H_0$,
    which by naturality is equivalent to
    $\xi(x)\in\sem{\hearts}_{C}(l^{-1}[\sem{a_{\hearts\psi}}^H_0])=\sem{\hearts}_{C}(D_{\hearts\psi})$;
    but this holds by~\eqref{eq:U-psi}.

    This concludes the proof that~$s$ is a strategy, that is, yields
    legal moves in~$\game_\target$. It remains to show that~$s$ is
    winning. In showing that~$\exists$ can maintain the invariant, we
    have in particular shown that~$\exists$ never gets stuck, and
    hence wins all finite plays. So let~$\pi$ be an infinite play that
    starts at~$v_0$ and follows~$s$; we have to show that~$\exists$
    wins~$\pi$.  As noted after Definition~\ref{defn:satgames},~$\pi$
    gives rise to a run $r$ of~$\mathsf{B}_\target$ on a word
    $w\in\Sigma^\omega$. In more detail, the play~$\pi$ consists of
    concatenated subplays~$\pi_i$, either of the shape
    $q_i,(q_i,0),q_{i+1}$ where $q_{i+1}=\delta(q_i,\sigma_i)$ for
    some $\sigma_i\in\mathsf{choices}$, or of the shape
    $q_i,(q_i,1),(q_i,\Xi),q_{i+1}$ where
    $q_{i+1}=\delta(q_i,\sigma_i)$ for some $\sigma_i\in\Xi$. In this
    notation, $r=q_0,q_1,\ldots$, where $q_0=\initstate$, is the run
    of~$\mathsf{B}_\chi$ on~$w=\sigma_0,\sigma_1,\dots$. Again as
    noted after Definition~\ref{defn:satgames}, the winning objective
    of the existential player in~$\game_\target$ is
    $w\in L(\mathsf{B}_\chi)$. Since~$\mathsf{B}_\target$
    complements~$\mathsf{A}_\target$, we show equivalently that every
    run~$\rho=\psi_0,\psi_1,\dots$ of~$\mathsf{A}_\chi$ on the
    word~$w$, where $\psi_0=\target$, is non-accepting.  From~$\rho$,
    we obtain a play~$\pi'$ of the model checking
    game~$\game_{{\target,(C,\xi)}}$ that starts at $(x_0,\target)$
    and follows~$s'$, hence is won by~$\exists$, and moreover
    induces~$w$ in the sense discussed after
    Definition~\ref{def:mc-games}. Specifically, the positions of the
    form $(x,\psi)$ visited by~$\pi'$ are precisely
    $(x_0,\psi_0),(x_1,\psi_1),\dots$ where~$x_i$ is the realizer
    of~$q_i$ according to the invariant; interceding moves to
    positions of the form $(D,\psi)$ with $D\subseteq C$ are determined
    by~$s'$. Since as noted after Definition~\ref{def:mc-games}, the
    winning objective of~$\exists$ in the model checking game is
    non-acceptance of the associated run of~$\mathsf{A}_\chi$, this
    implies that $\rho$ is non-accepting. 
  \end{proof}
\noindent The results so far are tied up
as follows:
\begin{thm}[Soundness and completeness]\label{thm:soundcomp}
  The following are equivalent:
  \begin{enumerate}
  \item\label{item:win} The existential player wins the position $\initstate$ in $\game_\target$.
  \item\label{item:tableau} There is a tableau for $\chi$.
  \item \label{item:sat} The formula $\target$ is satisfiable.
\end{enumerate}
\end{thm}
\begin{proof}
  The
  implication~\eqref{item:win}$\implies$\eqref{item:tableau} is
  \autoref{lem:games}. We have shown the implication
  \eqref{item:sat}$\implies$\eqref{item:win} in
  \autoref{lem:modeltogame}. We prove
  \eqref{item:tableau}$\implies$\eqref{item:sat}: Let $(W,\delta')$ be a
  tableau for $\target$.
By the existence 
lemma (\autoref{lem:existence}), there is a coherent coalgebra
built over  $(W,\delta')$, which by the truth lemma (\autoref{lem:truth}) is a model for $\target$.
\end{proof}
\noindent
\noindent Our model construction in the proof of
\autoref{lem:existence} moreover yields the same bound on minimum
model size as in earlier work on the coalgebraic
$\mu$-calculus~\cite{CirsteaEA11a,FontaineEA10}:
\begin{cor}[Small-model property]\label{cor:smp}
  Let $\target$ be a satisfiable coalgebraic $\mu$-calculus formula,
  with parameters $n_0,k$ and $k'=\lfloor (k+1)/2\rfloor+1$ as in the
  running notation. Then~$\target$ has a model of size
  $\mathcal{O}(((nk')!)^2)\in 2^{\mathcal{O}(nk\log n)}$.
\end{cor}

\begin{rem}[Polynomially branching models]\label{rem:poly-branching}
  In addition to having an exponentially bounded number of states
  (\autoref{cor:smp}), the models $(V,\xi)$ constructed in the
  above completeness proof are also polynomially branching, provided
  that the logic has the one-step polysize model property, which holds
  in all our running examples (\autoref{rem:ospmp}). By this we
  mean that there is a polynomial~$p$ such that for every $q\in V$,
  there is a subset $V_0\subseteq V$ such that $|V_0|\le p(n)$ and
  $\xi(q)$ has the form $\xi(q)=Fi(t)$ where $i\colon V_0\to V$ is the
  subset inclusion. This property is immediate from the construction
  of coherent coalgebras in the proof of the existence lemma
  (\autoref{lem:existence}), in which $\xi(q)$ is obtained from a
  model of a one-step pair over~$\FLtarget$. With the exception of the
  standard $\mu$-calculus, this bound appears to be new in all our
  example logics. Of course, for graded and Presburger $\mu$-calculi,
  polynomial branching holds only in their coalgebraic semantics,
  i.e.\ over multigraph models but not over Kripke models.
\end{rem}

\section{Lazy Game Solving for the
Coalgebraic \texorpdfstring{$\mu$}{mu}-Calculus}\label{section:alg}

We proceed to show that the satisfiability game introduced in the
previous section can be solved in singly exponential time (under mild
assumptions on the complexity of the underlying one-step
satisfiability problem).  To this end, we introduce a
satisfiability checking algorithm that solves the game
\emph{on-the-fly} (that is, in a lazy fashion), and analyse the
runtime of the algorithm. As mentioned above, the obstacle to be
overcome here is that the game is doubly exponentially large,
specifically has singly exponentially many positions owned by the
existential player but doubly exponentially many owned by the
universal player. We deal with this issue by a characterization of the
existential player's winning region as a nested fixpoint that lives on
an exponential-sized subset of the game positions (those corresponding
directly to states in the co-determinized tracking
automaton~$\mathsf{B}_\target$), with interceding moves absorbed into
the definition of the function whose fixpoint is computed
(\autoref{lem:sat-game-fp}). The satisfiability checking algorithm
may then be understood as computing this fixpoint, respectively
determining whether the root position of the game belongs to the
fixpoint.

We recall that $\initstate\in\detcarrier$ is the initial node of the
co-determinized tracking automaton~$\mathsf{B}_\target$.  The
algorithm expands $\mathsf{B}_\target$ step by step starting from
$\initstate$; the expansion step adds nodes according to all possible
choice functions and all selections of modalities in an unexpanded
node $q$. The order of expansion can be chosen freely, e.g. by
heuristic methods. Optional intermediate game solving steps can be
used judiciously to realize on-the-fly solving.

\begin{alg}[Satisfiability checking]\label{alg:global}
  To decide satisfiability of the input formula~$\target$, initialize
  the sets of \emph{unexpanded} and \emph{expanded} nodes, $U=\{\initstate\}$
  and $Q=\emptyset$, respectively.
\begin{enumerate}
\item Expansion: Choose some unexpanded node $q\in U$, remove~$q$ from
  $U$, and add~$q$ to~$Q$.  Add all nodes in the sets
  $\{\delta(q,\tau)\in\detcarrier\mid\tau\in\mathsf{choices}\}\setminus Q$ and
  $\{\delta(q,\kappa)\in\detcarrier\mid \kappa\in\mathsf{selections},\kappa\subseteq
  l^A(q)\}\setminus Q$ to $U$.
\item Optional solving: Compute $\mathsf{win}^\exists_Q$ and/or $\mathsf{win}^\forall_Q$. If
$\initstate\in\mathsf{win}^\exists_Q$, then return `satisfiable`, if $\initstate\in\mathsf{win}^\forall_Q$,
then return `unsatisfiable`.
\item If $U\neq\emptyset$, then continue with Step~1.
\item Final game solving: Compute $\mathsf{win}^\exists$. If
$\initstate\in\mathsf{win}^\exists$, then return `satisfiable`, otherwise return
`unsatisfiable`.
\end{enumerate}
\end{alg}

\noindent Before analysing the run time behaviour of the algorithm, we
first show how to compute the sets $\mathsf{win}^\exists_Q$ and
$\mathsf{win}^\forall_Q$ in singly exponential time. Put
\begin{equation*}
  N=\prios
\end{equation*}
with~$k'$ as per~\eqref{eq:k-prime}, the number of priorities
in~$\mathsf{B}_\target$ (cf.\ \autoref{sec:tracking}).  We define
$N$-ary set functions $f_Q$ and $g_Q$ that compute one-step
(tn)satisfiability w.r.t.\ their argument sets.  These functions
essentially encode short sequences of moves in~$\game_\target$ leading
from one node in~$\detcarrier$ to the next.

\begin{defi}[Small-step game solving functions]\label{def:os-propagation} For sets $Q\subseteq \detcarrier$
  and $X_1,\ldots, X_{N} \subseteq Q$, we put 
\begin{align*}
  f_Q(X_1,\ldots, X_{N})=&\{q\in Q\mid 
(\gamma_q,\Theta_q^{\Xi(X_{\Omega(q)})} ) \text{ is satisfiable}, \bot\notin l^A(q) \text{ and }\\
  &\qquad \qquad \exists \tau\in\mathsf{choices}.\,
                  \delta(q,\tau)\in X_{\detprio(q)}\}\\
  g_Q(X_1,\ldots, X_{N})=&\{q\in Q\mid 
(\gamma_q,\Theta_q^{\Xi(\overline{X_{\Omega(q)}})})  \text{ is not satisfiable}, \bot\in l^A(q) \text{ or}\\
  &\qquad \qquad \forall \tau\in\mathsf{choices}.\,
                  \delta(q,\tau)\in X_{\detprio(q)}\},
\end{align*}
where
$\Xi(X)=\{\kappa\in\mathsf{selections}\mid \kappa\subseteq l^A(q)\text{ and }\delta(q,\kappa)\in X\}$
and $\overline{X}=\detcarrier\setminus X$ for $X\subseteq\detcarrier$.
\end{defi}
\noindent
Note how $f_Q$ propagates winning positions for~$\exists$
in~$\game_\target$, checking whether~$\exists$ has a response to both
immediate next $\forall$-moves from~$q\in\detcarrier$ (to $(q,0)$ or
$(q,1)$), while~$g_Q$ propagates winning positions for~$\forall$,
checking that $\forall$ wins by moving to either $(q,0)$ or $(q,1)$.

The time required for small-step game solving steps thus depends on
the time complexity of the one-step satisfiability problem. In
\autoref{lem:complexity}, we correspondingly give an estimate of the
overall time complexity of the satisfiability checking algorithm under
the assumption that the strict one-step satisfiability problem is in
\ExpTime.

Next we characterize the winning regions $\mathsf{win}^\exists_Q$ and 
$\mathsf{win}^\forall_Q$ by fixpoint expressions over $\Pow(\detcarrier)$, using
the small-step game solving functions $f_Q$ and $g_Q$, respectively.

\begin{defi}[Fixpoint descriptions of winning regions]\label{def:propagation} Given a set
  $Q\subseteq\detcarrier$, we put
\begin{align*}
\mathbf{E}_Q&=\eta_{N} X_{N}.\,\ldots \eta_1 X_1. f_Q(\mathbf{X}) &
\mathbf{A}_Q&=\overline{\eta_{N}} X_{N}\,\ldots \overline{\eta_1} X_1. g_Q(\mathbf{X}),
\end{align*}
where $\mathbf{X}=X_1,\ldots,X_{N}$ is a vector of variables~$X_i$
ranging over subsets of~$Q$, where $\eta_i= \mu$ for odd $i$,
$\eta_i=\nu$ for even~$i$, and where $\overline{\nu}=\mu$ and
$\overline{\mu}=\nu$.
\end{defi}

\noindent We will show that this fixpoint characterization is indeed
correct, that is, that
\begin{equation*}
  \mathbf{E}_Q=\mathsf{win}^\exists_Q\quad\text{and}\quad
  \mathbf{A}_Q=\mathsf{win}^\forall_Q
\end{equation*}
for $Q\subseteq\detcarrier$.  As
the sets $\mathbf{E}_Q$ and $\mathbf{A}_Q$ grow monotonically
with~$Q$, and since clearly $\mathbf{A}_{\detcarrier}$ is the
complement of $\mathbf{E}_{\detcarrier}$, it suffices to prove that
the winning region $\mathsf{win}^\exists$ in $\game_\target$ coincides
with the set $\mathbf{E}:=\mathbf{E}_{\detcarrier}$.

\begin{lem}\label{lem:sat-game-fp}
For all $q\in\detcarrier$, we have
$q\in\mathbf{E}$ if and only if the existential player wins the position $q$ in
the satisfiability game $\game_\target$.
\end{lem}
\begin{proof}
 The fixpoint
  \begin{equation*}
    \mathbf{E}=\eta_{N} X_{N}.\,\ldots \eta_1 X_1. f_{\detcarrier}(X_1,\ldots,X_{N})
  \end{equation*}
  on~$\detcarrier$ may, as discussed in \autoref{rem:fp-games}, be
  seen as described by a formula $\eta_{N} X_{N}.\,\ldots \eta_1 X_1.$
  $\Diamond(X_1,\ldots,X_{N})$ in a generalized form of the monotone
  $\mu$-calculus, and thus is characterized by the corresponding
  instance of the subformula model checking game. The simple structure
  of the fixpoint allows for further simplification of the
  game. First, all positions of the form $(q,\psi)$ where~$\psi$ is
  not a fixpoint literal or a fixpoint variable (so the next move
  might not be uniquely determined) have
  $\psi=\Diamond(X_1,\dots,X_{N})$; in particular, all such positions
  belong to~$\exists$. Second, positions $(q',X_l)$ reached from such
  a position after $\exists$'s move and a subsequent $\forall$-move
  automatically proceed to $(q',\Diamond(X_1,\dots,X_{N}))$, with~$l$
  being the maximal priority visited on the way. We thus eliminate the
  intermediate positions, and rename positions
  $(q,\Diamond(X_1,\dots,X_{N}))$ into just~$q$; similarly, we omit
  formula annotations on subsets of~$\detcarrier$ played by~$\exists$
  in modal moves. We write $\game_\mathbf{E}$ for the simplified form
  of the game, which is summarized by the following table:
  \begin{center}
    \begin{tabular}{|c|c|c|c|}
      \hline
      position & owner & set of allowed moves & priority\\
      \hline
      $q$ & $\exists$ & $\{(A_1,\dots,A_{N})\in\Pow(\detcarrier)^N\mid q\in f_{\detcarrier}(A_1,\dots,A_{N})\}$ & $0$\\
      $(A_1,\dots,A_{N})$ & $\forall$ & $\{(q,X_l)\mid l\in\{1,\dots,N\},q\in A_l\}$ & $0$\\
      $(q,X_l)$ & $\exists$ & $\{q\}$ & $l$\\ 
      \hline
    \end{tabular}
  \end{center} 
  
  \noindent By correctness of the model checking game
  (\autoref{thm:satisfaction-game}), the claim is thus reduced to
  showing that~$\exists$ wins~$q$ in~$\game_{\mathbf{E}}$
  iff~$\exists$ wins~$q$ in $\game_\target$. We say that~$\exists$ can
  \emph{force} a set $U\subseteq\{0,\dots,N\}\times\detcarrier$ in
  position~$q$ in one of the games if~$\exists$ has a strategy
  ensuring that~$\exists$ does not lose by getting stuck and that the
  pair $(j,q')$ consisting of the next position~$q'\in\detcarrier$
  reached in the play (if any;~$\forall$ might still get stuck) and
  the maximal priority~$j$ encountered on the way to~$q'$, including
  the priority of~$q$ but excluding the priority of~$q'$, lies
  in~$U$. Since every infinite play in~$\game_{\mathbf{E}}$
  or~$\game_{\target}$ infinitely often visits positions
  in~$\detcarrier$, it suffices to show that at
  every~$q\in\detcarrier$,~$\exists$ can force the same sets~$U$ in
  either of the games.

  For one direction, suppose that~$\exists$ can force
  $U\subseteq\{0,\dots,N\}\times\detcarrier$ at~$q$ in
  $\game_\mathbf{E}$ by moving to $(A_1,\dots,A_N)$; in particular,
  $q\in f_\detcarrier(A_1,\dots,A_N)$. In~$\game_\target$,~$\exists$
  then enforces~$U$ at~$q$ as follows.
  \begin{itemize}
  \item First, suppose that~$\forall$ moves from~$q$ to $(q,0)$. Since
    $q\in f_\detcarrier(A_1,\dots,A_N)$, we have $\bot\notin l^A(q)$,
    and there is $\tau\in\mathsf{choices}$ such that
    $\delta(q,\tau)\in A_{\detprio(q)}$. Thus,~$\exists$ can move from
    $(q,0)$ to $\delta(q,\tau)$ in~$\game_\target$, the highest
    priority encountered on the way from~$q$ to $\delta(q,\tau)$ being
    $\Omega(q)$. The pair $(\Omega(q),\delta(q,\tau))$ is in~$U$ as
    required, since in~$\game_{\mathbf{E}}$,~$\forall$ can move from
    $(A_1,\dots,A_N)$ to~$(\delta(q,\tau),X_{\Omega(q)})$, which has
    priority~$\Omega(q)$, and the game then automatically proceeds
    to~$\delta(q,\tau)$.
  \item Second, suppose that~$\forall$ moves to $(q,1)$. Since
    $q\in f_\detcarrier(A_1,\dots,A_N)$, we have
    that the one-step pair $(\gamma_q,\Theta_q^{\Xi(A_{\Omega(q)})})$    
    is satisfiable, recalling that $\Xi(A_{\Omega(q)})=\{\kappa\in\mathsf{selections}\mid \kappa\subseteq l^A(q) \text{ and }
\delta(q,\kappa)\in A_{\Omega(q)}\}$; so~$\exists$ can move to $(q,\Xi(A_{\Omega(q)}))$. After the next
    move by~$\forall$, we thus end up in $\delta(q,\kappa)$ for some
    $\kappa\in\Xi(A_{\Omega(q)})$, with the highest priority encountered on the way
    being~$\Omega(q)$. Since $\kappa\in\Xi(A_{\Omega(q)})$, we have
    $\delta(q,\kappa)\in A_{\detprio(q)}$, so by the same analysis as
    in the previous case, the pair $(\Omega(q),\delta(q,\kappa))$ is
    in~$U$ as required.\medskip
  \end{itemize}

  For the converse direction, suppose that~$\exists$ can force
  $U\subseteq\{0,\dots,N\}\times\detcarrier$ at~$q$ in~$\game_\target$
  by moving to $\delta(q,\tau)$ in case~$\forall$ moves to~$(q,0)$ and
  to $(q,\Xi)$ in case~$\forall$ moves to~$(q,1)$, where
  $\tau\in\mathsf{choices}$ and $\Xi\in\Pow(\mathsf{selections}(q))$. In
  particular, we then have $\bot\notin l^A(q)$ and
  the one-step pair $(\gamma_q,\Theta_q^{\Xi})$ is satisfiable.
  We claim that in~$\game_{\mathbf{E}}$,~$\exists$ then forces~$U$ by
  moving to
  $(\emptyset,\ldots,\emptyset,A_{\Omega(q)},\emptyset,
  \ldots,\emptyset)$ (with $A_{\Omega(q)}$ in position $\Omega(q)$)
  where
  \begin{equation*}\label{eq:y-omega}
    A_{\Omega(q)}=\{\delta(q,\tau)\}\cup
    \delta(q,\Xi).
  \end{equation*}
  We note that this is a legal move, as
  $q\in
  f_{\detcarrier}(\emptyset,\ldots,\emptyset,A_{\Omega(q)},\emptyset,
  \ldots,\emptyset)$ by construction of~$A_{\Omega(q)}$ (and
  \autoref{rem:os-logic}). In~$\game_{\mathbf{E}}$,~$\forall$ then
  necessarily moves to a position $(q',X_{\detprio(q)})$ where
  $q'\in A_{\detprio(q)}$, and the game then automatically proceeds
  to~$q'$, with the maximal priority encountered on the way
  being~$\detprio(q)$. We distinguish cases on~$q'$:
  \begin{itemize}
  \item If $q'=\delta(q,\tau)$, then $(\detprio(q),q')\in U$ as
    required, since~$\delta(q,\tau)$ is~$\exists$'s reply to $(q,0)$
    forcing~$U$ in $\game_\target$.
  \item Otherwise, $q'=\delta(q,\kappa)$ for
    some~$\kappa\in\Xi$. Since~$(q,\Xi)$ is~$\exists$'s reply to
    $(q,0)$ forcing~$U$ in $\game_\target$, and~$\forall$ can move
    from $(q,\Xi)$ to $\delta(q,\kappa)$ in~$\game_\target$, we again
    have $(\detprio(q),q')\in U$ as required. \qedhere
  \end{itemize}
\end{proof}

Having shown how the sets $\mathsf{win}^\exists_Q$ and 
$\mathsf{win}^\forall_Q$ can be computed
by evaluating fixpoint expressions over $\Pow(\detcarrier)$, we 
next analyse the run time behaviour of the introduced algorithm.

\begin{lem}[Time analysis]\label{lem:complexity}
  If the strict one-step satisfiability problem is decidable in time
  $t(n)$, then the above satisfiability checking algorithm runs in
  time $\mathcal{O}(((\prios)!)^{2c}\cdot t(n))$ for some constant~$c$
  if no optional game solving steps are performed, with~$n,k'$ being
  the parameters of the target formula as per
  \autoref{sec:tracking}.
\end{lem}
\noindent (The run time with optional game solving steps is still
singly exponential; in view of the fact that exponential run time of
some fixed strategy on intermediate game solving suffices to obtain the \ExpTime bound
on satisfiability checking, we restrict to the case without optional
game solving for the sake of simplicity.)

\begin{proof}
  The loop of the algorithm expands the co-determinized tracking
  automaton node by node and hence is executed at most
  $|\detcarrier|\in\mathcal{O}(\detsize)$ times.  A single expansion
  step can be implemented in time $\mathcal{O}(2^{n_0})$ since in both
  propositional and modal expansion steps, at most $2^{n_0}$ new nodes
  are added, corresponding to the maximal possible number of choice
  functions and matching selections, respectively.  By
  \autoref{lem:sat-game-fp}, the final solving step can be performed
  by computing a fixpoint of nesting depth $N$ of the
  function~$f$ 
  over $\Pow(\detcarrier)^{N}$.  A single computation of
  $f(\mathbf{X})$ 
  for a tuple $\mathbf{X}\in \Pow(\detcarrier)^{N}$ can be implemented
  in time
  $\lO(|\detcarrier|\cdot(t(n)+2^{n_0}))=\lO(\detsize\cdot(t(n_0)+2^{n_0}))$
  by going through all elements~$q$ of~$\detcarrier$, calling the
  one-step satisfiability checker on $l^A(q)\cap\Lambda(\FLtarget)$,
  and verifying the existence of a suitable choice letter.  Since we
  have $N =\lO( \log|\detcarrier|)$, it follows from recent work on
  the computation of nested fixpoints~\cite{HausmannSchroder21} that
  these fixpoints can be computed in time
  $\mathcal{O}((n_0k')!^{2c}\cdot t(n_0))$, where $c=5$. (Classical
  methods for computing nested fixpoints~\cite{LongEA94,Seidl96} are
  exponential in the nesting depth, which however still leads to a
  singly exponential overall time bound, computed explicitly in the
  conference version of the paper~\cite{HausmannSchroder19}.) Thus the
  complexity of the whole algorithm is dominated by the complexity of
  the final game solving step, and adheres to the claimed asymptotic
  bound.
\end{proof}

\noindent Relying on 
 the correctness of satisfiability games as shown in 
 \autoref{sec:satgames} above,
we obtain the following results.
\begin{thm}[Exponential-time upper bound]\label{thm:exptime}
  If the strict one-step satisfiability problem of a coalgebraic logic
  is in {\ExpTime}, then the satisfiability problem of the corresponding
  coalgebraic $\mu$-calculus is in \ExpTime.
\end{thm}

\noindent Since as discussed in \autoref{rem:rules} below, the existence
of a tractable set of tableau rules implies the required time bound on
one-step satisfiability, the above result subsumes earlier bounds
obtained by tableau-based approaches
in~\cite{CirsteaEA11a,HausmannEA16,HausmannEA18}; however, it covers
additional example logics for which no suitable tableau rules are
known. In particular, by \autoref{ex:onestepsat}, we have:
\begin{prop}\label{prop:expls}
The satisfiability problems of the following logics are in \ExpTime:
\begin{enumerate}
\item the standard $\mu$-calculus,
\item the monotone $\mu$-calculus (including its fragment game logic), 
\item the graded $\mu$-calculus, 
\item the (two-valued) probabilistic $\mu$-calculus,
\item the graded $\mu$-calculus with polynomial inequalities,
\item the (two-valued) probabilistic $\mu$-calculus with polynomial inequalities.
\end{enumerate}
\end{prop}
\begin{rem}[Modal tableau rules]\label{rem:rules}
  As indicated in the introduction, previous generic algorithms for
  the coalgebraic $\mu$-calculus~\cite{CirsteaEA11} employ tractable
  sets of modal tableau rules~\cite{SchroderPattinson09b} in place of
  one-step satisfiability checking. This method roughly works as
  follows. A \emph{(monotone) modal tableau rule} over a finite
  set~$W$ of propositional variables (local to the rule) has the form
  $\phi/\psi$ where $\psi$ is a propositional formula over~$W$, given
  in disjunctive normal form as a subset of~$\Pow(W)$, and $\phi$ is a
  finite subset of $\Lambda(W)$, representing a finite conjunction,
  that mentions every variable in~$W$ exactly once. Given a set~$V$, a
  \emph{rule match} to a finite subset~$\osOne$ of $\Lambda(V)$ is a
  pair $(\phi/\psi,\iota)$ consisting of a rule $\phi/\psi$ and a
  substitution~$\iota\colon W\to V$ that acts injectively on~$\phi$
  (i.e.\ for $\hearts a,\hearts b\in\phi$, $\iota(a)=\iota(b)$ implies
  $a=b$) such that $\phi\iota\subseteq \osOne$, where we write
  application of the substitution~$\iota$ in postfix notation. In the
  notation of the present paper, a set~$\Rules$ of modal tableau rules
  is \emph{one-step tableau sound and complete} if the following
  condition holds for each one-step pair $(\osOne,\osZero)$ over~$V$:
  The pair $(\osOne,\osZero)$ is satisfiable iff
  $\psi\iota\cap \osZero\neq\emptyset$ for each rule
  match~$(\phi/\psi,\iota)$ as above to~$\osOne$; note that applying
  the substitution~$\iota\colon W\to V$ to~$\psi\subseteq\Pow(W)$
  yields a propositional formula~$\psi\iota$ over~$V$ that is
  represented as a subset of~$\Pow(V)$. 

  A rule set~$\Rules$ is \emph{exponentially tractable} if rule
  matches can be encoded as strings in such a way that every rule
  match to a given finite set $\osOne\subseteq\Lambda(V)$ has a code
  of polynomial size in~$\size(\osOne)$, and moreover~(i) it can be
  decided in exponential time whether a given code actually encodes a
  rule match to~$\osOne$ and~(ii) the conclusion $\psi\iota$ of a
  rule match $(\phi/\psi,\iota)$ can be computed from its code in
  exponential time. Using an exponentially tractable rule set, we can
  decide the strict one-step satisfiability problem in exponential
  time: Given a one-step pair $(\osOne,\osZero)$, go through all codes
  of possible rule matches, filtering for actual matches
  $(\phi/\psi,\iota)$, and check for each such match that
  $\psi\iota\cap \osZero\neq\emptyset$. Thus, the present approach
  applies more generally than the approach via modal tableau
  rules. See also~\cite{KupkeEA22} for a more detailed discussion of
  the relationship.

  \label{rem:graded-rules}
  Tractable sets of tableau rules for the graded $\mu$-calculus and
  the Presburger $\mu$-calculus have been claimed in previous
  work~\cite{SchroderPattinson09b,KupkePattinson10}. However, these
  rule sets have since turned out to be incomplete. Indeed the rule
  sets are very similar to rule sets for real-valued systems, and
  remain sound over an evident real-valued relaxation of the
  semantics, an observation from which concrete examples showing
  incompleteness are obtained rather immediately both for the
  Presburger~\cite[Remark 3.8]{KupkeEA22} and for the graded
  case~\cite[Appendix of extended version]{GorlitzEA23}.
\end{rem}

\begin{rem}[$\Lambda$-automata]\label{rem:flv}
  As indicated in \autoref{sec:intro}, the satisfiability game
  considered by Fontaine et al.~\cite{FontaineEA10} actually checks
  emptiness of so-called $\Lambda$-automata, into which formulae of
  the coalgebraic $\mu$-calculus can be translated with only
  polynomial blow-up. Non-emptiness of a
  $\Lambda$-automaton~$\mathbb{A}$ with set~$A$ of states is checked
  using a game $\Sat(\mathbb{A})$ that has pairs of automata states
  from~$A$ as $\exists$-positions, and sets of binary relations on~$A$
  as $\forall$-positions. The winning condition of $\Sat(\mathbb{A})$
  is regular but not parity, so winning strategies in
  $\Sat(\mathbb{A})$ need to depend on memory states $m\in M$ from an
  exponential-sized set~$M$. The model construction then has states
  being pairs $(v,m)$ consisting of an $\exists$-position~$v$ and a
  memory state $m\in M$ (while our model construction uses states from
  the co-determinized tracking automaton). As we indicate in the
  introduction, it does not seem likely that our approach to solving a
  doubly-exponential-sized satisfiability game in singly exponential
  time will in general transfer to $\Sat(\mathbb{A})$ (for the case
  where a tractable set of tableau rules is known, $\Sat(\mathbb{A})$
  has been reformulated to be solvable in exponential
  time~\cite[Section~5]{FontaineEA10}), as none of the two types of
  positions in $\Sat(\mathbb{A})$ has the right size (there are doubly
  exponentially many $\forall$-positions and polynomially many
  $\exists$-positions). Also, our fixpoint description relies on the
  fact that our game has a parity winning condition, and it is not
  clear how it would transfer to a regular game.
\end{rem}

\begin{rem}[One-step satisfiability and fusion]\label{rem:fusion-oss}
  The criterion of \autoref{thm:exptime} is stable under fusion
  of logics; that is: Suppose that the strict one-step satisfiability
  problems for logics with disjoint modal similarity types~$\Lambda_i$
  interpreted over functors~$F_i$, for $i=1,2$, are both in
  \ExpTime. Then the strict one step satisfiability problem of the
  fusion (\autoref{rem:fusion}), with modal similarity type
  $\Lambda=\Lambda_1\cup\Lambda_2$ interpreted over $F=F_1\times F_2$,
  is in \ExpTime as well. To see this, just note that a one-step pair
  $(\osOne,\osZero)$ over~$V$ in the fusion is satisfiable over~$F=F_1\times F_2$
  iff for $i=1,2$. the one-step pair $(\osOne\cap\Lambda_i(V),\osZero)$ over~$V$
  is satisfiable over~$F_i$.

  Thus, we obtain by \autoref{thm:exptime} that the
  satisfiability problem of any combination of the logics mentioned in
  \autoref{prop:expls} is in \ExpTime; for instance, this
  holds for the logic of Markov decision processes described in
  \autoref{ex:logics}.\ref{item:comp}.
\end{rem}

\section{Conclusion}\label{section:conclusion}

We have shown that the satisfiability problem of the coalgebraic
$\mu$-calculus is in \ExpTime, subject to establishing a suitable time
bound on the much simpler one-step satisfiability problem. Our method
does not require guardedness of fixpoint variables. Prominent examples
include the graded $\mu$-calculus, the monotone $\mu$-calculus and its
fragment game logic, the (two-valued) probabilistic $\mu$-calculus,
and extensions of the probabilistic and the graded $\mu$-calculus,
respectively, with (monotone) polynomial inequalities; the \ExpTime
bound appears to be new for the last two logics.  We have also
presented a generic satisfiability algorithm that realizes the time
bound and supports on-the-fly solving in the spirit of global caching
algorithms. Moreover, we have obtained a polynomial bound on minimum
branching width in models for all example logics mentioned above.

\bibliographystyle{alphaurl}
\bibliography{coalgml}
\end{document}